\documentclass[english,11pt, reqno]{amsart}
\usepackage[T1]{fontenc}
\usepackage{amscd}
\usepackage{amsthm,amsmath,amsfonts,amssymb}
\usepackage{verbatim}
\usepackage{lmodern}
\usepackage{graphicx}
\usepackage{enumerate}
\usepackage[percent]{overpic}
\usepackage{amsaddr}
\usepackage[numbers,sort&compress]{natbib}

\newcommand{\re}{\text{\upshape Re\,}}
\newcommand{\im}{\text{\upshape Im\,}}

\usepackage[colorlinks=true]{hyperref}
\hypersetup{urlcolor=blue, citecolor=red, linkcolor=blue}

\makeatletter
\numberwithin{equation}{section}
\theoremstyle{plain}
\newtheorem{thm}{\protect\theoremname}
  \theoremstyle{plain}
  \newtheorem{lemma}[thm]{\protect\lemmaname}
   \newtheorem{prop}[thm]{\protect\propname}

\newtheorem{figuretext}{Figure}

\newtheorem{remark}[thm]{Remark}
\numberwithin{thm}{section}

\theoremstyle{remark}

\makeatother

\usepackage{babel}
\providecommand{\propname}{Proposition}
\providecommand{\lemmaname}{Lemma}
\providecommand{\theoremname}{Theorem}

\newcommand{\imag}{\operatorname{Im} }
\newcommand{\real}{\operatorname{Re} }
\renewcommand{\Im}{\imag}
\renewcommand{\Re}{\real}

\newcommand{\dist}{\operatorname{dist}}

\newcommand{\Z}{{\mathbb Z}}

\newcommand{\R}{{\mathbb{R}}}
\newcommand{\C}{{\mathbb C}}

\newcommand{\GG}{{\mathcal G}} 

\let \le \leqslant
\let \leq \leqslant

\let \geq \geqslant

\input epsf
\title[Asymptotic analysis of Dotsenko-Fateev integrals]
{Asymptotic analysis of Dotsenko-Fateev integrals}

\author{Jonatan Lenells and Fredrik Viklund }
\address{Department of Mathematics, KTH Royal Institute of Technology, \\ 100 44 Stockholm, Sweden.}
\email{jlenells@kth.se, fredrik.viklund@math.kth.se}

\begin{document}
\begin{abstract} 
We develop a method for evaluating asymptotics of certain contour integrals that appear in Conformal Field Theory under the name of Dotsenko-Fateev integrals and which are natural generalizations of the classical hypergeometric functions. We illustrate the method by establishing a number of estimates that are useful in the context of martingale observables for multiple Schramm-Loewner evolution processes.
\end{abstract}

\maketitle



\tableofcontents

\section{Introduction}
In this paper, we will be interested in integrals of the form
\begin{align}\label{pochhammerexpression}
\int_A^{(z_1+,z_2+,z_1-,z_2-)} \prod_{j = 1}^N (u - z_j)^{a_j} du,
\end{align}
where $N \geq 2$ is an integer, $\{a_j\}_{j=1}^N \subset \R$ is a set of exponents, $\{z_j\}_{j=1}^N \subset \C$ is a finite collection of points, and the Pochhammer integration contour encloses two of these points, say $z_1$ and $z_2$. More precisely, the Pochhammer contour in (\ref{pochhammerexpression}) begins at the base point $A \in \C$, encircles $z_1$ once in the positive (counterclockwise) direction, returns to $A$, encircles $z_2$ once in the positive direction, returns to $A$, encircles $z_1$ once in the negative direction, returns to $A$, and finally encircles $z_2$ once in the negative direction before returning to $A$, see Figure \ref{Pochhammerintro.pdf}. The remaining points $\{z_j\}_{j=3}^N$ are assumed to lie exterior to the contour and the powers in the integrand are assumed to take their principal values at the starting point and are then analytically continued along the contour.

By employing what is sometimes called the screening method, Dotsenko and Fateev obtained representations of correlation functions in Conformal Field Theory (CFT) in terms of such integrals \cite{DF1984, DF1985} (see also \cite{DMS1997}). Integrals of the form (\ref{pochhammerexpression}) are therefore often referred to as Dotsenko-Fateev integrals. Various aspects of such integrals and related special functions have been studied in a number of places, e.g., \cite{FL2004, GI2016}. For $N =2$ and $N = 3$, the integral (\ref{pochhammerexpression}) can be expressed in terms of Gamma functions and the classical hypergeometric function ${}_2F_1$, respectively. In the case of $N = 4$, a comparison with the representation of Appell's $F_{1}$ function as an Euler integral (cf. \cite{FL2004}) shows that (\ref{pochhammerexpression}) can be viewed as an analytic continuation of Appell's function. We also mention \cite{D2006,  KM2013, KP2016a} which discuss related questions in the context of the Schramm-Loewner Evolution (SLE). 
%

The basic problem we are interested in here is how to compute the asymptotics of the integral (\ref{pochhammerexpression}) as one (or several) of the points $z_k$, $k = 3, \dots, N$, approaches $z_1$ or $z_2$. The computation of such asymptotics presents difficulties, because if the point $z_k$, $k = 3, \dots, N$, approaches $z_1$, say, then the contour gets squeezed between $z_k$ and $z_1$ which, in general, gives rise to a singular behavior. For $N =3$, the relevant asymptotics can be obtained from well-known identities and expansions of hypergeometric functions. However, for $N \geq 4$, the analysis is more intricate. 

In this paper, we propose a method which makes it possible to compute the asymptotics of (\ref{pochhammerexpression}) to all orders for any $N \geq 3$ as one (or many) of the points $z_k$, $k \geq 3$, approaches $z_1$ or $z_2$. The obtained expansions are power series in the relevant small quantities with coefficients explicitly given in terms of Gamma functions and integrals of the form (\ref{pochhammerexpression}) of a lower order $\leq N-1$. 
In the case of $N = 4$, it is conceivable that the asymptotic formulas we derive can be obtained also from known properties of Appell's function. However, as far as we are aware, formulas of this type have not appeared in the literature even in this simplest case.

By applying a linear fractional transformation, we may assume that $z_1 = 0$ and $z_2 = 1$ in (\ref{pochhammerexpression}); this yields an expression of the form
\begin{align}\label{pochhammerexpression2}
\int_A^{(0+,1+,0-,1-)} v^{a_1} (1-v)^{a_2} \prod_{j = 3}^{N} (v - w_j)^{a_j} dv.
\end{align}
For simplicity, we will present the results in this paper for the class of integrals corresponding to $N = 4$.
Although the proposed method works for arbitrary values of $N$ and an arbitrary number $M \leq N$ of merging points, its application involves a large number of steps when $M$ is large (the number of steps typically grows like $M$, because, roughly speaking, each time two points merge one of the series contributing to the asymptotics splits into two). 

\begin{figure}
\medskip
\begin{center}
\begin{overpic}[width=.55\textwidth]{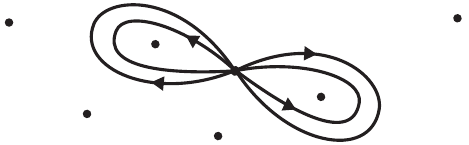}
      \put(47,11){\small $A$}
      \put(38,19){\small $1$}
      \put(62,10){\small $2$}
      \put(32,9){\small $3$}
      \put(65,21.5){\small $4$}
      \put(28,21){\small $z_1$}
      \put(70,9.5){\small $z_2$}
      \put(-3.5,25.4){\small $z_3$}
      \put(13,6){\small $z_4$}
      \put(41,1){\small $z_5$}
      \put(92,26.4){\small $z_6$}
\end{overpic}
\medskip
     \begin{figuretext}\label{Pochhammerintro.pdf}
       The Pochhammer integration contour in (\ref{pochhammerexpression}) consists of four loops based at $A$ encircling $z_1$ and $z_2$. The loop denoted by $1$ is traversed first, then the loop denoted by $2$ is traversed, and so on.
       \end{figuretext}
     \end{center}
\end{figure}

\subsection{Brief description of method}
Let $F(w_1, w_2) \equiv F(a,b,c,d; w_1, w_2)$ denote the integral in (\ref{pochhammerexpression2}) for $N = 4$, that is, $F(w_1, w_2)$ is defined for $w_1,w_2 \in \C \setminus [0, \infty)$ by
\begin{align}\label{Fdef}
F(w_1, w_2) = \int_A^{(0+,1+,0-,1-)} v^a (v-w_1)^b (v-w_2)^{c} (1-v)^{d} dv, 
\end{align}
where $A \in (0,1)$ is a base point, $a,b,c,d\in \R$ are real exponents, and $w_1, w_2$ are assumed to lie outside the contour. In order to make $F$ single-valued, we have restricted the domain of definition in (\ref{Fdef}) to $w_1,w_2 \in \C \setminus [0, \infty)$. 
We will assume that $a,d\in \R \setminus \Z$ are not integers, because otherwise the integral in (\ref{Fdef}) vanishes identically or can be computed by a residue calculation. Note that $F$ can be analytically continued to a multiple-valued analytic function of $w_1,w_2 \in \C \setminus \{0, 1\}$.
Our goal is to compute the asymptotic behavior of $F$ to all orders as one or both of the points $w_1$ and $w_2$ approach $0$ or $1$. 

The basic idea is the following: If we want to consider the limit $w_1 \to 0$ say, then we rewrite $F$ as a sum of two terms (see equation (\ref{FQidentity})). One term which is defined by the same integral as $F$ except that $w_1$ is now assumed to lie inside the contour in the same component as $0$ (see equation (\ref{Q1def})), and a second term which is defined by a similar expression but with the Pochhammer contour enclosing $\{0, w_1\}$ instead of $\{0, 1\}$ (see (\ref{Q2def}) and (\ref{FminusQ1})). The asymptotics of both of these terms can easily be computed to all orders by simply replacing the factors in the integrands by their asymptotic expansions as $w_1 \to 0$ such as
\begin{align}\label{vw1expansion}
(v - w_1)^b \sim \sum_{k=0}^\infty \frac{\Gamma(b+1)}{\Gamma(b+1-k)} (-1)^k v^{b-k} \frac{w_1^k}{k!}, \qquad w_1 \to 0.
\end{align}

We emphasize that it is, in general, not possible to compute the asymptotics of $F$ as $w_1 \to 0$ by substituting the expansion (\ref{vw1expansion}) directly into (\ref{Fdef}). Indeed, such a procedure gives the correct contribution from the first term, but completely ignores the contribution from the second term.

\subsection{Two examples}
Our initial motivation for studying asymptotics of integrals of the form \eqref{pochhammerexpression} was that they appear naturally when constructing observables for multiple SLE curves via the screening method within CFT. Thus, in the second part of the paper, in order to illustrate our method, we consider two concrete examples of such observables where integrals of the form (\ref{pochhammerexpression}) are important. By applying the techniques developed in the first part of the paper, we are able to derive asymptotic estimates for the integrals in these examples. The estimates we establish are used in the derivation of the Green's function and Schramm's formula for multiple SLE presented in \cite{LV2018_A}. We expect similar asymptotic estimates to be needed also in the context of other SLE observables derived via the screening method.

\subsection{Organization of the paper}
In Section \ref{mainsec}, we present four different examples of asymptotic expansions to all orders which can be derived by our method. Before turning to the full description of the method and the proofs of the above expansions in Section \ref{methodsec}, we consider the hypergeometric case of $N = 3$ in Section \ref{hypergeometricsec} as motivation. 
The two examples from SLE theory are introduced in Section \ref{SLEsect}. 
In Section \ref{greensec}, we derive the estimates relevant for the first example corresponding to the Green's function. 
In Section \ref{schrammsec}, we derive the estimates relevant for the second example corresponding to Schramm's formula.

\section{Four asymptotic theorems}\label{mainsec}
The purpose of this paper is to propose a method which makes it possible to compute the asymptotics of (\ref{pochhammerexpression2}) as one (or several) of the points $w_k$, $k = 3, \dots, N$, approaches $0$ or $1$. There are clearly many different cases to consider depending on the value of $N$, on the number of points $w_k$ involved in the limiting process, and on whether each $w_k$ approaches $0$ or $1$. For the sake of presentation, we have chosen to discuss four cases in detail. The presented cases all have $N = 4$ and correspond to the following limits: (a) $w_2 \to 1$, (b) $w_1 \to 0$, (c) $w_1 \to 0$ and $w_2 \to 0$, and (d) $w_1 \to 0$ and $w_2 \to 1$. These four cases, which are treated one by one in Theorem \ref{mainth1}-\ref{mainth4}, illustrate the different situations that may arise and it will be clear from the analysis of these cases how to apply the method also in other cases.

\subsection{Notation}
We let $F(w_1, w_2) \equiv F(a,b,c,d; w_1, w_2)$, $G(a,c,d; w)$, and $H(a,d)$ denote the integral given in (\ref{pochhammerexpression2}) for $N = 4$, $N = 3$, and $N =2$, respectively. That is, $F(w_1, w_2)$ is defined by (\ref{Fdef}), $G(a,c,d; w)$ is defined by
\begin{align}\label{Gdef}
  G(a,c,d; w) = \int_A^{(0+,1+,0-,1-)} v^a (v - w)^{c} (1-v)^{d} dv, \qquad w \in \C \setminus [0, \infty),
\end{align}
where $w$ is assumed to lie outside the contour, and $H(a,d)$ is defined by
\begin{align}\label{Hdef}
  H(a,d) = \int_A^{(0+,1+,0-,1-)} v^a (1-v)^{d} dv.
\end{align}
The Pochhammer integration contour in (\ref{Gdef}) begins at the base point $A \in (0,1)$, encircles the point $0$ once in the positive (counterclockwise) direction, returns to $A$, encircles $1$ once in the positive direction, returns to $A$, and so on. The point $w$ lies exterior to all loops; the factors in the integrand take their principal values at the starting point and are then analytically continued along the contour. A similar interpretation of Pochhammer contours applies to the definition (\ref{Fdef}) of $F(w_1,w_2)$ and elsewhere.
Throughout the paper we adopt the convention that unless stated otherwise, the principal branch is used for all complex powers and logarithms, i.e., $\ln w = \ln |w| + i \arg w$ with $\arg w \in (-\pi, \pi]$. 
For $c \in \C$, we define $\rho_c(w)$ by
\begin{align}\label{rhodef}
\rho_c(w) = \begin{cases} e^{-i\pi c}, & \im w \geq 0, \\
e^{i\pi c}, & \im w < 0.
\end{cases}
\end{align}

\subsection{Asymptotic theorems}
We can now state the asymptotic results. The proofs are given in Section \ref{methodsec}.

\begin{thm}[Asymptotics as $w_2 \to 1$]\label{mainth1}
Let $a,b,c,d \in \R \setminus \Z$ be such that $c+d \notin\Z$. 
Then $F$ satisfies the following asymptotic expansion to all orders as $w_2 \to 1$ with $w_1,w_2 \in \C \setminus [0, \infty)$:
\begin{align*}
F(w_1, w_2) \sim &\; \sum_{k=0}^\infty \bigg\{A_k^{(1)}(w_1) (w_2-1)^k 
 + A_k^{(2)}(w_1) (w_2-1)^{c+d+1+k}\bigg\}\rho_c(w_2), 
\end{align*}
where the coefficients $A_k^{(j)}(w_1) \equiv A_k^{(j)}(a,b,c,d;w_1)$, $j = 1,2$, are given by
\begin{align*}
& A_k^{(1)}(w_1) = \frac{e^{2\pi i d} - 1}{e^{2\pi i(c+d)} -1} \frac{\Gamma(c+1)}{\Gamma(c+1-k)k!} 
G(a,b,c+d-k; w_1),
	\\
& A_k^{(2)}(w_1) = \frac{(e^{2\pi ia} -1)e^{\pi i d}}{1 - e^{2\pi i(c+d)}} 
\sum_{l=0}^k  \frac{\Gamma(a+1)\Gamma(b+1)(1 - w_1)^{b-l}}{\Gamma(a+1-k+l)\Gamma(b+1-l)(k-l)!\, l!} H(c,d+k). 
\end{align*} 
\end{thm}

\begin{thm}[Asymptotics as $w_1 \to 0$]\label{mainth2}
Let $a,b,c,d \in \R \setminus \Z$ be such that $a+b \notin\Z$. 
Then $F$ satisfies the following asymptotic expansion to all orders as $w_1 \to 0$ with $w_1,w_2 \in \C \setminus [0, \infty)$:
\begin{align*}
F(w_1, w_2) \sim &\; \sum_{k=0}^\infty \bigg\{B_k^{(1)}(w_2) w_1^k 
+ \rho_{a+b}(w_1)  B_k^{(2)}(w_2) w_1^{a+b+1 + k}\bigg\}
\end{align*}
where the coefficients $B_k^{(j)}(w_2) \equiv B_k^{(j)}(a,b,c,d;w_2)$, $j = 1,2$, are given by
\begin{align*}
& B_k^{(1)}(w_2) = \frac{e^{2\pi i a} -1}{e^{2\pi i(a+b)} -1}\frac{\Gamma(b+1)(-1)^k}{\Gamma(b+1-k)k!}G(a+b-k, c,d; w_2),
	\\
& B_k^{(2)}(w_2) = \frac{(e^{2\pi i d} -1)e^{\pi i a}}{e^{2\pi i(a+b)} - 1} \sum_{l=0}^k \frac{\Gamma(c+1)\Gamma(d+1)(-1)^l (-w_2)^{c-k+l}}{\Gamma(c+1-k+l)\Gamma(d+1-l)(k-l)! \, l!} H(a+k,b).
\end{align*} 
\end{thm}

\begin{thm}[Asymptotics as $w_1 \to 0$ and $w_2 \to 0$ with $|w_1/w_2| < 1 - \delta$]\label{mainth3}
Let $a,b,c,d \in \R \setminus \Z$ be such that $a+b, a+b+c \notin\Z$. 
Then $F$ satisfies the following asymptotic expansion to all orders as $w_1 \to 0$ and $w_2 \to 0$ with $w_1,w_2 \in \C \setminus [0, \infty)$ such that $|w_1/w_2| < 1 - \delta$ for some $\delta > 0$:
\begin{align*}
F(w_1, w_2) \sim &\; \sum_{k=0}^\infty\sum_{l=0}^\infty \frac{1}{k!l!}\bigg\{C_{kl}^{(1)} w_1^k w_2^l
+ C_{kl}^{(2)} \rho_{a+b+c}(w_2) w_2^{a+b+c+1+l} \bigg(\frac{w_1}{w_2}\bigg)^k
	\\
& + C_{kl}^{(3)} \rho_{a+b}(w_1)\rho_c(w_2) w_1^{a+b+1+l} w_2^{c} \bigg(\frac{w_1}{w_2}\bigg)^k\bigg\},
\end{align*}
where the coefficients $C_{kl}^{(j)} \equiv C_{kl}^{(j)}(a,b,c,d)$, $j = 1,2,3$, are given by
\begin{align*}
& C_{kl}^{(1)} = \frac{e^{2\pi i a} -1}{e^{2\pi i(a+b+c)} -1} \frac{\Gamma(b+1)\Gamma(c+1) (-1)^{k+l}}{\Gamma(b+1-k) \Gamma(c+1-l)} H(a+b+c-k-l,d),
	\\
& C_{kl}^{(2)} = \frac{(e^{2\pi i a} -1)(e^{2\pi i d} -1)e^{\pi i(a+b)}}{(e^{2\pi i(a+b)} - 1)(e^{2\pi i(a+b+c)}-1)}  \frac{\Gamma(b+1)\Gamma(d+1) (-1)^{k+l}}{\Gamma(b+1-k) \Gamma(d+1-l)} H(a+b-k+l,c),
	\\
& C_{kl}^{(3)} = \frac{(e^{2\pi i d} -1)e^{\pi i a}}{e^{2\pi i(a+b)} - 1}  \frac{\Gamma(c+1)\Gamma(d+1) (-1)^{k+l}}{\Gamma(c+1-k) \Gamma(d+1-l)} H(a+k+l,b).
\end{align*} 
\end{thm}

\begin{thm}[Asymptotics as $w_1 \to 0$ and $w_2 \to 1$]\label{mainth4}
Let $a,b,c,d \in \R \setminus \Z$ be such that $a+b,c+d \notin\Z$. 
Then $F$ satisfies the following asymptotic expansion to all orders as $w_1 \to 0$ and $w_2 \to 1$ with $w_1,w_2 \in \C \setminus [0, \infty)$:
\begin{align*}
F(w_1, w_2) \sim &\; \sum_{k=0}^\infty\sum_{l=0}^\infty \frac{ \rho_c(w_2)}{k!\, l!}\bigg\{ D_{kl}^{(1)} w_1^k (w_2- 1)^l 
+ D_{kl}^{(2)} w_1^k (w_2-1)^{c+d+1+ l}
	\\
& + \rho_{a+b}(w_1) D_{kl}^{(3)} w_1^{a+b+1+ k}(w_2-1)^l\bigg\}
\end{align*}
where the coefficients $D_{kl}^{(j)} \equiv D_{kl}^{(j)}(a,b,c,d)$, $j = 1,2,3$, are given by
\begin{align*}
& D_{kl}^{(1)} = \frac{(e^{2\pi i a} -1)(e^{2\pi i d} -1)}{(e^{2\pi i(a+b)} - 1)(e^{2\pi i(c+d)} - 1)} \frac{\Gamma(b+1)\Gamma(c+1)(-1)^k}{\Gamma(b+1-k)\Gamma(c+1-l)} H(a+b-k, c+d-l),
	\\
& D_{kl}^{(2)} = \frac{(e^{2\pi i a} -1) e^{\pi i d}}{1 - e^{2\pi i (c+d)}} \frac{\Gamma(b+1)\Gamma(a+b+1-k)(-1)^k}{\Gamma(b+1-k)\Gamma(a+b+1-k-l)}H(c, d+l),
	\\
& D_{kl}^{(3)} = \frac{(e^{2\pi i d} -1) e^{\pi i a}}{e^{2\pi i (a+b)} - 1} \frac{\Gamma(c+1)\Gamma(-c-d+k-l)}{\Gamma(c+1-l)\Gamma(-c-d+l)}H(a+k, b).
\end{align*} 
\end{thm}

\begin{remark}\upshape
For definiteness, we have stated Theorem \ref{mainth3} under the assumption that $|w_1/w_2| < 1 - \delta$. It will be clear from the full description of the method in Section \ref{methodsec} that the case $|w_2/w_1| < 1 - \delta$ can be treated similarly. 
The case $|w_1| \asymp |w_2|$ can also be handled by similar steps, but in this case the coefficients depend on the quotient $\alpha := w_2/w_1$. In fact, a slight modification of the proof of Theorem \ref{mainth3} yields the following result (see Remark \ref{thmEremark}): If $a,b,c,d \in \R \setminus \Z$ satisfy $a+b, a+b+c \notin\Z$ and $w_2 = \alpha w_1$, then $F$ satisfies the following asymptotic expansion to all orders as $w_1 \to 0$ and $w_2 \to 0$ with $w_1,w_2 \in \C \setminus [0, \infty)$ such that $\alpha \in \C \setminus [0, \infty)$ and $\delta < |\alpha| < \delta^{-1}$ for some $\delta > 0$:
\begin{align}\nonumber
F(w_1, w_2) \sim &\; \sum_{k=0}^\infty\sum_{l=0}^\infty \frac{E_{kl}^{(1)}}{k!l!} w_1^k w_2^l
+ \sum_{l=0}^\infty \frac{1}{l!} \bigg\{ E_{l}^{(2)}(\alpha)\rho_{a+b+c}(w_2) w_2^{a+b+c+1+l}
	\\ \label{FexpansionE}
& + E_{l}^{(3)}(\alpha) \rho_{a+b+c}(w_1) w_1^{a+b+c+1+l} \bigg\},
\end{align}
where $E_{kl}^{(1)} = C_{kl}^{(1)}$  and the coefficients $E_{kl}^{(j)}(\alpha)  \equiv E_{kl}^{(j)}(a,b,c,d; \alpha)$, $j = 2,3$, are given by
\begin{align*}
& E_{kl}^{(2)}(\alpha) = \frac{(e^{2\pi i a} -1)(e^{2\pi i d} -1)e^{\pi i(a+b)}}{(e^{2\pi i(a+b)} - 1)(e^{2\pi i(a+b+c)}-1)}  \frac{\Gamma(d+1) (-1)^{l}}{\Gamma(d+1-l)} G_{\text{inside}}(a+l,b,c; \alpha^{-1}),
	\\
& E_{kl}^{(3)}(\alpha) = \frac{(e^{2\pi i d} -1)e^{\pi i (a+c)}}{e^{2\pi i(a+b)} - 1}  \frac{\Gamma(d+1) (-1)^{l}}{\Gamma(d+1-l)} G(a+l,c,b;\alpha).
\end{align*} 
Here $G_{\text{inside}}(a,b,c;w)$ is defined by the same formula (\ref{Gdef}) as $G(a,b,c;w)$ except that $w$ is assumed to lie inside the contour in the same component as $0$. 
\end{remark}

\begin{remark}\upshape
The function $H(a,d)$ defined in (\ref{Hdef}) can be expressed in terms of Gamma functions as follows:
$$H(a,d) = (-1 + e^{2\pi i a} - e^{2\pi i (a+d)} + e^{2\pi i d})\frac{\Gamma(a+1)\Gamma(d+1)}{\Gamma(a+d+2)}.$$
This formula is easily derived by collapsing the Pochhammer contour onto the interval $[0,1]$ and comparing the resulting expression with the definition of the Euler beta function.
\end{remark}

\begin{remark}\upshape
The function $G(a,d)$ defined in (\ref{Gdef}) can be expressed in terms of the hypergeometric function ${}_2F_1$ as follows:
\begin{align}\label{G2F1}
G(a,c,d; w) = \rho_c(w)\frac{4\pi^2w^c \, {}_2F_1(-c,a+1,a+d+2; 1/w)}{e^{-\pi i (a+d+2)}\Gamma(a+d+2)\Gamma(-a)\Gamma(-d)}, \qquad w \in \C \setminus [0, \infty).
\end{align}
Indeed, if $\tilde{G}$ denotes the function
\begin{align*}
\tilde{G}(a,c,d; w) = \int_A^{(0+,1+,0-,1-)} v^a (w - v)^{c} (1-v)^{d} dv, \qquad w \in \C \setminus (-\infty,1],
\end{align*}
where $w$ lies exterior to the contour, then $G$ and $\tilde{G}$ are related by
\begin{align}\label{GtildeGrelation}
G(a,c,d; w) = \rho_c(w) \tilde{G}(a,c,d; w), \qquad w \in \C \setminus \R.
\end{align}
The expression (\ref{G2F1}) follows because the Pochhammer integral expression for the hypergeometric function ${}_2F_1$ (see e.g., \cite[Eq. (15.6.5)]{DLMF}) implies that
\begin{align}\label{tildeG2F1}
 \tilde{G}(a,c,d; w) = \frac{4\pi^2w^c \, {}_2F_1(-c,a+1,a+d+2; 1/w)}{e^{-(a+d+2)\pi i}\Gamma(a+d+2)\Gamma(-a)\Gamma(-d)}, \qquad w \in \C \setminus (-\infty,1].
\end{align}
\end{remark}

\section{The hypergeometric case of $N =3$}\label{hypergeometricsec}
Before turning to the full description of the method and the proofs of Theorem \ref{mainth1}-\ref{mainth4}, it is helpful to consider, as motivation, the case $N = 3$ in which the integral in (\ref{pochhammerexpression2}) reduces to a hypergeometric function. 

Let $G(a,c,d; w)$ be the function defined in (\ref{Gdef}) and corresponding to (\ref{pochhammerexpression2}) with $N = 3$. Equation (\ref{G2F1}) expresses $G(a,c,d; w)$ in terms of the hypergeometric function ${}_2F_1$ and we can use known properties of this function to derive the asymptotics of $G$ as $w \to 0$ or $w \to 1$. For definiteness, we consider the limit $w \to 1$. We will show the following analog of Theorem \ref{mainth1}.

\begin{prop}[Asymptotics of $G$ as $w \to 1$]
Let $a,c,d \in \R \setminus \Z$ be such that $c+d \notin\Z$. 
Then $G$ satisfies the following asymptotic expansion to all orders as $w \to 1$ with $w \in \C \setminus [0, \infty)$:
\begin{align}\label{Gexpansion}
G(a,c,d; w) \sim &\; \sum_{k=0}^\infty \bigg\{\hat{A}_k^{(1)} (w-1)^k 
 + \hat{A}_k^{(2)} (w-1)^{c+d+1+k}\bigg\}\rho_c(w), 
\end{align}
where the coefficients $\hat{A}_k^{(j)} \equiv \hat{A}_k^{(j)}(a,c,d)$, $j = 1,2$, are given by
\begin{subequations}\label{hatAdef}
\begin{align}
& \hat{A}_k^{(1)} = \frac{e^{2\pi i d} - 1}{e^{2\pi i(c+d)} -1} \frac{\Gamma(c+1)}{\Gamma(c+1-k)k!} 
H(a,c+d-k),
	\\
& \hat{A}_k^{(2)} = \frac{(e^{2\pi ia} -1)e^{\pi i d}}{1 - e^{2\pi i(c+d)}} 
 \frac{\Gamma(a+1)}{\Gamma(a+1-k)k!} H(c,d+k). 
\end{align} 
\end{subequations} 
\end{prop}
\begin{proof}
The function ${}_2F_1(a,c,d; z)$ is an analytic function of $z$ with a branch cut along $[1, \infty)$; in particular, it is not analytic at $z = 1$. In order to find the asymptotics of $G$ as $w \to 1$, we therefore first use the hypergeometric identity (see \cite[Eq. (10.12)]{O1974})
\begin{align}\nonumber
{}_2F_1(a, b, c; z) = &\; \frac{\Gamma(c)\Gamma(c-a-b)}{\Gamma(c-a)\Gamma(c-b)} {}_2F_1(a, b, 1+a+b-c; 1-z)
	\\  \nonumber
& + \frac{\Gamma(c)\Gamma(a+b-c)}{\Gamma(a)\Gamma(b)}(1-z)^{c-a-b} {}_2F_1(c-a, c-b, 1+c-a-b; 1-z), 
	\\  \label{2F1identity}
&\hspace{6cm} z \in \C \setminus ((-\infty, 0] \cup [1, \infty)),
\end{align}
to rewrite equation (\ref{tildeG2F1}) as 
\begin{align}\nonumber
\tilde{G}(a,c,d; w) = & -\frac{4 \pi  e^{i \pi  (a+d+2)} \sin (\pi  a) \Gamma (a+1) \sin (\pi  d) \csc (\pi 
   (c+d))}{\Gamma (-c-d) \Gamma (a+c+d+2)}  w^c 
   	\\\nonumber
&\times    {}_2F_1\bigg(1+a,-c,-c-d; 1-\frac{1}{w}\bigg)
	\\\nonumber
& + \frac{4 \pi  e^{i \pi  (a+d+2)} \sin (\pi  a)  \sin (\pi  d) \Gamma (d+1)}{\Gamma (-c) \Gamma (c+d+2) \sin(\pi(c+d))}w^c \bigg(\frac{w-1}{w}\bigg)^{c+d+1}	
      	\\ \label{Fhypergeometric2}
&\times  {}_2F_1\bigg(d+1,a+c+d+2;c+d+2; 1-\frac{1}{w}\bigg), \quad w \in \C \setminus(-\infty,1].
\end{align}
The hypergeometric functions in (\ref{Fhypergeometric2}) are analytic at $w = 1$. Hence we can write (\ref{Fhypergeometric2}) as
\begin{align}\label{hyperFPidentity}
\tilde{G}(a,c,d; w) = \hat{P}_1(w) + (w-1)^{c+d+1} \hat{P}_2(w)
\end{align}
where $\hat{P}_j(w) \equiv \hat{P}_j(a,c,d; w)$, $j = 1,2$, are analytic at $w = 1$. Recalling (\ref{GtildeGrelation}), it follows that $G$ admits an expansion of the form (\ref{Gexpansion}) for some complex coefficients $\hat{A}_k^{(j)}$.
It is possible to derive the expressions (\ref{hatAdef}) for these coefficients from (\ref{Fhypergeometric2}) by employing the expansions
$$w^c \sim \sum_{j = 0}^\infty \frac{\Gamma(c+1)}{\Gamma(c+1-j)}\frac{(w-1)^j}{j!}, \qquad \bigg(\frac{w-1}{w}\bigg)^a \sim \sum_{l = 0}^\infty \frac{(-1)^l \Gamma(a+l)}{\Gamma(a)} \frac{(w-1)^{a+l}}{l!},$$
which are valid as $w \to 1$, together with the identity
$${}_2F_1(a,b,c; z) = \sum_{k=0}^\infty \frac{(a)_k(b)_k}{(c)_k} \frac{z^k}{k!},$$
where the Pochhammer symbol $(a)_k$ is defined by
$$(a)_k = \frac{\Gamma(a+k)}{\Gamma(a)} = a(a+1)(a+2) \cdots (a+k-1).$$
However, in order to arrive at the simple expressions in (\ref{hatAdef}), this approach requires a somewhat elaborate resummation of the coefficients and it is actually more convenient to derive (\ref{hatAdef}) by proceeding as in the proof of Theorem \ref{mainth1} below.
\end{proof}

\section{Description of method}\label{methodsec}
In this section, we describe our method by considering, in turn, the following four asymptotic sectors for the function $F(w_1, w_2)$ defined in (\ref{Fdef}): (a) $w_2$ near $1$, (b) $w_1$ near $0$, (c) $w_1$ and $w_2$ both near $0$, and (d) $w_1$ near $0$ and $w_2$ near $1$. The limits considered in Theorem \ref{mainth1}-\ref{mainth4} belong to these four sectors, respectively, and the proofs of these theorems will also be given.

The basic idea of our method is to show that it is possible, for each asymptotic sector under consideration, to derive a generalization of the hypergeometric identity (\ref{hyperFPidentity}) from which asymptotics to arbitrary order can be obtained by simply replacing the factors in the integrand with their asymptotic expansions. Actually, we will see in Section \ref{greensec} that it is often convenient in applications to work with the generalizations of the hypergeometric identity (\ref{hyperFPidentity}) themselves. These generalizations are presented in Proposition \ref{FPprop}-\ref{FSQprop}, respectively. Throughout the discussion, $b,c \in \R$ and $a,d \in \R \setminus \Z$ denote some given parameters and we write $F(w_1, w_2)$ for $F(a,b,c,d; w_1,w_2)$. Furthermore, we let $\mathcal{D}_0 \subset \C^2$ and $\mathcal{D}_1 \subset \C^2$ denote the domains
$$\mathcal{D}_0 = \{(w_1,w_2) \in \C^2\,|\, w_2 \in \C \setminus [0, \infty), \; w_1 \in \C \setminus ([0,\infty) \cup \gamma_{(w_2,\infty)})\}$$
and
\begin{align}\label{D1def}
\mathcal{D}_1 = \{(w_1,w_2) \in \C^2\,|\, w_1 \in \C \setminus [0, \infty), \; w_2 \in \C \setminus ((-\infty,1] \cup \gamma_{(w_1,\infty)})\},
\end{align}
where $\gamma_{(w_j,\infty)} \subset \C$ denotes a branch cut from  $w_j$  to $\infty$. These branch cuts will be needed to make certain functions below single-valued; to be specific, we henceforth choose $\gamma_{(w_j,\infty)} = \{rw_j \, | \, r \geq 1\}$.

\subsection{The sector $w_2 \to 1$}
We will determine the behavior of $F(w_1, w_2)$ for $w_2$ close to $1$ by deriving a generalization of (\ref{hyperFPidentity}).

Define $\tilde{F}(w_1, w_2) \equiv \tilde{F}(a,b,c,d; w_1, w_2)$ by
\begin{align}\label{Ftildedef}
& \tilde{F}(w_1, w_2) = \int_A^{(0+,1+,0-,1-)} v^a (v-w_1)^b (w_2- v)^{c} (1-v)^{d} dv, \qquad (w_1, w_2) \in \mathcal{D}_1,
\end{align}
where $w_1$ and $w_2$ lie exterior to the contour. Then 
\begin{align}\label{FtildeF}
F(w_1,w_2) = \rho_c(w_2) \tilde{F}(w_1,w_2), \qquad (w_1, w_2) \in \mathcal{D}_0 \cap \mathcal{D}_1,
\end{align}
where $\rho_c$ is the function in (\ref{rhodef}). 

Assuming that $c+d \notin\Z$, we define two functions $P_j:\mathcal{D}_1 \to \C$, $j = 1,2$, as follows. 
The function $P_1$ is defined (up to a constant) by the same formula as $\tilde{F}$ except that the point $w_2$ is assumed to lie inside the contour in the same component as $1$; more precisely, for $(w_1,w_2) \in \mathcal{D}_1$,
\begin{align}\label{P1def}
& P_1(w_1, w_2) = \frac{e^{2\pi i d} - 1}{e^{2\pi i(c+d)} -1} \int_A^{(0+,1+,0-,1-)} v^{a} (v-w_1)^b (w_2 - v)^{c} (1-v)^{d} dv,  
\end{align}
where $w_1$ lies outside the contour and $w_2$ lies inside the contour in the same component as $1$, see Figure \ref{Lcontourw1w2.pdf}. 
\begin{figure}
\begin{center}
      \begin{overpic}[width=.75\textwidth]{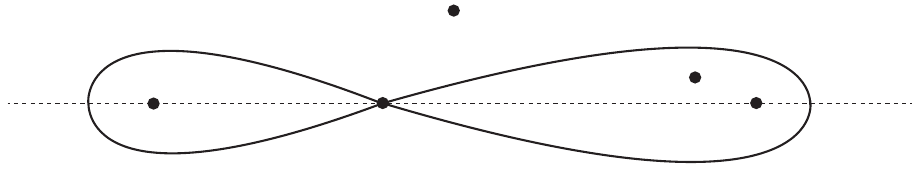}
     \put(101,6){\small $\re v$}
     \put(15.8,3.5){\small $0$}
     \put(40,3.5){\small $A$}
     \put(81.5,3.5){\small $1$}
     \put(70.8,9){\small $w_2$}
     \put(50.7,16.3){\small $w_1$}
     \end{overpic}
      \begin{figuretext}\label{Lcontourw1w2.pdf}
       In the definition of the function $P_1$, the point  $w_1$ lies exterior to the contour, whereas $w_2$ lies inside the contour in the same component as $1$. 
       \end{figuretext}
     \end{center}
\end{figure}
The function $P_2:\mathcal{D}_1 \to \C$ is defined as follows.
First, given $w_1 \in \C \setminus [0, \infty)$, we define $P_2(w_1, w_2)$ for $\re w_2 \in (0,1)$ with $\im w_2 > 0$ sufficiently small by
\begin{align}\nonumber
P_{2}&(w_1, w_2) = \frac{(e^{2\pi ia} -1)e^{\pi i d}}{1 - e^{2\pi i(c+d)}}  
	\\ \label{P2def}
&\times \int_A^{(0+,1+,0-,1-)} (w_2+s(1-w_2))^a (w_2 - w_1 + s(1-w_2))^b s^{c} (1-s)^{d} ds, 
\end{align}
where $A \in (0,1)$ and the points $\frac{w_2}{w_2-1}$ and $\frac{w_1 - w_2}{1 - w_2}$  are assumed to lie exterior to the contour. Then, for each $w_1 \in \C \setminus [0, \infty)$, we use analytic continuation to extend $P_2$ to a (single-valued) analytic function of $w_2 \in \C \setminus ((-\infty,1] \cup \gamma_{(w_1,\infty)})$. The latter step is permissible because the function $P_2$ can be analytically continued as long as the points $\frac{w_2}{w_2-1}$ and $\frac{w_1 - w_2}{1 - w_2}$ stay away from the set $\{0,1, \infty\}$, i.e., as long as $w_2 \notin \{0,1,w_1, \infty\}$.


Let $f(w_1, w_2\pm i0)$ denote the boundary values of a function $f(w_1, w_2)$ as $w_2$ approaches the real axis from above and below, respectively. The following lemma provides the desired generalization of the hypergeometric identity (\ref{hyperFPidentity}).

\begin{prop}\label{FPprop}
Suppose $a,b,c,d \in \R$ and $a,d, c+d \notin\Z$. Then the function $F$ obeys the identity
\begin{align}\label{FPidentity} 
& F(w_1, w_2) = \rho_c(w_2)\big[P_1(w_1, w_2) + (w_2-1)^{c+d+1} P_2(w_1, w_2)\big], \qquad (w_1, w_2) \in \mathcal{D}_0. 
\end{align}
\end{prop}
\begin{proof}
By (\ref{FtildeF}), it is enough to show that
\begin{align}\label{tildeFPidentity} 
& \tilde{F}(w_1, w_2) = P_1(w_1, w_2) + (w_2-1)^{c+d+1} P_2(w_1, w_2), \qquad (w_1, w_2) \in \mathcal{D}_1. 
\end{align}
It is actually enough to show that
\begin{align}\nonumber
& \tilde{F}(w_1, w_2+i0) = P_{1}(w_1, w_2+i0) + |1-w_2|^{c+d+1} e^{i\pi(c+d+1)}P_{2}(w_1, w_2+i0), 
	\\ \label{tildeFPplus}
&\hspace{7cm} w_1 \in \C \setminus [0, \infty), \ w_2 \in (0,1).
\end{align}
Indeed, for each $w_1 \in \C \setminus [0, \infty)$, both sides of the equation (\ref{tildeFPidentity}) are analytic functions of $w_2 \in \C \setminus ((-\infty,1] \cup \gamma_{(w_1,\infty)})\}$ which can be extended to multiple-valued analytic functions of  $w_2 \in \C \setminus \{0,1,w_1\}$. Hence (\ref{tildeFPidentity}) follows from (\ref{tildeFPplus}) by analytic continuation. 

Let us prove (\ref{tildeFPplus}). Let $\epsilon > 0$ be small. Let $w_1 \in \C \setminus [0, \infty)$ and $w_2 \in (0,1)$. Given $w \in \C$, let 
\begin{align*}
& S_w^+ = \{w + \epsilon e^{i\phi}  \,  | \, 0 \leq \phi \leq \pi\},
\qquad  S_w^- = \{w + \epsilon e^{i\phi}  \,  | \, -\pi \leq \phi \leq 0\},
\end{align*}
denote counterclockwise semicircles of radius $\epsilon$ centered at $w$, see Figure \ref{Swcontours.pdf}. 
Furthermore, for $A \in (0,1)$, let $L_A^j$, $j = 1, \dots, 4$, denote the contours 
\begin{align*}
& L_A^1 = [A, i\epsilon] \cup [i\epsilon, -\epsilon], && 
L_A^2 = [-\epsilon, -i\epsilon] \cup [-i\epsilon, A],
	\\
& L_A^3 = [A, 1-i\epsilon] \cup [1-i\epsilon, 1+\epsilon], && L_A^4 = [1+\epsilon, 1+ i\epsilon] \cup [1+i\epsilon, A],
\end{align*}
oriented so that $\sum_1^4 L_A^j$ is a counterclockwise contour enclosing $0$ and $1$, see Figure \ref{Lcontours.pdf}.

\begin{figure}
\medskip
\begin{center}
      \begin{overpic}[width=.12\textwidth]{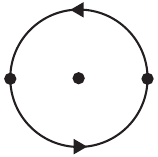}
       \put(42,33){\small $w$}
       \put(103,45){\small $w+\epsilon$}
       \put(-50,45){\small $w-\epsilon$}
       \put(39,107){\small $S_w^+$}
       \put(39,-16){\small $S_w^-$}
     \end{overpic}
      \begin{figuretext}\label{Swcontours.pdf}
       The semicircles $S_w^\pm$.
       \end{figuretext}
     \end{center}
\end{figure}

\begin{figure}
\begin{center}
      \begin{overpic}[width=.7\textwidth]{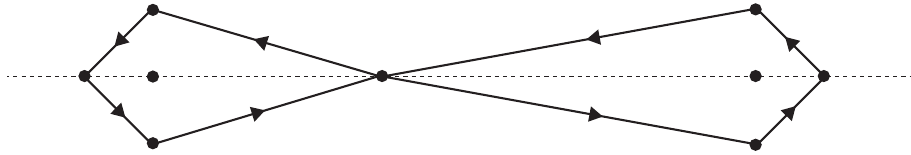}
     \put(27,15){\small $L_A^1$}
     \put(27,0){\small $L_A^2$}
     \put(63,0){\small $L_A^3$}
     \put(63,15.5){\small $L_A^4$}
     \put(6,5.5){\small $-\epsilon$}
     \put(16,5){\small $0$}
     \put(14,-2.5){\small $-i\epsilon$}
     \put(15.5,17.5){\small $i\epsilon$}
     \put(40.5,5){\small $A$}
     \put(82,5){\small $1$}
     \put(79,-2.5){\small $1-i\epsilon$}
     \put(79,18){\small $1+i\epsilon$}
     \put(91,5.5){\small $1+\epsilon$}
     \end{overpic}
     \bigskip
      \begin{figuretext}\label{Lcontours.pdf}
       The contours $L_A^j$, $j = 1, \dots, 4$.
       \end{figuretext}
     \end{center}
\end{figure}

Then we can write
\begin{align*}
\tilde{F}(w_1, w_2+i0) = &\; \bigg\{\int_{L_{w_2-\epsilon}^1} + e^{2\pi i a}\int_{L_{w_2-\epsilon}^2 + S_{w_2}^- + L_{w_2+\epsilon}^3} + e^{2\pi i(a+c+d)}\int_{L_{w_2+\epsilon}^4}
	\\
& 
- e^{2\pi i(a+d)}\int_{S_{w_2}^- + L_{w_2-\epsilon}^2} 
+ e^{2\pi i d} \int_{-L_{w_2-\epsilon}^1 + S_{w_2}^-}
- e^{2\pi i(c+d)}\int_{L_{w_2+\epsilon}^4} 
	\\
& 
- \int_{L_{w_2 + \epsilon}^3 + S_{w_2}^-}\bigg\}
v^{a} (v - w_1)^b (w_2-v)^{c} (1-v)^{d} dv
\end{align*}	
and
\begin{align*}
P_{1}(w_1, w_2+i0) = &\; \frac{e^{2\pi i d} - 1}{e^{2\pi i(c+d)} -1} \bigg\{\int_{L_{w_2-\epsilon}^1} + e^{2\pi i a}\int_{L_{w_2-\epsilon}^2 + S_{w_2}^- + L_{w_2+\epsilon}^3}
	\\
&+ e^{2\pi i (a+d+c)}\int_{L_{w_2+\epsilon}^4 + S_{w_2}^+ - L_{w_2-\epsilon}^2}
- e^{2\pi i (d+c)} \int_{L_{w_2-\epsilon}^1 + S_{w_2}^+ + L_{w_2+\epsilon}^4}
	\\
& - \int_{L_{w_2+\epsilon}^3+ S_{w_2}^-}\bigg\}
 v^{a} (v - w_1)^b (w_2 - v)^{c} (1-v)^{d} dv.
\end{align*} 
Simplification gives
\begin{align*}
\tilde{F}(w_1, w_2+i0) = &\; \bigg\{(1-e^{2\pi i d})\bigg(\int_{L_{w_2-\epsilon}^1} + e^{2\pi i a}\int_{L_{w_2-\epsilon}^2}\bigg)
	\\
& + (e^{2\pi ia} -1)\bigg(\int_{L_{w_2+\epsilon}^3 + S_{w_2}^-} + e^{2\pi i(d+c)}\int_{L_{w_2 + \epsilon}^4}
- e^{2\pi i d}\int_{S_{w_2}^-}\bigg)\bigg\}
	\\
&\times v^{a} (v - w_1)^b (w_2-v)^{c} (1-v)^{d} dv
\end{align*}	
and
\begin{align*}
P_{1}(w_1, w_2+i0) = &\; \frac{e^{2i\pi d} - 1}{e^{2\pi i(d+c)} -1} 
\bigg\{(1 - e^{2\pi i(d+c)})\bigg(\int_{L_{w_2-\epsilon}^1} + e^{2\pi i a} \int_{L_{w_2-\epsilon}^2}\bigg)
	\\
& + (e^{2\pi i a} -1)\bigg(\int_{L_{w_2+\epsilon}^3 + S_{w_2}^-} + e^{2\pi i(d+c)}\int_{L_{w_2+\epsilon}^4 + S_{w_2}^+}\bigg)\bigg\} 
	\\
& \times v^{a} (v - w_1)^b (w_2 - v)^{c} (1-v)^{d} dv.
\end{align*} 
Hence
\begin{align*}
\tilde{F}(w_1, w_2+i0) & - P_{1}(w_1, w_2+i0) =  \frac{(e^{2\pi ia} -1)\sin(c\pi) e^{i\pi d}}{\sin(\pi(d+c))} \bigg\{\int_{L_{w_2+\epsilon}^3} + e^{2\pi i(d+c)}\int_{L_{w_2+\epsilon}^4} 
	\\
& + \frac{e^{2\pi ic}(e^{2\pi i d} -1)}{1 - e^{2\pi i c}} \int_{S_{w_2}^+ + S_{w_2}^-}  \bigg\} 
 v^{a} (v - w_1)^b (w_2 - v)^{c} (1-v)^{d} dv.
\end{align*} 
Using the identity
$$(w_2 - v)^{c}  = \begin{cases} e^{-i\pi c} (v - w_2)^{c}, & v \in S_{w_2}^+ \cup L_{w_2+\epsilon}^4, \\
e^{i\pi c} (v - w_2)^{c}, & v \in S_{w_2}^- \cup L_{w_2+\epsilon}^3, \end{cases}$$
we can write this as
\begin{align*}
& \tilde{F}(w_1, w_2+i0) - P_{1}(w_1, w_2+i0) =  \frac{(e^{2\pi ia} -1)\sin(c\pi) e^{i\pi (d+c)}}{\sin(\pi(d+c))} \bigg\{\int_{L_{w_2+\epsilon}^3} + e^{2\pi i d}\int_{L_{w_2+\epsilon}^4} 
	\\
& + \frac{e^{2\pi i d} -1}{1 - e^{2\pi i c}} \int_{S_{w_2}^+} + \frac{e^{2\pi i c}(e^{2\pi i d} - 1)}{1 - e^{2\pi i c}} \int_{S_{w_2}^-} \bigg\} 
	\\
& \times v^{a} (v - w_1)^b (v - w_2)^{c} (1-v)^{d} dv, \qquad w_1 \in \C \setminus [0, \infty), \ w_2 \in (0,1),
\end{align*} 
Factoring out $\frac{1}{e^{2\pi i c} - 1}$, we obtain
\begin{align*}
& \tilde{F}(w_1, w_2+i0) - P_{1}(w_1, w_2+i0) = \frac{(e^{2\pi ia} -1)\sin(c\pi) e^{i\pi (d+c)}}{\sin(\pi(d+c))(e^{2\pi i c} - 1)} \bigg\{(e^{2\pi i c} - 1)\int_{L_{w_2+\epsilon}^3} 
	\\
& + e^{2\pi i d}(e^{2\pi i c} - 1)\int_{L_{w_2+\epsilon}^4} + (1-e^{2\pi i d}) \int_{S_{w_2}^+} 
+ e^{2\pi i c}(1- e^{2\pi i d})\int_{S_{w_2}^-} \bigg\} 
	\\
& \times v^{a} (v - w_1)^b (v - w_2)^{c} (1-v)^{d} dv.
\end{align*} 
That is,
\begin{align*}\nonumber
\tilde{F}(w_1, w_2+i0) - P_{1}(w_1, w_2+i0) = &\; \frac{(e^{2\pi ia} -1)\sin(c\pi) e^{i\pi (d+c)}}{\sin(\pi(d+c))(e^{2\pi i c} - 1)} 
\int_{w_2+\epsilon}^{(w_2+,1+,w_2-,1-)}
	\\
& \times v^{a} (v - w_1)^b (v - w_2)^{c} (1-v)^{d} dv.
\end{align*} 
Performing the change of variables $s = \frac{v-w_2}{1-w_2}$, which maps the interval $(w_2,1)$ to the interval $(0,1)$, this yields
\begin{align*}
& \tilde{F}(w_1,w_2+i0) - P_{1}(w_1,w_2+i0) = \frac{e^{i(a+d)\pi}\sin(a\pi)}{\sin(\pi(d+c))} \int_A^{(0+,1+,0-,1-)}(w_2 + s(1-w_2))^{a} 
	\\
& \times (w_2 + s(1-w_2) - w_1)^b (s(1-w_2))^{c} ((1-w_2)(1-s))^{d} (1-w_2) ds.
\end{align*} 
Comparing this expression with the definition (\ref{P2def}) of $P_2$, equation (\ref{tildeFPplus}) follows.
\end{proof}

Since both terms on the right-hand side of the identity (\ref{FPidentity}) are well-behaved for $w_2$ close to $1$, the behavior of $F$ as $w_2 \to 1$ can easily be extracted from this identity. 

\begin{proof}[Proof of Theorem \ref{mainth1}]
Suppose $a,b,c,d \in \R \setminus \Z$ and $c+d \notin\Z$. Using the identity (\ref{FPidentity}), we can easily prove Theorem \ref{mainth1}. Indeed, the functions $P_1$ and $P_2$ in (\ref{FPidentity}) admit asymptotic expansions to all orders as follows. Substituting the expansion
\begin{align*}
(w_2 - v)^{c} \sim \sum_{k=0}^\infty \frac{\Gamma(c+1)}{\Gamma(c+1-k)} \frac{(w_2-1)^k}{k!} (1-v)^{c-k}, \qquad w_2 \to 1,
\end{align*}
into the definition of $P_1(w_1, w_2)$ and recalling that $w_2$ and $1$ lie in the same component inside the contour, we find 
\begin{align}\nonumber
P_1(w_1, w_2) \sim &\; \frac{e^{2\pi i d} - 1}{e^{2\pi i(c+d)} -1} 
\sum_{k=0}^\infty \frac{\Gamma(c+1)}{\Gamma(c+1-k)} \frac{(w_2-1)^k}{k!} 
	\\ \label{P1expanded}
&\times \int_A^{(0+,1+,0-,1-)} v^{a} (v - w_1)^b (1-v)^{c+d - k} dv, \qquad w_2 \to 1,
\end{align}
where the integral on the right-hand side is exactly $G(a,b,c+d-k; w_1)$. 
Similarly, substituting the expansions
\begin{align*}
(w_2+s(1-w_2))^a \sim \sum_{m=0}^\infty \frac{\Gamma(a+1)}{\Gamma(a+1-m)} (1-s)^m \frac{(w_2 -1)^m}{m!}, \qquad w_2 \to 1,
\end{align*}
and
$$(w_2 + s(1-w_2) - w_1)^b \sim \sum_{l=0}^\infty \frac{\Gamma(b+1)}{\Gamma(b+1-l)}  (1-s)^l (1 - w_1)^{b-l}\frac{(w_2-1)^l}{l!}, \qquad w_2 \to 1,$$
into the definition (\ref{P2def}) of $P_2(w_1, w_2)$, we find, as $w_2 \to 1$,
\begin{align} \nonumber
P_2(w_1, w_2) \sim &\; \frac{(e^{2\pi ia} -1)e^{\pi i d}}{1 - e^{2\pi i(c+d)}} 
\sum_{m=0}^\infty\sum_{l=0}^\infty \frac{\Gamma(a+1)\Gamma(b+1)(1 - w_1)^{b-l}}{\Gamma(a+1-m)\Gamma(b+1-l)} 
	\\ \label{P2expanded}
&\times  \frac{(w_2-1)^{m+l}}{m! \, l!}
H(c, d+m+l).
\end{align}
Substituting (\ref{P1expanded}) and (\ref{P2expanded}) (with the summation variable $m$ replaced by $k = m+l$) into (\ref{FPidentity}), we arrive at the expansion given in Theorem \ref{mainth1}.
\end{proof}

\subsection{The sector $w_1 \to 0$}
In order to determine the behavior of $F(w_1,w_2)$ as $w_1 \to 0$, we define two functions $Q_j:\mathcal{D}_0 \to \C$, $j = 1,2$, as follows. 
The function $Q_1(w_1, w_2)$ is defined for $(w_1, w_2) \in \mathcal{D}_0$ by
\begin{align}\label{Q1def}
Q_1(w_1, w_2) 
= &\; \frac{e^{2\pi i a} - 1}{e^{2\pi i (a+b)} -1}
\int_A^{(0+,1+,0-,1-)} v^{a} (v-w_1)^b (v-w_2)^{c}  (1-v)^{d} dv, 
\end{align} 
where $A \in (0,1)$, $w_1$ lies inside the contour in the same component as $0$, and $w_2$  lies outside the contour. 
Given $w_2 \in \C \setminus [0, \infty)$, we define $Q_2(w_1, w_2)$ for $\re w_1 \in (0,1)$ with $\im w_1 < 0$ sufficiently small by
\begin{align}\label{Q2def}
Q_2(w_1, w_2) 
= &\; \frac{(e^{2\pi i d} - 1)e^{-i\pi b}}{1 - e^{-2i\pi (a+b)}} 
\int_A^{(0+,1+,0-,1-)}
 s^{a} (1 - s)^b (sw_1-w_2)^{c}  (1-sw_1)^{d} ds, 
\end{align} 
where $A \in (0,1)$ and the points $\frac{w_2}{w_1}$ and $\frac{1}{w_1}$ lie exterior to the contour. For each $w_2 \in \C \setminus [0, \infty)$, we then use analytic continuation to extend $Q_2$ to a function of $w_1 \in \C \setminus ([0,\infty) \cup \gamma_{(w_2,\infty)})$. We have the following analog of Proposition \ref{FPprop}.

\begin{prop} \label{FQprop}
Suppose $a,b,c,d \in \R$ and $a,d,a+b \notin\Z$. Then the function $F$ obeys the identity
\begin{align}\label{FQidentity} 
F(w_1, w_2) = Q_{1}(w_1, w_2) + w_1^{a+b+1} Q_2(w_1, w_2), \qquad (w_1, w_2) \in \mathcal{D}_0.
\end{align}
\end{prop}
\begin{proof}
By analyticity, is enough to show that
\begin{align}\label{FQminusidentity}
F(w_1-i0, w_2) = Q_{1}(w_1-i0, w_2) + w_1^{a+b+1} Q_{2}(w_1-i0, w_2)
\end{align}
for $w_1 \in (0,1)$ and $w_2 \in \C \setminus [0, \infty)$.

Let $\epsilon > 0$ be small. Let $w_1 \in (0,1)$ and $w_2 \in \C \setminus [0, \infty)$. 
Then
\begin{align*}
F(w_1-i0, w_2) = &\; \bigg\{\int_{L_{w_1-\epsilon}^1} 
+ e^{2\pi i (a+b)} \int_{L_{w_1-\epsilon}^2}
- e^{2\pi i a} \int_{S_{w_1}^+}
+ e^{2\pi i a} \int_{L_{w_1+\epsilon}^3}
+ e^{2\pi i(a+d)}\int_{L_{w_1+\epsilon}^4} 
	\\
&+ e^{2\pi i(a+d)}\int_{S_{w_1}^+} 
- e^{2\pi i(a+b+d)}\int_{L_{w_1-\epsilon}^2}
- e^{2\pi i d} \int_{L_{w_1-\epsilon}^1+ S_{w_1}^+ + L_{w_1+\epsilon}^4}
	\\
& - \int_{L_{w_1 + \epsilon}^3} + \int_{S_{w_1}^+}\bigg\}
v^{a} (v-w_1)^b (v-w_2)^{c}  (1-v)^{d} dv, 
\end{align*}	
where $w_2$  lies exterior to the contours. 
Moreover,
\begin{align*}
Q_1(w_1-i0, w_2) = &\; \frac{e^{2\pi i a} - 1}{e^{2\pi i (a+b)} -1} 
\bigg\{\int_{S_w^+ +L_{w_1-\epsilon}^1} + e^{2\pi i (a+b)}\int_{L_{w_1-\epsilon}^2 + S_w^- + L_{w+\epsilon}^3}
	\\
& + e^{2\pi i (a+b+d)}\int_{L_{w+\epsilon}^4 - S_w^- - L_{w-\epsilon}^2}
- e^{2\pi i d} \int_{L_{w-\epsilon}^1 + S_w^+ + L_{w+\epsilon}^4} 
+ \int_{-L_{w+\epsilon}^3}\bigg\}
	\\
& \times v^{a} (v-w_1)^b (v-w_2)^{c} (1-v)^{d} dv.
\end{align*} 
Simplification gives
\begin{align*}
F(w_1-i0, w_2) = &\; \bigg\{(1 - e^{2\pi i d}) \int_{L_{w_1-\epsilon}^1} 
+ e^{2\pi i(a+b)} (1 - e^{2\pi i d})\int_{L_{w_1-\epsilon}^2 }
	\\
&+ (e^{2\pi i a} - 1)\bigg((e^{2\pi i d} - 1) \int_{S_{w_1}^+ }
+  \int_{L_{w_1+\epsilon}^3}
+ e^{2\pi i d}\int_{L_{w_1+\epsilon}^4}\bigg) \bigg\}
	\\
&\times v^{a} (v-w_1)^b (v-w_2)^{c}  (1-v)^{d} dv
\end{align*}	
and
\begin{align*}
Q_1(w_1-i0, w_2) = &\;  \frac{e^{2\pi i a} - 1}{e^{2\pi i (a+b)} -1} 
\bigg\{(1 - e^{2\pi id})\int_{L_{w_1-\epsilon}^1} 
+ e^{2\pi i (a+b)}(1 - e^{2\pi i d})\int_{L_{w_1-\epsilon}^2 }
	\\
& + e^{2\pi i (a+b)}(1 - e^{2\pi i d})\int_{S_w^-}
+ (e^{2\pi i (a+b)} - 1)\int_{L_{w+\epsilon}^3}
	\\
& + e^{2\pi i d} (e^{2\pi i (a+b)} - 1)\int_{L_{w+\epsilon}^4}
+ (1 - e^{2\pi i d} ) \int_{S_w^+}\bigg\}
	\\
& \times v^{a} (v-w_1)^b (v-w_2)^{c} (1-v)^{d} dv.
\end{align*} 
Hence
\begin{align*}
& F(w_1-i0, w_2) - Q_1(w_1-i0, w_2) 
=  \frac{(e^{2\pi i d} - 1)e^{-2i\pi b}}{1 - e^{-2i\pi (a+b)}} \bigg\{
(1 - e^{2i\pi b})\int_{L_{w_1-\epsilon}^1} 
	\\
&+ e^{2\pi i (a+b)}(1 - e^{2i\pi b}) \int_{L_{w_1-\epsilon}^2 } 
+ e^{2\pi i b}(e^{2\pi i a} - 1) \int_{S_{w_1}^+ + S_{w_1}^- } \bigg\}
 v^{a} (v-w_1)^b (v-w_2)^{c}  (1-v)^{d} dv.
\end{align*}	
Using the identity
$$(v-w_1)^b  = \begin{cases} e^{i\pi b} (w_1-v)^b, & v \in S_w^+ \cup L_{w+\epsilon}^1, \\
e^{-i\pi b} (w_1 - b)^b, & v \in S_w^- \cup L_{w+\epsilon}^2, \end{cases}$$
we can write this as
\begin{align*}
& F(w_1-i0, w_2) - Q_{1}(w_1-i0, w_2) 
= \frac{(e^{2\pi i d} - 1)e^{-i\pi b}}{1 - e^{-2i\pi (a+b)}} \bigg\{
(1 - e^{2i\pi b})\int_{L_{w_1-\epsilon}^1} 
	\\
&+ e^{2\pi i a}(1 - e^{2i\pi b}) \int_{L_{w_1-\epsilon}^2 }
+ e^{2\pi i b}(e^{2\pi i a} - 1) \int_{S_{w_1}^+} 
+ (e^{2\pi i a} - 1) \int_{S_{w_1}^- } \bigg\}
	\\
&\times v^{a} (w_1-v)^b (v-w_2)^{c}  (1-v)^{d} dv, \qquad w_1 \in (0,1), \quad w_2 \in \C \setminus [0, \infty).
\end{align*}
That is,
\begin{align}\nonumber
F(w_1-i0, w_2) - Q_{1}(w_1-i0, w_2)  
= &\; \frac{(e^{2\pi i d} - 1)e^{-i\pi b}}{1 - e^{-2i\pi (a+b)}}
\int_{w_1 - \epsilon}^{(0+,w_1+,0-,w_1-)}
	\\ \label{FminusQ1}
& \times v^{a} (w_1 - v)^b (v-w_2)^{c}  (1-v)^{d} dv.
\end{align} 
Performing the change of variables $s = \frac{v}{w_1}$, which maps the interval $(0,w_1)$ to the interval $(0,1)$, we obtain
\begin{align*}\nonumber
F(w_1-i0, w_2) - Q_{1}(w_1-i0, w_2)  
= &\; \frac{(e^{2\pi i d} - 1)e^{-i\pi b}}{1 - e^{-2i\pi (a+b)}} w_1^{a+b+1}
\int_A^{(0+,1+,0-,1-)}
	\\ \nonumber
& \times s^{a} (1 - s)^b (sw_1-w_2)^{c}  (1-sw_1)^{d} ds.
\end{align*} 
The lemma follows.
\end{proof}

\begin{proof}[Proof of Theorem \ref{mainth2}]
Suppose $a,b,c,d \in \R \setminus \Z$ and $a+b \notin\Z$. In the same way that (\ref{FPidentity}) can be used to determine the asymptotics of $F$ as $w_2 \to 1$, the identity (\ref{FQidentity}) can be used to determine the asymptotics of $F(w_1, w_2)$ as $w_1 \to 0$. Indeed, the expansion given in Theorem \ref{mainth2} follows from (\ref{FQidentity}) after substituting the following asymptotic expansions as $w_1 \to 0$ into the definitions of $Q_1$ and $Q_2$: 
\begin{subequations}\label{thm2expansions}
\begin{align}\label{thm2expansionsa}
& (v - w_1)^b \sim \sum_{k=0}^\infty \frac{\Gamma(b+1)}{\Gamma(b+1-k)} (-1)^k v^{b-k} \frac{w_1^k}{k!}, 
	\\\label{thm2expansionsb}
& (sw_1-w_2)^{c} \sim \sum_{k=0}^\infty \frac{\Gamma(c+1)}{\Gamma(c+1-k)} \frac{s^k (-w_2)^{c-k} w_1^k}{k!},
	\\\label{thm2expansionsc}
& (1-sw_1)^{d} \sim  \sum_{k=0}^\infty \frac{\Gamma(d+1)}{\Gamma(d+1-k)} \frac{(-1)^k s^k w_1^k}{k!}.
\end{align}
\end{subequations}
\end{proof}

\subsection{The sector $w_1 \to 0$ and $w_2 \to 0$}
We next determine the asymptotics of $F(w_1,w_2)$ in the regime where both $w_1$ and $w_2$ approach zero. 
Assuming $a+b, a+b+c \notin \Z$, we define two functions $R_j:\mathcal{D}_0 \to \C$, $j = 1,2$, as follows.
The function $R_1(w_1, w_2)$ is defined for $(w_1, w_2) \in \mathcal{D}_0$ by
\begin{align*}
R_1(w_1, w_2) 
= &\; \frac{e^{2\pi i a} - 1}{e^{2\pi i (a+b + c)} -1}
\int_A^{(0+,1+,0-,1-)} v^{a} (v-w_1)^b (v-w_2)^{c}  (1-v)^{d} dv,
\end{align*} 
where $A \in (0,1)$ and both points $w_1$ and $w_2$ are assumed to lie inside the contour in the same component as $0$. 
For $0 < \re w_1 < \re w_2 < 1$ with $\im w_1 < 0$ and $\im w_2 < 0$, we define $R_2(w_1, w_2)$ by
\begin{align}\nonumber
R_2(w_1, w_2) = &\; \frac{e^{2\pi i(a+b)}(e^{2\pi ia} -1)(e^{2\pi i d}-1) e^{i\pi c}}{(e^{2\pi i(a+b)} -1)(e^{2\pi i(a+b+c)} -1)} 
	\\ \label{R2def}
&\times \int_A^{(0+,1+,0-,1-)} s^{a} (sw_2-w_1)^b (1- s)^{c}  (1-sw_2)^{d} ds,
\end{align} 
where we assume $A \in (0,1)$ is so large that $\re(Aw_2-w_1) > 0$, that the point $\frac{w_1}{w_2}$ lies inside the contour in the same component as $0$,  and that $\frac{1}{w_2}$ lies outside the contour. We then use analytic continuation to extend $R_2$ to all of $\mathcal{D}_0$. 
We have the following analog of Proposition \ref{FPprop}.

\begin{prop}\label{FRQprop}
Suppose $a,b,c,d \in \R$ and $a,d,a+b,a+b+c \notin\Z$. Then the function $F$ obeys the following identity for $(w_1, w_2) \in \mathcal{D}_0$:
\begin{align}\label{FRQidentity} 
F(w_1, w_2) = R_{1}(w_1, w_2) + w_2^{a+c+1} R_{2}(w_1, w_2) + w_1^{a+b+1} Q_2(w_1, w_2). \end{align}
\end{prop}
\begin{proof}
Both sides of (\ref{FRQidentity}) are analytic functions of $(w_1, w_2) \in \mathcal{D}_0$ which extend to multiple-valued analytic functions of 
\begin{align*}
(w_1, w_2) \in \C^2 \setminus \big(\{w_1 = 0\} \cup \{w_1 = 1\} \cup \{w_2 = 0\} \cup \{w_2 = 1\} \cup \{w_1 = w_2\}\big).
\end{align*}
Hence, by Proposition \ref{FQprop}, it is enough to show that
\begin{align}\label{QRminusidentity}
Q_{1-}(w_1, w_2) = R_{1-}(w_1, w_2) + w_2^{a+c+1} R_{2-}(w_1, w_2), \qquad 0 < w_1 < w_2 < 1,
\end{align}
where, for a function $f$, we use the short-hand notation $f_-(w_1, w_2) := f(w_1-i0, w_2-i0)$. 

Let $0 < w_1 < w_2 < 1$ and suppose $0 < \epsilon < \frac{1}{2}\min\{w_1, w_2-w_1, 1-w_2\}$. Then
\begin{align*}
& Q_{1-}(w_1, w_2) = \frac{e^{2\pi i a} - 1}{e^{2\pi i (a+b)} -1} \bigg\{\int_{L_{w_2-\epsilon}^1} 
+ e^{2\pi i (a+b+c)} \int_{L_{w_2-\epsilon}^2}
- e^{2\pi i (a+b)} \int_{S_{w_2}^+}
	\\
& + e^{2\pi i (a+b)} \int_{L_{w_2+\epsilon}^3}
+ e^{2\pi i(a+b+d)}\int_{L_{w_2+\epsilon}^4 + S_{w_2}^+} 
- e^{2\pi i(a+b+c+d)}\int_{L_{w_2-\epsilon}^2}
- e^{2\pi i d} \int_{L_{w_2-\epsilon}^1+ S_{w_2}^+ + L_{w_2+\epsilon}^4}
	\\
& - \int_{L_{w_2 + \epsilon}^3} + \int_{S_{w_2}^+}\bigg\}
v^{a} (v-w_1)^b (v-w_2)^{c}  (1-v)^{d} dv
\end{align*}	
and
\begin{align*}
R_{1-}(w_1, w_2) = &\; \frac{e^{2\pi i a} - 1}{e^{2\pi i (a+b+c)} -1} 
\bigg\{\int_{S_{w_2}^+ + L_{w_2-\epsilon}^1} + e^{2\pi i (a+b+c)}\int_{L_{w_2-\epsilon}^2 + S_{w_2}^- + L_{w+\epsilon}^3}
	\\
& + e^{2\pi i (a+b+c+d)}\int_{L_{w_2+\epsilon}^4 - S_{w_2}^- - L_{w_2-\epsilon}^2}
- e^{2\pi i d} \int_{L_{w_2-\epsilon}^1 + S_{w_2}^+ + L_{w_2+\epsilon}^4} 
- \int_{L_{w_2+\epsilon}^3}\bigg\}
	\\
& \times v^{a} (v-w_1)^b (v-w_2)^{c} (1-v)^{d} dv,
\end{align*} 
where the principal branch is used for all powers.
A computation gives
\begin{align*}
Q_{1-}&(w_1, w_2) - R_{1-}(w_1, w_2) 
= \frac{e^{2\pi i(a+b)}(e^{2\pi ia} -1)(e^{2\pi i d}-1)}{(e^{2\pi i(a+b)} -1)(e^{2\pi i(a+b+c)} -1)}
	\\
& \times \bigg\{
(1 - e^{2i\pi c})\int_{L_{w_2-\epsilon}^1} + e^{2\pi i (a+b+c)}(1 - e^{2i\pi c}) \int_{L_{w_2-\epsilon}^2 } + e^{2\pi i c}(e^{2\pi i (a+b)} - 1) \int_{S_{w_2}^+ + S_{w_2}^- } \bigg\}
	\\
&\times v^{a} (v-w_1)^b (v-w_2)^{c}  (1-v)^{d} dv, \qquad 0 < w_1 < w_2 < 1.
\end{align*}	
Using the identity
$$(v-w_2)^c  = \begin{cases} e^{i\pi c} (w_2-v)^c, & v \in S_{w_2}^+ \cup L_{w_2-\epsilon}^1, \\
e^{-i\pi c} (w_2 - v)^c, \quad & v \in S_{w_2}^- \cup L_{w_2-\epsilon}^2, \end{cases}$$
we can write this as
\begin{align*}
& Q_{1-}(w_1, w_2) - R_{1-}(w_1, w_2) 
= \frac{e^{2\pi i(a+b)}(e^{2\pi ia} -1)(e^{2\pi i d}-1) e^{i\pi c}}{(e^{2\pi i(a+b)} -1)(e^{2\pi i(a+b+c)} -1)} \bigg\{ (1 - e^{2i\pi c})\int_{L_{w_2-\epsilon}^1} 
	\\
&+ e^{2\pi i (a+b)}(1 - e^{2i\pi c}) \int_{L_{w_2-\epsilon}^2 }
+ e^{2\pi i c}(e^{2\pi i (a+b)} - 1) \int_{S_{w_2}^+ } 
+ (e^{2\pi i (a+b)} - 1) \int_{S_{w_2}^- } \bigg\}
	\\
&\times v^{a} (v-w_1)^b (w_2- v)^{c}  (1-v)^{d} dv, \qquad 0 < w_1 < w_2 < 1.
\end{align*}	
That is,
\begin{align*}
& Q_{1-}(w_1, w_2) - R_{1-}(w_1, w_2) 
= \frac{e^{2\pi i(a+b)}(e^{2\pi ia} -1)(e^{2\pi i d}-1) e^{i\pi c}}{(e^{2\pi i(a+b)} -1)(e^{2\pi i(a+b+c)} -1)} 
	\\
&\times \int_{w_2-\epsilon}^{(0+,w_2+,0-,w_2-)} v^{a} (v-w_1)^b (w_2- v)^{c}  (1-v)^{d} dv, \qquad 0 < w_1 < w_2 < 1,
\end{align*} 
where $w_1$ lies inside the contour in the same component as $0$. 
Applying the change of variables $s = \frac{v}{w_2}$, which maps the interval $(0,w_2)$ to the interval $(0,1)$, we obtain
\begin{align*}
& Q_{1-}(w_1, w_2) - R_{1-}(w_1, w_2) 
= \frac{e^{2\pi i(a+b)}(e^{2\pi ia} -1)(e^{2\pi i d}-1) e^{i\pi c}}{(e^{2\pi i(a+b)} -1)(e^{2\pi i(a+b+c)} -1)} w_2^{a + c + 1}
	\\
&\times \int_A^{(0+,1+,0-,1-)} s^{a} (sw_2-w_1)^b (1- s)^{c}  (1-sw_2)^{d} ds, \qquad 0 < w_1 < w_2 < 1,
\end{align*}
where $A \in (0,1)$  is so large that $Aw_2-w_1 > 0$.
Equation (\ref{QRminusidentity}) follows.
\end{proof}

\begin{proof}[Proof of Theorem \ref{mainth3}]
Suppose $a,b,c,d \in \R \setminus \Z$ and $a+b, a+b+c \notin\Z$. Theorem \ref{mainth3} follows by expanding the integrands in the definitions of $R_1$, $R_2$, $Q_2$ as $w_1 \to 0$ and $w_2 \to 0$, and substituting the resulting expressions into (\ref{FRQidentity}). We have stated Theorem \ref{mainth3} under the assumption that $|w_1/w_2| < 1 - \delta$; hence we use the expansion (\ref{thm2expansionsb}) of $(sw_1-w_2)^c$ and the expansion
\begin{align}\label{thm3expansion}
(sw_2 - w_1)^b \sim \sum_{k=0}^\infty \frac{\Gamma(b+1)(-1)^k}{\Gamma(b+1-k)} \frac{(sw_2)^{b-k} w_1^k}{k!}
\end{align}
of the factor $(sw_2 - w_1)^b$. 
\end{proof}

\begin{remark}\label{thmEremark}\upshape
To derive the expansion (\ref{FexpansionE}) of $F$ as $w_1, w_2 \to 0$ with $|w_1| \asymp |w_2|$, we proceed in the same way as the proof of Theorem \ref{mainth3}, i.e., we expand $R_1$, $R_2$, $Q_2$ as $w_1,w_2 \to 0$ and substitute the resulting expressions into (\ref{FRQidentity}). 
However, in this case, since $|w_1/w_2|$ is not necessarily smaller than $1$, we do not use the expansions (\ref{thm2expansionsb}) and (\ref{thm3expansion}) of $(sw_1-w_2)^c$ and $(sw_2 - w_1)^b$; instead we simplify the expressions for $Q_2$ and $R_2$ using the identities $sw_1-w_2 = w_1(s - \alpha)$ and $sw_2 - w_1 = w_2(s - \alpha^{-1})$ where $\alpha = w_2/w_1$; then the remaining factors are expanded as in the proof of Theorem \ref{mainth3}.
\end{remark}

\subsection{The sector $w_1 \to 0$ and $w_2 \to 1$}
We finally consider the behavior of $F(w_1,w_2)$ when $w_1$ is near $0$ and $w_2$ is near $1$.
Assuming that $a+b, c+d \notin \Z$, we define two functions $\tilde{Q}_1:\mathcal{D}_1 \to \C$ and $T_1:\mathcal{D}_0 \to \C$ as follows. 
We define $\tilde{Q}_1$ for $(w_1,w_2) \in \mathcal{D}_1$ by
\begin{align*}
\tilde{Q}_1(w_1, w_2) 
= &\; \frac{e^{2\pi i a} - 1}{e^{2\pi i (a+b)} -1}
\int_A^{(0+,1+,0-,1-)} v^{a} (v-w_1)^b (w_2-v)^{c}  (1-v)^{d} dv,
\end{align*} 
where $A\in (0,1)$, $w_1$ lies inside the contour in the same component as $0$, and $w_2$  lies outside the contour. 
Then
$$Q_1(w_1, w_2) = \rho_c(w_2) \tilde{Q}_1(w_1, w_2), \qquad (w_1, w_2) \in \mathcal{D}_0 \cap \mathcal{D}_1,$$
where $\rho_c$ is given by (\ref{rhodef}).
For $0 < \re w_1 < \re w_2 < 1$  with $\im w_1 <0$ and $\im w_2 > 0$, we define $T_1(w_1,w_2)$  by
\begin{align*}
T_1(w_1, w_2) 
= &\; \frac{(e^{2\pi i a} - 1)(e^{2\pi i d} -1)}{(e^{2\pi i (a+b)} -1)(e^{2\pi i(c+d)} -1)}
\int_A^{(0+,1+,0-,1-)} v^{a} (v-w_1)^b (w_2-v)^{c}  (1-v)^{d} dv,
\end{align*} 
where $w_1$ lies inside the contour in the same component as $0$, $w_2$ lies inside the contour in the same component as $1$, and we assume that $\re w_1 < A < \re w_2$. We then use analytic continuation to extend $T_1$ to all of $\mathcal{D}_1$.

\begin{prop}\label{FSQprop}
Suppose $a,b,c,d \in \R$ and $a,d,a+b, c+d \notin \Z$. Then the function $F$ obeys the identity
\begin{align}\nonumber
F(w_1, w_2) = &\; \rho_c(w_2)\big[T_1(w_1, w_2) + (w_2-1)^{c+d+1} P_2(w_1, w_2)\big]
	\\\label{FSQidentity} 
& + w_1^{a+b+1} Q_2(w_1, w_2), \qquad (w_1, w_2) \in \mathcal{D}_0 \cap \mathcal{D}_1.
\end{align}
\end{prop}
\begin{proof}
In view of (\ref{FQidentity}), it is enough to show that
\begin{align}\label{QSidentity} 
 \tilde{Q}_1(w_1, w_2) = T_1(w_1, w_2) + (w_2-1)^{c+d+1} P_2(w_1, w_2), \qquad (w_1, w_2) \in \mathcal{D}_1.
\end{align}
If $f(w_1, w_2)$ is a function of $w_1$ and $w_2$, we use the notation $f_*(w_1, w_2) := f(w_1-i0, w_2+i0)$.
By analyticity, equation (\ref{QSidentity}) will follow if we can show that the following identity holds for $0 < w_1 < w_2 < 1$:
\begin{align}\label{QSminusidentity}
& \tilde{Q}_{1*}(w_1,w_2) = T_{1*}(w_1,w_2) + |1-w_2|^{c+d+1} e^{i\pi(c+d+1)} 
P_{2*}(w_1,w_2).
\end{align}

Let $0 < w_1 < w_2 < 1$ and let $0 < \epsilon < \frac{1}{2}\min\{w_1, w_2-w_1, 1-w_2\}$. Then
\begin{align*}
\tilde{Q}_{1*}(w_1,w_2) = &\; \frac{e^{2\pi i a} - 1}{e^{2\pi i (a+b)} -1} \bigg\{\int_{L_{w_2-\epsilon}^1} 
+ e^{2\pi i (a+b)} \int_{L_{w_2-\epsilon}^2+S_{w_2}^-+L_{w_2+\epsilon}^3}
+ e^{2\pi i(a+b+c+d)}\int_{L_{w_2+\epsilon}^4} 
	\\
&
- e^{2\pi i(a+b+d)}\int_{S_{w_2}^-+L_{w_2-\epsilon}^2} 
+ e^{2\pi i d} \int_{-L_{w_2-\epsilon}^1+ S_{w_2}^-}
- e^{2\pi i (c+d)} \int_{L_{w_2+\epsilon}^4}
	\\
& - \int_{L_{w_2 + \epsilon}^3} - \int_{S_{w_2}^-}\bigg\}
v^{a} (v-w_1)^b (w_2-v)^{c}  (1-v)^{d} dv,  
\end{align*}	
and
\begin{align*}
T_{1*}(&w_1,w_2) = \frac{(e^{2\pi i a} - 1)(e^{2\pi i d} -1)}{(e^{2\pi i (a+b)} -1)(e^{2\pi i(c+d)} -1)}
\bigg\{\int_{L_{w_2-\epsilon}^1} + e^{2\pi i (a+b)}\int_{L_{w_2-\epsilon}^2 + S_{w_2}^- + L_{w_2+\epsilon}^3}
	\\
& + e^{2\pi i (a+b+c+d)}\int_{L_{w_2+\epsilon}^4 + S_{w_2}^+ - L_{w_2-\epsilon}^2}
- e^{2\pi i (c+d)} \int_{L_{w_2-\epsilon}^1 + S_{w_2}^+ + L_{w_2+\epsilon}^4} 
+ \int_{-L_{w_2+\epsilon}^3-S_{w_2}^-}\bigg\}
	\\
& \times v^{a} (v-w_1)^b (w_2-v)^{c} (1-v)^{d} dv.
\end{align*} 
It follows that
\begin{align*}
\tilde{Q}_{1*}(w_1,w_2) - T_{1*}(w_1, w_2) 
= &\; \frac{e^{2\pi i a} -1}{1 - e^{-2i\pi(c+d)} } \bigg\{
(1- e^{-2i\pi c})\int_{L_{w_2+\epsilon}^3} 
	\\
&+ e^{2\pi i d}(e^{2i\pi c} - 1) \int_{L_{w_2+\epsilon}^4 } 
+ (1 - e^{2\pi i d} ) \int_{S_{w_2}^+ + S_{w_2}^- } \bigg\}
	\\
&\times v^{a} (v-w_1)^b (w_2-v)^{c}  (1-v)^{d} dv.
\end{align*}	
Using the identity
$$(w_2-v)^c  = \begin{cases} e^{-i\pi c} (v-w_2)^c, & v \in S_{w_2}^+ \cup L_{w_2+\epsilon}^4, 	\\
e^{i\pi c} (v- w_2 )^c, & v \in S_{w_2}^- \cup L_{w_2+\epsilon}^3, \end{cases}$$
we can write this as
\begin{align*}
& \tilde{Q}_{1*}(w_1,w_2) - T_{1*}(w_1, w_2) 
= \frac{(e^{2\pi i a} -1)e^{-i\pi c}}{1- e^{-2i\pi(c+d)} } \bigg\{
(e^{2i\pi c} -1)\int_{L_{w_2+\epsilon}^3} 
+ e^{2\pi i d}(e^{2i\pi c} - 1) \int_{L_{w_2+\epsilon}^4 } 
	\\
&\qquad + (1-e^{2\pi i d}) \int_{S_{w_2}^+}
+ e^{2i\pi c} (1-e^{2\pi i d})\int_{S_{w_2}^- } \bigg\}
 v^{a} (v-w_1)^b (v-w_2)^{c}  (1-v)^{d} dv
	\\
& = \frac{(e^{2\pi i a} -1)e^{-i\pi c}}{1- e^{-2i\pi(c+d)} } 
\int_{w_2 + \epsilon}^{(w_2+,1+,w_2-,1-)} v^{a} (v-w_1)^b (v-w_2)^{c}  (1-v)^{d} dv,
\end{align*} 
where $0 < w_1 < w_2 < 1$ and $w_1$ lies exterior to the Pochhammer contour.
Performing the change of variables $s = \frac{v-w_2}{1-w_2}$, which maps the interval $(w_2,1)$ to the interval $(0,1)$, we obtain, for $0 < w_1 < w_2 < 1$ and $A  \in (0,1)$,
\begin{align*}
& \tilde{Q}_{1*}(w_1,w_2) - T_{1*}(w_1, w_2) 
= \frac{(e^{2\pi i a} -1)e^{-i\pi c}}{1- e^{-2i\pi(c+d)} }  (1-w_2)^{c + d + 1}
	\\
&\times \int_A^{(0+,1+,0-,1-)} (w_2 + s(1-w_2))^{a} (w_2 + s(1-w_2)-w_1)^b s^{c}  (1-s)^{d} ds.
\end{align*} 
Equation (\ref{QSminusidentity}) now follows from the definition (\ref{P2def}) of $P_2$.
\end{proof}

\begin{proof}[Proof of Theorem \ref{mainth4}]
By expanding the integrands in the definitions of $T_1$, $P_2$, $Q_2$ as $w_1 \to 0$ and $w_2 \to 1$, and substituting the resulting expressions into (\ref{FSQidentity}), Theorem \ref{mainth4} is obtained after a lengthy computation.
\end{proof}

\section{Two examples from SLE theory}\label{SLEsect}
In this section, in an effort to illustrate the method described above, we present two examples from SLE theory which involve Dotsenko-Fateev integrals of the form (\ref{pochhammerexpression}). 
The first example is related to the Green's function observable for two commuting SLE curves and involves an integral of the form (\ref{pochhammerexpression}) with $N = 4$ and (see (\ref{Idef}))
\begin{align*}
a_1 = a_2 = \alpha - 1, \qquad a_3 = a_4 = -\frac{\alpha}{2},
\end{align*}
where $\alpha > 1$ is a parameter. The second example is related to Schramm's formula for the same SLE system and involves an integral of the form (\ref{pochhammerexpression}) with $N = 4$ and (see (\ref{Jdef}))
$$a_1 = \alpha, \qquad a_2 = \alpha -2, \qquad a_3 = a_4 = -\frac{\alpha}{2}.$$
For each example, we derive the asymptotic estimates which are needed in order to establish that the relevant integral describes the given observable.

\subsection{Multiple SLE and Dotsenko-Fateev integrals}
Let us briefly recall the definition of multiple SLE systems and describe how Dotsenko-Fateev integrals arise when constructing observables for such systems via the screening method, see \cite{LV2018_A} for a more complete discussion. SLE$_\kappa$ curves are constructed by solving Loewner's differential equation
\begin{align} \label{LE}
\partial_t g_t(z) = \frac{2/\kappa}{g_t(z)- \xi^1_t}, \quad g_0(z) = z
\end{align}
where the driving term $\xi^1_t$ is standard Brownian motion. The curve itself is defined by $\lim_{y \to 0+} g_t^{-1}(\xi^1_t + iy)$; this is a random continuous curve growing from $0$ to $\infty$ in the upper half-plane $\mathbb{H} = \{\im z > 0\}$. If $\kappa\le 4$, the curve is simple and stays in $\mathbb{H}$ for $t > 0$. SLE curves appear as scaling limits of interfaces in various critical lattice models. It is natural to consider scaling limits of multiple interfaces simultaneously, and this leads to multiple SLE. We are interested in multiple SLE with two curves started from $\xi^1,\xi^2 \in \mathbb{R}$, respectively, and growing towards $\infty$ in $\mathbb{H}$, see \cite{D2007}. The marginal law of the SLE started from $\xi^1$ is that of a variant of SLE with an additional marked boundary point at $\xi^2$. For this variant, which is called SLE$_\kappa(2)$, the dynamics of the driving term is given by the system $d\xi^1_t = dB_t +   (2/\kappa)/(\xi^1_t - \xi^2_t)$, where $B_t$ is standard Brownian motion, $\xi^2_t := g_t(\xi^2)$, and $g_t$ solves \eqref{LE} with $\xi^1_t$ as driving term. 
An important feature is that the system can be grown in a commutative way \cite{D2007}. The extra drift term entails serious difficulties when constructing observables for such SLE systems. 

In \cite{LV2018_A}, the screening method is used to derive explicit formulas for two of the most natural SLE observables: the renormalized probability that the system passes infinitesimally near a given point in $\mathbb{H}$ (the Green's function) and the probability that the system passes to the right of a given point in $\mathbb{H}$ (Schramm's formula). The derivation in \cite{LV2018_A} proceeds as follows: First, using a CFT description of the multiple SLE system, the screening method is employed to generate explicit ``guesses'' for the given observables in terms of Dotsenko-Fateev integrals. The ``guesses'' are then shown to indeed describe the desired probabilities via a sequence of probabilistic arguments. The latter arguments rely heavily on appropriate asymptotic estimates for the relevant Dotsenko-Fateev integrals. In the remainder of this paper, we derive the estimates needed in \cite{LV2018_A}.
 
\begin{figure}
\begin{center}
\begin{overpic}[width=.55\textwidth]{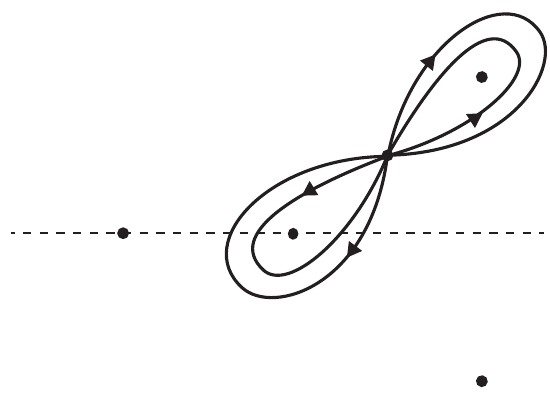}
      \put(71,40.5){\small $A$}
      \put(89.5,58.5){\small $z$}
      \put(89.5,3){\small $\bar{z}$}
      \put(21.2,25.8){\small $\xi^1$}
      \put(52,25.8){\small $\xi^2$}
      \put(103,30){\small $\re z$}
\end{overpic}
\vspace{-.2cm}
     \begin{figuretext}\label{Pochhammernew.pdf}
       The Pochhammer integration contour in (\ref{Idef}) is the composition of four loops based at the point $A = (z + \xi^2)/2$.
       \end{figuretext}
     \end{center}
\end{figure}

\subsection{Example 1: Green's function}
Set $\alpha = 8/\kappa$ and define the function $I(z,\xi^1, \xi^2)$ for $z \in \mathbb{H}$ and $-\infty < \xi^1 < \xi^2 < \infty$ by
\begin{align}\label{Idef}
I(z,\xi^1, \xi^2) = \int_A^{(z+,\xi^2+,z-,\xi^2-)} (u - z)^{\alpha -1} (u - \bar{z})^{\alpha -1} (u - \xi^1)^{-\frac{\alpha}{2}} (\xi^2 - u)^{-\frac{\alpha}{2}} du,
\end{align}
where $A = (z + \xi^2)/2$ is a basepoint and the Pochhammer integration contour is displayed in Figure \ref{Pochhammernew.pdf}. 
Moreover, define the function $\GG(z, \xi^1, \xi^2)$ for $\alpha \in (1, \infty) \setminus \Z$, $z = x + iy\in \mathbb{H}$, and $\xi^1 < \xi^2$ by
\begin{align}\label{GGdef}
\GG(z, \xi^1, \xi^2) = \frac{1}{\hat{c}}y^{\alpha + \frac{1}{\alpha} - 2} |z - \xi^1|^{1 - \alpha} |z - \xi^2|^{1 - \alpha}  \Im\big(e^{-i\pi\alpha} I(z,\xi^1,\xi^2) \big), 
\end{align}
where the constant $\hat{c} \equiv \hat{c}(\kappa)$ is given by
\begin{align}\label{hatcdef}
\hat{c} = \frac{4 \sin ^2\left(\frac{\pi  \alpha }{2}\right) \sin (\pi  \alpha ) \Gamma
   \left(1-\frac{\alpha }{2}\right) \Gamma \left(\frac{3 \alpha }{2}-1\right)}{\Gamma
   (\alpha )} \quad \text{with} \quad \alpha = \frac{8}{\kappa}.
   \end{align}
This definition of $\GG(z,\xi^1,\xi^2)$ can be extended to all $\alpha > 1$ by continuity, see \cite{LV2018_A}. It follows from (\ref{Idef}) and (\ref{GGdef}) that the product $y^{1-\frac{1}{\alpha}}  \GG(z,\xi^1, \xi^2)$ only depends on the two angles $\theta^1$ and $\theta^2$ defined by $\theta^j := \arg(z - \xi^j)$, $j = 1,2$, see \cite[Section 6.2.1]{LV2018_A}. Hence we may define the function $h$ by
\begin{align}\label{hdef}
\GG(z,\xi^1, \xi^2) = y^{\frac{1}{\alpha}-1} h(\theta^1, \theta^2), \qquad z \in \mathbb{H}, \; -\infty < \xi^1 < \xi^2 < \infty.
\end{align}
Let $\Delta \subset \R^2$ denote the triangular domain 
$$\Delta = \{(\theta^1, \theta^2) \in \R^2\, | \, 0 < \theta^1 < \theta^2 < \pi\}.$$ 
 
By applying the method of Section \ref{mainsec}-\ref{methodsec}, we can prove the following proposition which is used in \cite{LV2018_A} to derive a formula for the Green's function.

\begin{prop}[Estimates for Green's function]\label{greenprop}
Let $\alpha \geq 2$. Then the function $h(\theta^1, \theta^2)$ defined in (\ref{hdef}) is a smooth function of $(\theta^1, \theta^2) \in \Delta$ and has a continuous extension to the closure $\bar{\Delta}$ of $\Delta$. This extension satisfies
\begin{align}\label{h2ontopedge}
& h(\theta^1, \pi) = \sin^{\alpha -1}{\theta^1}, \qquad \theta^1 \in [0, \pi],
	\\
& h(\theta, \theta) = h_f(\theta), \qquad \theta \in (0, \pi),
\end{align}
where $h_f(\theta)$ is defined by
\begin{align}\nonumber
h_f(\theta) = &\; \frac{  2^{\alpha +1} \pi}{\hat{c}}
 \sin \left(\frac{\pi  \alpha }{2}\right) \sin^{2 \alpha-2}(\theta)
 	\\ \label{hfdef}
&\times  \Re\left[e^{-\frac{1}{2} i \pi  \alpha } \, _2F_1\left(1-\alpha,\alpha, 1;\frac{1}{2} (1-i \cot (\theta))\right)\right]. 
 \end{align}
Moreover, there exists a constant $C > 0$ such that
\begin{align}\label{h2estimate1}
0 \leq h(\theta^1, \theta^2) \leq C \sin^{\alpha -1}\theta^1, \qquad (\theta^1, \theta^2) \in \bar{\Delta},
\end{align}
and
\begin{align}\label{h2estimate2}
\frac{|h(\theta^1, \theta^2) - h(\theta^1, \pi)|}{\sin^{\alpha -1}\theta^1} \leq C\frac{|\pi - \theta^2|}{\sin\theta^1}, \qquad (\theta^1, \theta^2) \in \Delta.
\end{align}
\end{prop}
\begin{proof}
See Section \ref{greensec}.
\end{proof}

\subsection{Example 2: Schramm's formula}
Let $\alpha > 1$ and define the function $\mathcal{M}(z, \xi)$ by
\begin{align}\label{Mdef}
\mathcal{M}(z, \xi) = &\; y^{\alpha - 2} z^{-\frac{\alpha}{2}}(z- \xi)^{-\frac{\alpha}{2}}\bar{z}^{1-\frac{\alpha}{2}}(\bar{z} - \xi)^{1-\frac{\alpha}{2}}J(z, \xi), \qquad z \in   \mathbb{H}, \ \xi > 0,
\end{align}
where $J(z, \xi)$ is the integral defined by
\begin{align}\label{Jdef}
J(z, \xi) = \int_{\bar{z}}^z(u-z)^{\alpha}(u- \bar{z})^{\alpha - 2} u^{-\frac{\alpha}{2}}(u - \xi)^{-\frac{\alpha}{2}} du, \qquad z \in \mathbb{H}, \ \xi > 0,
\end{align}
with the integration contour from $\bar{z}$ to $z$ passing to the right of $\xi$, see Figure \ref{Jcontour.pdf}. 
\begin{figure}
\bigskip\medskip
\begin{center}
\begin{overpic}[width=.5\textwidth]{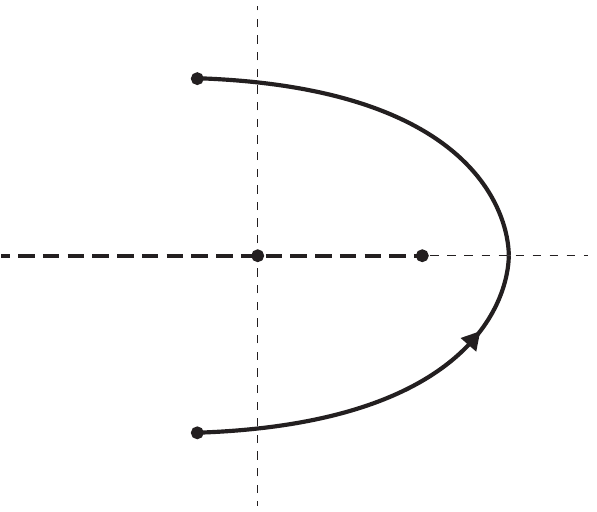}
      \put(28,71.5){\small $z$}
      \put(28,12){\small $\bar{z}$}
      \put(39.5,38){\small $0$}
      \put(69.5,37.5){\small $\xi$}
      \put(103,41.6){\small $\re z$}
      \put(39,89){\small $\im z$}
\end{overpic}
     \begin{figuretext}\label{Jcontour.pdf}
       The integration contour used in the definition (\ref{Mdef}) of $\mathcal{M}(z, \xi)$ is a path from $\bar{z}$ to $z$ which passes to the right of $\xi$. 
       \end{figuretext}
     \end{center}
\end{figure}
Moreover, define the function $P(z, \xi)$ by
\begin{align}\label{Pdef}
P(z, \xi) = \frac{1}{c_\alpha} 
\int_x^\infty \re \mathcal{M}(x' +i y, \xi) dx', \qquad z \in \mathbb{H}, \  \xi > 0.
\end{align}
where the normalization constant $c_\alpha \in \R$ is given by
\begin{align}\label{calphadef}
c_\alpha & = -\frac{2 \pi ^{3/2} \Gamma \left(\frac{\alpha -1}{2}\right) \Gamma \left(\frac{3 \alpha
   }{2}-1\right)}{\Gamma \left(\frac{\alpha }{2}\right)^2 \Gamma (\alpha )}.
\end{align}

We will prove the following proposition which establishes the properties of $P$ needed for the proofs in \cite{LV2018_A}.

\begin{prop}[Estimates for Schramm's formula]\label{schrammprop}
For each $\alpha > 1$, the function $P(z, \xi)$ defined in \eqref{Pdef} is a well-defined smooth function of $(z, \xi) \in \mathbb{H} \times (0, \infty)$ which satisfies
\begin{subequations}\label{Peverywhereinprop}
\begin{align}\label{Peverywherea1}
&  |P(z, \xi)| \leq C(\arg z)^{\alpha -1}, \qquad z \in \mathbb{H}, \ \xi > 0,
  	\\ \label{Peverywhereb1}
 & |P(z, \xi) - 1| \leq C(\pi - \arg z)^{\alpha -1}, \qquad z \in \mathbb{H}, \ \xi > 0.
\end{align}
\end{subequations}
\end{prop}
\begin{proof}
See Section \ref{schrammsec}.
\end{proof}

\section{Proof of Proposition \ref{greenprop}}\label{greensec}
By applying the method developed in Section \ref{mainsec}-\ref{methodsec}, we can determine the behavior of the Dotsenko-Fateev integral in (\ref{Idef}). This will lead to asymptotic formulas for the behavior of $h(\theta^1, \theta^2)$ near the boundary of $\Delta$ from which Proposition \ref{greenprop} will follow.

Let $F(w_1, w_2) \equiv F(a,b,c,d; w_1, w_2)$ be the function defined in (\ref{Fdef}) with $a,b,c,d$ given by
\begin{align}\label{abcddef}
a = b = \alpha - 1, \qquad c = d = -\frac{\alpha}{2},
\end{align}
i.e., for $w_1, w_2 \in \mathcal{D}_0$,
\begin{align}\label{Fdefalpha}
F(w_1, w_2) 
& = \int_A^{(0+,1+,0-,1-)} v^{\alpha -1}(v-w_1)^{\alpha -1} (v-w_2)^{-\frac{\alpha}{2}} (1-v)^{-\frac{\alpha}{2}} dv, 
\end{align}
where $A \in (0,1)$ is a basepoint and $w_1, w_2$ are assumed to lie outside the contour. 
The Pochhammer contour in (\ref{Idef}) encloses the variable points $z$ and $\xi^2$. In order to easily apply the results from Section \ref{mainsec}-\ref{methodsec}, we first need to express $h$ in terms of the integral $F$ whose contour encloses the fixed points $0$ and $1$. 
This can be achieved by applying a linear fractional transformation which maps $z$ and $\xi^2$ to $0$ and $1$, respectively.

\begin{lemma}[Representation for $h$]
For each non-integer $\alpha > 1$, the function $h$ defined in (\ref{hdef}) admits the representation
\begin{align}\label{hF}
h(\theta^1, \theta^2; \alpha) = \frac{\sin^{\alpha-1}\theta^1}{\hat{c}} \im\Big[\sigma(\theta_2) (-e^{i\theta^2})^{\alpha-1} F(w_1, w_2)\Big], \qquad (\theta^1, \theta^2) \in \Delta,
\end{align}
where $w_1 \equiv w_1(\theta^1, \theta^2)$ and $w_2\equiv w_2(\theta^1, \theta^2)$ are given by
\begin{align}\label{w1w2def}
w_1 = 1 - e^{-2i\theta^2}, \qquad w_2 = \frac{1 - e^{-2i\theta^2}}{1 - e^{-2i\theta^1}} = \frac{\sin\theta^2}{\sin\theta^1}e^{-i(\theta^2 - \theta^1)},
\end{align}
the constant $\hat{c}$ is defined in (\ref{hatcdef}), and
\begin{align}\label{sigmadef}
\sigma(\theta^2) = \begin{cases} e^{-i\pi \alpha}, & \theta^2 \geq \frac{\pi}{2}, \\
e^{i\pi \alpha}, & \theta^2 < \frac{\pi}{2}.
\end{cases}
\end{align}
\end{lemma}
\begin{proof}
Introducing the new variable $v = \frac{u-z}{\xi^2 - z}$ in (\ref{Idef}), we find
\begin{align}\nonumber
I(z, \xi^1, \xi^2) = & \int_A^{(0+,1+,0-,1-)} (v(\xi^2 -z))^{\alpha -1} (z-\bar{z} + v(\xi^2 -z))^{\alpha -1} 
	\\\nonumber
&\times (z + v(\xi^2 -z) - \xi^1)^{-\frac{\alpha}{2}} ((\xi^2 - z)(1-v))^{-\frac{\alpha}{2}} (\xi^2 -z) dv
	\\ \label{Izxi1}
= &\;  (\xi^2 -z)^{\alpha - 1} F(w_1, w_2) \times \begin{cases}
1, & x \leq \xi^2,\\
e^{2i\pi(\alpha-1)}, & x > \xi^2,
\end{cases}	
\end{align}
where $\frac{\xi + z}{z-\xi}$ and $\frac{z-\bar{z}}{z-\xi}$ are not enclosed by the contour, and the variables
\begin{align*}
 w_1 = \frac{z-\bar{z}}{z-\xi^2} =  \frac{2i}{\cot\theta^2 + i}, \qquad w_2 = \frac{\xi^1 - z}{\xi^2-z} = \frac{\cot\theta^1 + i}{\cot\theta^2 + i},	
\end{align*}
can be expressed as in (\ref{w1w2def}). 
The extra factor of $e^{2i\pi(\alpha-1)}$ in (\ref{Izxi1}) which is present for $x > \xi^2$ arises from the factor $(z-\bar{z} + v(\xi^2 -z))^{\alpha -1}$ as follows. Let $v$ belong to the contour in (\ref{Izxi1}). Then the complex number $z-\bar{z} + v(\xi^2 -z)$ lies in the upper half-plane. If $\frac{\pi}{2} \leq \theta^2 < \pi$ (i.e. if $x \leq \xi^2$), then $v - w_1$ also lies in the upper half-plane, but if $0 < \theta^2 < \pi/2$ (i.e. if $x > \xi^2$), then $v - w_1$ has crossed the negative real axis into the lower half-plane. The factor $e^{2\pi i(\alpha -1)}$ is inserted to compensate for this crossing of the branch cut. 
Equations (\ref{GGdef}), (\ref{hdef}), and (\ref{Izxi1}) give
\begin{align}\nonumber
 h(\theta^1, \theta^2) 
& = \frac{1}{\hat{c}} y^{\alpha - 1} |z - \xi^1|^{1 - \alpha} |z - \xi^2|^{1 - \alpha} |z - \xi^1|^{1 - \alpha} |z - \xi^2|^{1 - \alpha}  \Im\big[e^{-i\pi\alpha} I(z,\xi^1,\xi^2) \big]
	\\ \label{h2expression}
& =  \frac{1}{\hat{c}} \sin^{\alpha -1}(\theta^1)\sin^{\alpha -1}(\theta^2) 
\im\big[\sigma(\theta^2) (-\cot\theta^2 -i)^{\alpha -1} F(w_1, w_2)\big],
\end{align}
where $\sigma$ is given by (\ref{sigmadef}). The representation (\ref{hF}) follows.
\end{proof}

\begin{figure}
\begin{center}
      \begin{overpic}[width=.7\textwidth]{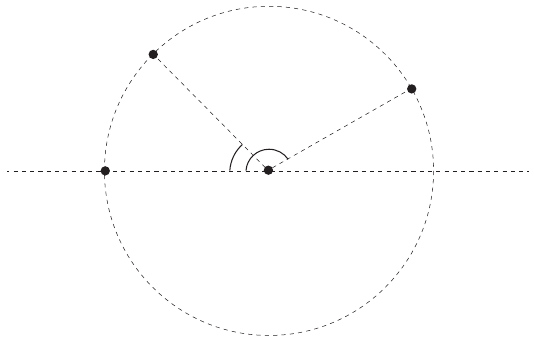}
     \put(17,29){\small $0$}
     \put(49.4,28){\small $1$}
     \put(50,38){\small $2\theta^2$}
     \put(38,34.5){\small $2\theta^1$}
     \put(23,55){\small $\frac{w_1}{w_2}$}
     \put(79,48){\small $w_1$}
     \end{overpic}
      \begin{figuretext}\label{woncircle.pdf}
       The complex numbers $w_1 = 1 - e^{-2i\theta^2}$ and $w_1/w_2 = 1 - e^{-2i\theta^1}$ lie on the circle of radius one centered at $1$.
       \end{figuretext}
     \end{center}
\end{figure}

\begin{remark}\upshape
For $(\theta^1, \theta^2) \in \Delta$, $w_1$ and $w_1/w_2$ lie on the circle of radius one centered at $1$ (see Figure \ref{woncircle.pdf}), while $w_2$ lies in the open lower half-plane, i.e., $\im w_2 < 0$.
\end{remark}

\begin{remark}\upshape
The value of $F(w_1, w_2)$ in (\ref{hF}) is, strictly speaking, not well-defined by (\ref{Fdefalpha}) for $\theta^2 = \pi/2$, because in this case $w_1 = 2 \notin \mathcal{D}_0$. However, by analytic continuation, the function $F$ in (\ref{hF}) extends to a multiple-valued function of $w_1, w_2 \in \C \setminus \{0,1\}$. Equation (\ref{hF}) then extends continuously across the line $\theta^2 = \pi/2$.
\end{remark}

Given $\delta > 0$ and $c > 0$, we define the open subsets $S_j$, $j = 1, \dots, 5$, of $\Delta$ by (see Figure \ref{sectors12fig}-\ref{sector5.pdf})
\begin{align*}
  S_1 = &\; \{(\theta^1, \theta^2) \in \Delta \, | \, 0 < \theta^2 - \theta^1 < c\sqrt{2}, \; \theta^1 > \delta, \; \theta^2 < \pi - \delta\},
	\\
  S_2 = &\; \Big\{(\theta^1, \theta^2) \in \Delta \, \Big| \, \theta^2 < c, \; \arctan\frac{\theta^1}{\theta^2} < \frac{\pi}{4} - \delta\Big\} 
  	\\
&  \cup 
  \Big\{(\theta^1, \theta^2) \in \Delta \, \Big| \, \theta^2 > \pi - c, \; \arctan\frac{\theta^1}{\pi - \theta^2} < \frac{\pi}{4} - \delta\Big\},
	\\	  
  S_3 =&\; \Big\{(\theta^1, \theta^2) \in \Delta \, \Big| \, \theta^2 > \pi - c, \; \arctan\frac{\pi - \theta^2}{\theta^1} < \frac{\pi}{4} - \delta, \; \arctan\frac{\pi - \theta^2}{\pi - \theta^1} < \frac{\pi}{4} - \delta\Big\},
	\\	  
  S_4 = &\;\Big\{(\theta^1, \theta^2) \in \Delta \, \Big| \, \theta^1 + \theta^2 < c\sqrt{2}, \; \delta <\arctan\frac{\theta^1}{\theta^2} < \frac{\pi}{4}\Big\}
  	\\
&  \cup \Big\{(\theta^1, \theta^2) \in \Delta \, \Big| \, \theta^2 - \theta^1 >\pi - c\sqrt{2}, \; \delta <\arctan\frac{\pi - \theta^2}{\theta^1} < \frac{\pi}{2} - \delta\Big\}
  	\\
&  \cup  \Big\{(\theta^1, \theta^2) \in \Delta \, \Big| \, \theta^1+\theta^2 > 2\pi - c\sqrt{2}, \; \delta <\arctan\frac{\pi - \theta^2}{\pi - \theta^1} < \frac{\pi}{4}\Big\},
  	\\	  
  S_5 =&\; \Big\{(\theta^1, \theta^2) \in \Delta \, \Big| \, \theta^1 < c, \; \theta^2 - \theta^1 > \delta, \; \theta^1 + \theta^2 < \pi - \delta\Big\}.
\end{align*}
The asymptotics of $h(\theta^1, \theta^2)$ as $(\theta^1, \theta^2)$ approaches the boundary of the triangle $\Delta$, can be described in terms of the five asymptotic sectors $\{S_j\}_1^5$. Indeed, as the next lemma shows, the first four sectors $S_j$, $j=1,\dots, 4$, correspond to the asymptotic regions treated in Proposition \ref{FPprop}-\ref{FSQprop}, respectively, while the sector $S_5$ corresponds to the region where $w_2 \to \infty$.

\begin{figure}
\begin{center}
\begin{overpic}[width=.4\textwidth]{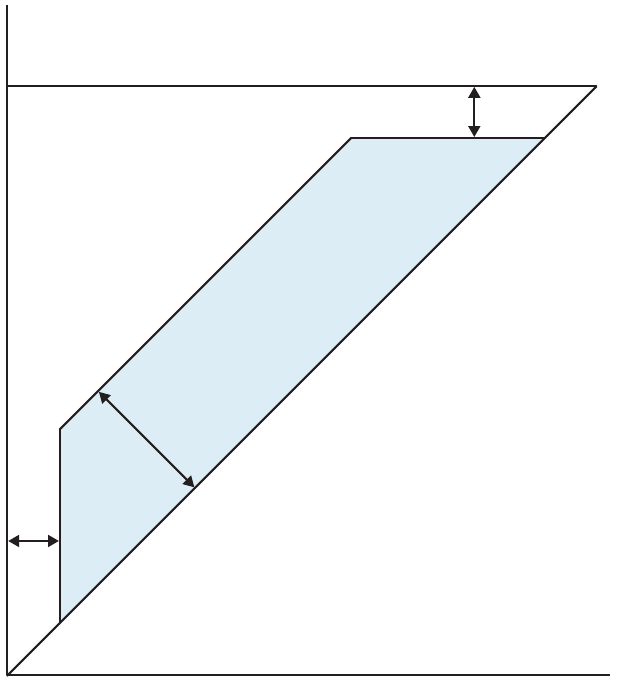}
      \put(93,-1){\small $\theta^1$}
      \put(-.5,103){\small $\theta^2$}
      \put(3.6,22.5){\small $\delta$}
      \put(65,82){\small $\delta$}
      \put(22,36){\small $c$}
      \put(-4,-1){\small $0$}
      \put(-4,86.5){\small $\pi$}
      \put(35,52){\small $S_1$}
      \end{overpic} \qquad\qquad
      \begin{overpic}[width=.413\textwidth]{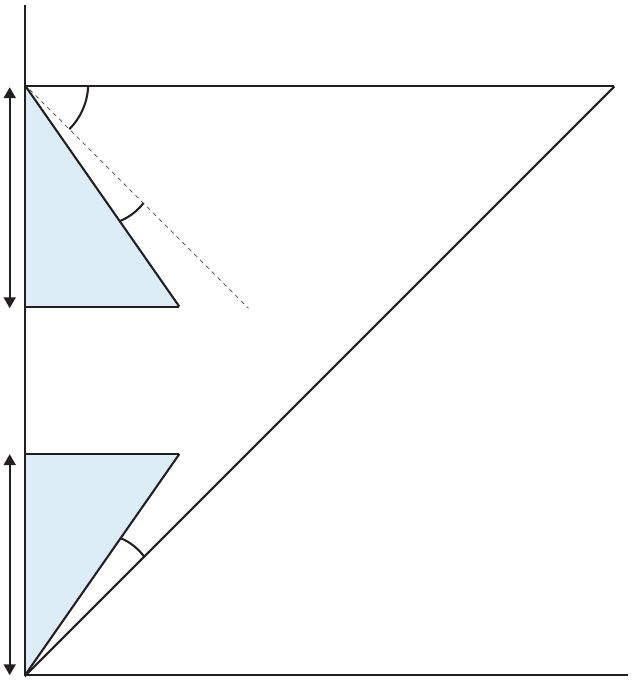}
      \put(95,-1){\small $\theta^1$}
      \put(2,102){\small $\theta^2$}
      \put(14,82){\small $\frac{\pi}{4}$}
      \put(21,21){\small $\delta$}
      \put(21,64){\small $\delta$}
      \put(-3,16){\small $c$}
      \put(-3,70){\small $c$}
      \put(8,24){\small $S_2$}
      \put(8,62){\small $S_2$}
      \end{overpic}
     \begin{figuretext}\label{sectors12fig}
       The asymptotic sectors $S_1$ and $S_2$. 
       \end{figuretext}
     \end{center}
\end{figure}

\begin{figure}
\begin{center}
\begin{overpic}[width=.41\textwidth]{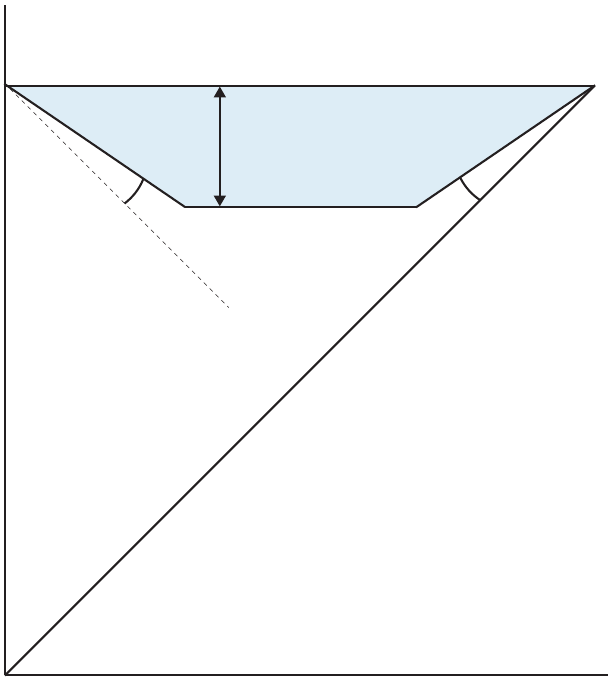}
      \put(93,-1){\small $\theta^1$}
      \put(-.5,103){\small $\theta^2$}
      \put(28,78){\small $c$}
      \put(21,68){\small $\delta$}
      \put(65,68){\small $\delta$}
      \put(-4,-1){\small $0$}
      \put(-4,86.5){\small $\pi$}
      \put(41,77){\small $S_3$}
      \end{overpic} \qquad\qquad
      \begin{overpic}[width=.41\textwidth]{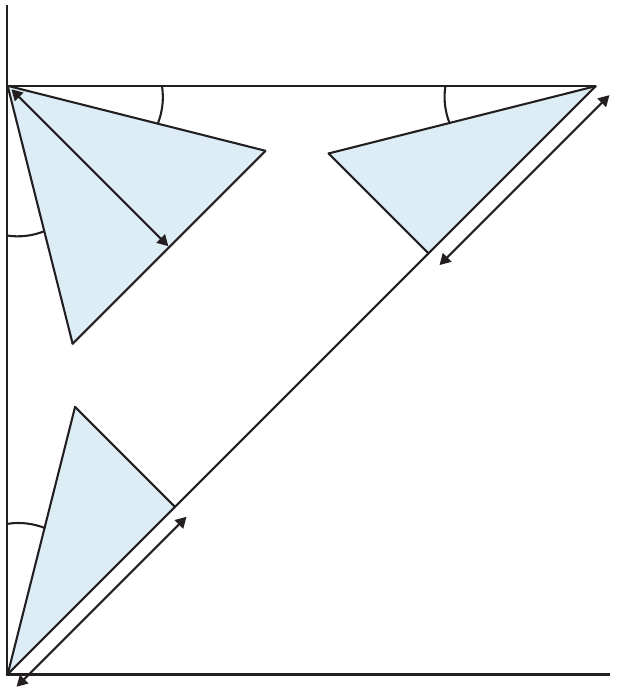}
     \put(93,-1){\small $\theta^1$}
      \put(-.5,103){\small $\theta^2$}
      \put(11,72){\small $c$}
      \put(75,71){\small $c$}
      \put(16,12){\small $c$}
      \put(3,61.5){\small $\delta$}
      \put(3,27.5){\small $\delta$}
      \put(59,83){\small $\delta$}
      \put(25,83){\small $\delta$}
      \put(-3,3){\small $0$}
      \put(-4,86.5){\small $\pi$}
      \put(21,74){\small $S_4$}
      \put(60,74){\small $S_4$}
      \put(12,25){\small $S_4$}
     \end{overpic}
     \begin{figuretext}\label{sectors34fig}
       The asymptotic sectors $S_3$ and $S_4$. 
       \end{figuretext}
     \end{center}
\end{figure}

\begin{figure}
\begin{center}
      \begin{overpic}[width=.43\textwidth]{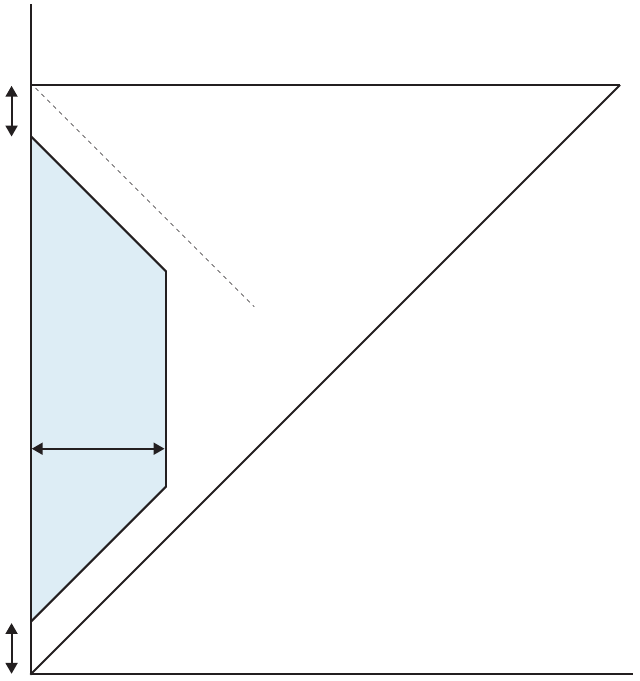}
     \put(96,-1){\small $\theta^1$}
      \put(3,102){\small $\theta^2$}
      \put(13,30){\small $c$}
      \put(-3,82){\small $\delta$}
      \put(-3,3){\small $\delta$}
      \put(12,45){\small $S_5$}
     \end{overpic}
     \begin{figuretext}\label{sector5.pdf}
       The asymptotic sector $S_5$. 
       \end{figuretext}
     \end{center}
\end{figure}

If $w \in \C$ and $A$ is a subset of $\C$, we write $\dist(w, A)$ for the Euclidean distance from $w$ to $A$; we write $\dist(w, A \cup \{\infty\}) > \epsilon$ to indicate that $\dist(w,A) > \epsilon$ and $|w| < 1/\epsilon$.

\begin{lemma}\label{w1w2lemma}
Let $\delta > 0$ and let $w_1, w_2$ be given by (\ref{w1w2def}).
Then there exist constants $c > 0$ and $\epsilon > 0$ such that the following estimates hold:
\begin{enumerate}[$(1)$]
\item $\dist(w_1, \{0, 1, \infty\}) > \epsilon$ and $|w_2 - 1| < 1-\epsilon$ for all $(\theta^1, \theta^2) \in S_1$,

\item $|w_1| < 1-\epsilon$ and $|w_2| > 1+\epsilon$ for all $(\theta^1, \theta^2) \in S_2$,

\item $|w_1| < 1-\epsilon$ and $|w_2| < 1-\epsilon$ for all $(\theta^1, \theta^2) \in S_3$,

\item $|w_1| < 1-\epsilon$ and $\dist(w_2, \{0, \infty\})> \epsilon$ for all $(\theta^1, \theta^2) \in S_4$,

\item $\dist(w_1, \{0,1, \infty\}) > \epsilon$ and $|w_2| > 1+ \epsilon$ for all $(\theta^1, \theta^2) \in S_5$.

\end{enumerate}
\end{lemma}
\begin{proof} The proof follows easily from the definition (\ref{w1w2def}) of $w_1$ and $w_2$.
\end{proof}

Assume that $\alpha \geq 2$ is such that $\frac{3\alpha}{2}, 2\alpha \notin\Z$; the cases when $\frac{3\alpha}{2}$ and/or $2\alpha$ is an integer will be considered separately. Then $a,b,c,d \in \R \setminus \Z$ and $a+b,c+d,a+b+c \notin \Z$. In fact, since $a+b = 2\alpha - 2 \geq 0$, it can be seen from Lemma \ref{w1w2lemma} and Theorem \ref{mainth2} that the function $F(w_1, w_2)$ in (\ref{hF}) is bounded in the sector $S_2$. On the other hand, since $c+d+1 = 1-\alpha < -1$, the sum in Theorem \ref{mainth1} which involves the coefficients $A_k^{(2)}$ is, in general, singular as $w_2 \to 1$. However, it turns out that the contribution from this sum to $h$ vanishes identically because of the taking of the imaginary part in (\ref{hF}). In order to see this, we need to consider the function $P_2$ of Proposition \ref{FPprop} from which the coefficients $A_k^{(2)}$ originated in more detail.

\begin{lemma}\label{imXlemma}
Let $P_2(w_1, w_2)$ denote the function defined in (\ref{P2def}) with $a,b,c,d$ given by (\ref{abcddef}) and define $X: \Delta \to \C$ by
\begin{align}\label{Xdef}
X(\theta^1, \theta^2) = \sigma(\theta^2) (-e^{i\theta^2})^{\alpha -1} e^{-\frac{i\pi \alpha}{2}}(w_2 -1)^{1 - \alpha}
P_2(w_1, w_2),
\end{align}
where the variables $w_1$ and  $w_2$ are given by (\ref{w1w2def}).
Then
\begin{align}\label{imP2zero}
\im X(\theta^1, \theta^2) = 0 \quad \text{in $\Delta$}.
\end{align}
\end{lemma}
\begin{proof}
From the definition (\ref{P2def}) of $P_2$ we see that for $w_1 \in \C \setminus [0, \infty)$ and $w_2 \in (0,1)$ we have
\begin{align}\nonumber
P_{2}&(w_1, w_2+i0) = \frac{(e^{2\pi ia} -1)e^{\pi i d}}{1 - e^{2\pi i(c+d)}}  |1-w_2|^{a+b}
	\\ \label{P2X}
&\times \int_A^{(0+,1+,0-,1-)} \Big(s- 1 - \frac{1}{w_2 -1}\Big)^a \Big(s - 1 + \frac{w_1 -1}{w_2-1}\Big)^b s^{c} (1-s)^{d} ds.
\end{align}
By definition, the value of $P_2(w_1,w_2)$ at a general point $(w_1, w_2) \in \mathcal{D}_1$  is determined by analytic continuation of (\ref{P2X}) within the connected set $\mathcal{D}_1 \subset \C^2$. The branches of the complex powers in (\ref{P2X}) are fixed by requiring that the principal branch is used initially at the basepoint $s = A \in (0,1)$; for definiteness, let us choose $A = 1/2$. This means that whenever the points 
\begin{align}\label{Awpoints}
A - 1 - \frac{1}{w_2 -1}\quad \text{and} \quad A - 1 + \frac{w_1 -1}{w_2-1},
\end{align}
cross the negative real axis during the analytic continuation, extra factors of $e^{\pm 2\pi i a}$ and $e^{\pm 2\pi i b}$, respectively, have to be inserted in (\ref{P2X}).

In order to evaluate the function $X$ in (\ref{Xdef}), we need the value of $P_2$ at points $(w_1, w_2) \in \mathcal{E}$, where $\mathcal{E}$ denotes the subset of $\C^2$ characterized by (\ref{w1w2def}), i.e.,
$$\mathcal{E} = \Big\{(w_1,w_2) = \Big(1 - e^{-2i\theta^2}, \frac{\sin\theta^2}{\sin\theta^1}e^{-i(\theta^2 - \theta^1)}\Big) \,\Big| \, (\theta^1, \theta^2) \subset \Delta\Big\}.$$
If $w_1$ and $w_2$ are given by (\ref{w1w2def}), then
$$\cot \theta^2 < \cot \theta^1, \qquad
\frac{1}{w_2-1} = \frac{\cot\theta^2 + i}{\cot\theta^1 - \cot\theta^2}, \qquad
-\frac{w_1-1}{w_2 -1} = \frac{\cot\theta^2 - i}{\cot\theta^1 - \cot\theta^2}.$$ 
Hence, we have, for all $(w_1, w_2) \in \mathcal{E}$,
\begin{align}\label{imaginaryparts}
\im\Big(A-1-  \frac{1}{w_2 -1}\Big) <  0, \qquad \im\Big(A-1+ \frac{w_1-1}{w_2 -1}\Big) > 0.
\end{align}
This shows that neither of the points in (\ref{Awpoints}) crosses the negative real axis as long as $(w_1,w_2)$ remains within $\mathcal{E}$. We can therefore find a formula for $P_2$ valid in $\mathcal{E}$ as follows. 

Let $(w_1, w_2)$ be a point in $\mathcal{E}$ corresponding to $(\theta^1, \theta^2)$ via (\ref{w1w2def}). Then
\begin{align}\label{w2minus1}
w_2  =1 + \frac{\sin(\theta^2 -\theta^1)}{\sin\theta^1} e^{-i\theta^2}.
\end{align}
Let $0 < \epsilon < |w_1-1|$ be small and let $(\tilde{w}_1(t), \tilde{w}_2(t))$, $t \in [0,1]$, be the path in $\mathcal{D}_1$ defined by $\tilde{w}_1(t) = w_1$ for all $t$, while the path $\tilde{w}_2(t)$ starts at $1-\epsilon + i0$, proceeds clockwise around the small circle of radius $\epsilon$ centered at $1$ until it reaches the point $1 + \epsilon e^{-i\theta^2}$, and then proceeds along the straight line segment $[1+\epsilon e^{-i\theta^2},w_2]$ until it reaches $w_2$. 

As $\tilde{w}_2$ moves along the arc from $1-\epsilon + i0$ to $1 + \epsilon e^{-i\theta^2}$, the point $A - 1 - \frac{1}{\tilde{w}_2(t) -1}$ crosses the negative real axis from the upper into the lower half-plane once (this adds a factor of $e^{2\pi i a}$ to (\ref{P2X})), and, provided that $\im w_1 \leq 0$ (i.e. $\theta^2 \geq \pi/2$), $A - 1 + \frac{\tilde{w}_1(t) -1}{\tilde{w}_2(t)-1}$ also crosses the negative real axis from the upper into the lower half-plane once (this adds a factor of $e^{2\pi i b}$ to (\ref{P2X})). If $\im w_1 > 0$, then $A - 1 + \frac{\tilde{w}_1(t) -1}{\tilde{w}_2(t)-1}$ does not cross the negative real axis. By varying $\theta^1$ in (\ref{w2minus1}), we see that the part of the path for which $\tilde{w}_2$ belongs to the segment $[1+\epsilon e^{-i\theta^2},w_2]$ lies in $\mathcal{E}$; hence the analytic continuation along this part adds no more factors to (\ref{P2X}). We end up with the following formula for $P_2$ on $\mathcal{E}$:
\begin{align}\nonumber
P_{2}(w_1, w_2) = &\; \frac{(e^{2\pi ia} -1)e^{\pi i d}}{1 - e^{2\pi i(c+d)}}
 e^{\pi i (a-b)} (w_2-1)^{a+b} 
	\\ \nonumber
&\times \int_A^{(0+,1+,0-,1-)} \Big(s- 1 - \frac{1}{w_2 -1}\Big)^a \Big(s - 1 + \frac{w_1 -1}{w_2-1}\Big)^b s^{c} (1-s)^{d} ds
	\\\nonumber
& \times \begin{cases} e^{2\pi i b}, \quad & \im w_1 \leq 0, \\ 1, & \im w_1 > 0, \end{cases} \   (w_1, w_2) \in \mathcal{E},
\end{align}
where $1 + \frac{1}{w_2 -1}$ and $1 - \frac{w_1 -1}{w_2 -1}$ lie exterior to the contour.
Substituting this formula into (\ref{Xdef}) and simplifying, we find
\begin{align*}
X(&\theta^2, \theta^2) =  (-e^{i\theta^2})^{\alpha -1} (w_2-1)^{\alpha -1} e^{2\pi i \alpha}
	\\ \nonumber
&\times \int_A^{(0+,1+,0-,1-)} \Big(s- 1 - \frac{1}{w_2 -1}\Big)^{\alpha-1} \Big(s - 1 + \frac{w_1 -1}{w_2-1}\Big)^{\alpha-1} s^{-\frac{\alpha}{2}} (1-s)^{-\frac{\alpha}{2}} ds,
\end{align*}
where $w_1, w_2$ are given by (\ref{w1w2def}). But
$$(-e^{i\theta^2})^{\alpha -1} (w_2-1)^{\alpha -1} 
= (\sin^{\alpha -1}\theta^2) (\cot\theta^1 - \cot\theta^2)^{\alpha -1} e^{-\pi i(\alpha-1)},$$
and, by (\ref{imaginaryparts}),
\begin{align*}
& \Big(s-1-  \frac{1}{w_2 -1}\Big)^{\alpha -1} \Big(s-1+ \frac{w_1-1}{w_2 -1}\Big)^{\alpha -1} 	
 = \Big(\frac{1 + ((1-s)\cot\theta^1 + s\cot\theta^2)^2}{(\cot\theta^1 - \cot\theta^2)^2}\Big)^{\alpha -1}.
\end{align*}
Hence
\begin{align}\nonumber
X(\theta^2, \theta^2) = & - \sin^{\alpha -1}(\theta^2) (\cot\theta^2 - \cot\theta^1)^{\alpha -1} e^{\pi i \alpha}
	\\  \label{X3}
&\times \int_A^{(0+,1+,0-,1-)} \Big(\frac{1 + ((1-s)\cot\theta^1 + s\cot\theta^2)^2}{(\cot\theta^1 - \cot\theta^2)^2}\Big)^{\alpha -1} s^{-\frac{\alpha}{2}} (1-s)^{-\frac{\alpha}{2}} ds.
\end{align}
If $g(s)$ is an analytic function, then the general identity $\overline{\int_\gamma g(w) dw} = \int_{\bar{\gamma}} \overline{g(\bar{v})} dv$ implies 
\begin{align*}
& \overline{\int_A^{(0+,1+,0-,1-)} g(s)  s^c(1-s)^d ds}
 = \int_A^{(0-,1-,0+,1+)} \overline{g(\bar{s})} s^c (1-s)^d  ds
	\\
& = \frac{-1+e^{-2\pi i c} - e^{-2\pi i(c+d)} + e^{-2\pi i d}}{-1+e^{2\pi i c} - e^{2\pi i(c+d)} + e^{2\pi i d}} \int_A^{(0+,1+,0-,1-)} \overline{g(\bar{s})} s^c (1-s)^d  ds
	\\
& = e^{-2\pi i (c+d)} \int_A^{(0+,1+,0-,1-)} \overline{g(\bar{s})} s^c (1-s)^d  ds.
\end{align*}
Using this identity to compute the imaginary part of (\ref{X3}) we arrive at
\begin{align*}
 \im X(\theta^1, \theta^2) =  & - \sin^{\alpha -1}(\theta^2) (\cot\theta^2 - \cot\theta^1)^{\alpha -1} 
\int_A^{(0+,1+,0-,1-)} \big(e^{\pi i \alpha} -
 e^{-\pi i \alpha} e^{-2i\pi(-\frac{\alpha}{2} - \frac{\alpha}{2})} 
\big) 
	\\
&\times \Big(\frac{1 + ((1-s)\cot\theta^1 + s\cot\theta^2)^2}{(\cot\theta^1 - \cot\theta^2)^2}\Big)^{\alpha -1} s^{-\frac{\alpha}{2}} (1-s)^{-\frac{\alpha}{2}} ds = 0.
\end{align*}
This completes the proof of the lemma. 
\end{proof}

Using the identities of Proposition \ref{FPprop}-\ref{FSQprop} to replace $F$ in the expression (\ref{hF}) for $h$, and using that the contribution from $P_2$ vanishes due to Lemma \ref{imXlemma}, we arrive at the next lemma, which provides four representations for $h$ which are suitable for determining the behavior of $h$ for $(\theta^1, \theta^2) \in S_j$, $j = 1, \dots, 4$, respectively. For $(\theta^1, \theta^2) \in S_5$, we will use the original representation (\ref{hF}).

\begin{lemma}\label{hsectorlemma}
Suppose $\alpha \geq 2$ satisfies $\frac{3\alpha}{2}, 2\alpha \notin\Z$. Then, for all $(\theta^1, \theta^2) \in \Delta$, 
\begin{align}\label{hsector1}
 h(\theta^1, \theta^2) 
= &\; 
\frac{1}{\hat{c}} \sin^{\alpha -1}(\theta^1) \im\Big[\sigma(\theta^2) (-e^{i\theta^2})^{\alpha -1} e^{-\frac{i\pi \alpha}{2}}P_1(w_1, w_2)\Big],
	\\\label{hsector2}
 h(\theta^1, \theta^2) = &\; \frac{1}{\hat{c}} \sin^{\alpha -1}(\theta^1)
\im\Big[
\sigma(\theta^2) (-e^{i\theta^2})^{\alpha -1} 
\big(Q_1(w_1, w_2) + w_1^{2\alpha - 1} Q_2(w_1, w_2)
\big)
\Big],
	\\\nonumber
 h(\theta^1, \theta^2) 
=&\;  \frac{1}{\hat{c}}\sin^{\alpha-1}(\theta^1) \im\Big[\sigma(\theta_2) (-e^{i\theta^2})^{\alpha-1} 
\big(R_1(w_1, w_2) + w_2^{\frac{\alpha}{2}} R_{2}(w_1, w_2) 
	\\ \label{hsector3}
& + w_1^{2\alpha - 1} Q_2(w_1, w_2)\big)\Big],
	\\ \nonumber
 h(\theta^1, \theta^2) 
=&\;  \frac{1}{\hat{c}}\sin^{\alpha-1}(\theta^1) \im\Big[\sigma(\theta_2) (-e^{i\theta^2})^{\alpha-1} 
(e^{-\frac{i\pi \alpha}{2}}T_1(w_1, w_2) 
	\\\label{hsector4}
& + w_1^{2\alpha -1} Q_2(w_1, w_2))\Big],
\end{align}
where $w_1, w_2$ are given by (\ref{w1w2def}).
\end{lemma}

In order to prove Proposition \ref{greenprop}, we only need leading and subleading estimates on $F$, so we shall be content with this level of precision. The required bounds on the functions $P_1, Q_j, R_j, T_1, F$ are then collected in the next lemma.

\begin{lemma}\label{PQRSFlemma}
Suppose $a,b \geq 0$ and $c,d \leq 0$ satisfy $a,d,a+b, c+d,a+b+c \notin \Z$. Let $\epsilon > 0$. Then the following estimates hold:
\begin{enumerate}[$(a)$]\setlength{\itemsep}{0cm}

\item $|P_1(w_1, w_2)| \leq C$ and $|P_1(w_1, w_2) - P_1(w_1, 1)| \leq C|w_2 - 1|$ uniformly for all $(w_1, w_2) \in \C^2$ such that $\dist(w_1, \{0, 1, \infty\}) > \epsilon$ and $|w_2 - 1|< 1-\epsilon$.

\item $|Q_1(w_1, w_2)| \leq C|w_2|^c$ uniformly for all $(w_1, w_2) \in \C^2$ such that $|w_1| < 1-\epsilon$ and $\dist(w_2, \{0, 1\}) > \epsilon$.

\item $|Q_2(w_1, w_2)| \leq C|w_1|^c$ uniformly for all $(w_1, w_2) \in \C^2 \setminus \{(0,0)\}$ such that $|w_1|<1-\epsilon$ and $\dist(\frac{w_2}{w_1}, \{0,1\}) > \epsilon$.

\item $|R_1(w_1, w_2)| \leq C$ and $|R_1(w_1, w_2) - R_1(0, 0)| \leq C(|w_1| + |w_2|)$ uniformly for all $(w_1, w_2) \in \C^2$ such that $|w_1| < 1-\epsilon$ and $|w_2| < 1-\epsilon$.

\item $|R_2(w_1, w_2)| \leq C|w_2|^b$  uniformly for all $(w_1, w_2) \in \C^2 \setminus \{(0,0)\}$ such that $\dist(\frac{w_2}{w_1}, \{0,1\}) > \epsilon$ and $|w_2| < 1-\epsilon$.

\item $|T_1(w_1, w_2)| \leq C$  uniformly for all $(w_1, w_2) \in \C^2$ such that 
$|w_1|<1-\epsilon$ and $\dist(w_2, \{0,\infty\}) > \epsilon$.

\item $|F(w_1, w_2)| \leq C|w_2|^c$ uniformly for all $(w_1, w_2) \in \C^2$ such that 
$\dist(w_1, \{0,1, \infty\}) > \epsilon$ and $|w_2| > 1+ \epsilon$.
\end{enumerate}
\end{lemma}
\begin{proof}
The estimates follow directly from the definitions of the functions $P_1, Q_j, R_j, T_1$, and $F$. 
\end{proof}

\begin{remark}\upshape
   If $(w_1,w_2)$ lies on a branch cut, the bounds in Lemma \ref{PQRSFlemma} should be interpreted as saying that both the left and right boundary values obey the bounds. 
\end{remark}

We are now in a position to prove Proposition \ref{greenprop}. Indeed, since $h$ clearly is smooth in the interior of $\Delta$ and the parameter $\delta > 0$ which defines the sectors $S_j$ is arbitrary, Proposition \ref{greenprop} follows from the following result.

\begin{lemma}\label{hboundarylemma2}
Let $\alpha \geq 2$. Then the function $h(\theta^1, \theta^2)$ defined in (\ref{hdef}) satisfies the following estimates:
\begin{align}\label{hboundsa}
& |h(\theta^1, \theta^2)| \leq C\sin^{\alpha -1}\theta^1, \qquad  (\theta^1, \theta^2) \in \cup_{j=1}^5 S_j,
	\\ \label{hboundsb}
& \frac{|h(\theta^1, \theta^2) - h_f(\theta^2)|}{\sin^{\alpha -1}\theta^1} 
\leq C\frac{|\theta^2 - \theta^1|}{\sin\theta^1}, \qquad  (\theta^1, \theta^2) \in S_1,
	\\\label{hboundsc}
& \frac{|h(\theta^1, \theta^2) - \sin^{\alpha -1}{\theta^1}|}{\sin^{\alpha -1}\theta^1} \leq C\frac{\sin\theta^2}{\sin\theta^1},  \qquad  (\theta^1, \theta^2) \in S_3,
\end{align}
where $h_f(\theta)$ is defined in (\ref{hfdef}).
\end{lemma}
\begin{proof}
Let us first assume that $\alpha \geq 2$ satisfies $\frac{3\alpha}{2}, 2\alpha \notin\Z$.
Equation (\ref{hsector1}), Lemma \ref{w1w2lemma} (1), and Lemma \ref{PQRSFlemma} $(a)$ show that (\ref{hboundsa}) holds in $S_1$. 
Also, by Lemma \ref{PQRSFlemma} $(a)$, the following estimate is valid in $S_1$:
\begin{align}\label{hhsinC}
\frac{|h(\theta^1, \theta^2) - h(\theta^2, \theta^2)|}{\sin^{\alpha -1}\theta^1} 
\leq C|w_2 - 1| = C\frac{\sin(\theta^2 - \theta^1)}{\sin\theta^1}
\leq C\frac{|\theta^2 - \theta^1|}{\sin\theta^1},
\end{align}
where
$$h(\theta^2, \theta^2) 
= \frac{1}{\hat{c}} \sin^{\alpha -1}(\theta^2) \im\Big[\sigma(\theta^2) (-e^{i\theta^2})^{\alpha -1} e^{-\frac{i\pi \alpha}{2}}P_1(w_1, 1)\Big].$$
For $a,b,c,d$ given by (\ref{abcddef}), we have (cf. (\ref{tildeG2F1}))
\begin{align*}
P_1(w_1, 1) = &\; \frac{e^{2\pi i d} - 1}{e^{2\pi i(c+d)} -1} \int_A^{(0+,1+,0-,1-)} v^{a} (v-w_1)^b (1-v)^{c+d} dv
	\\
= &\; 2 i \pi  (-1+e^{i \pi  \alpha }) (-w_1)^{\alpha -1} \,
   _2F_1\left(1-\alpha ,\alpha ;1;\frac{1}{w_1}\right).
\end{align*}
It follows that $h(\theta^1, \theta^1) = h_f(\theta^1)$ where $h_f$ is given by (\ref{hfdef}). Equation (\ref{hboundsb}) then follows from (\ref{hhsinC}).

Using the fact that $\dist(\frac{w_2}{w_1}, \{0,1\}) > \epsilon$ for all $(\theta^1, \theta^2) \in \Delta$, Lemma \ref{w1w2lemma} (2) and Lemma \ref{PQRSFlemma} $(b)$ and $(c)$ imply
$$|Q_1 + w_1^{2\alpha - 1} Q_2| \leq C|w_2|^{-\frac{\alpha}{2}} + C |w_1|^{\frac{3\alpha}{2} - 1}
\leq C, \qquad (\theta^1, \theta^2) \in S_2.$$
Hence equation (\ref{hsector2}) shows that (\ref{hboundsa}) holds in $S_2$. 

Similarly, Lemma \ref{w1w2lemma} (3), and Lemma \ref{PQRSFlemma} $(c)$, $(d)$, and $(e)$ show that
$$|R_1 + w_2^{\frac{\alpha}{2}} R_2 + w_1^{2\alpha - 1} Q_2| 
\leq C + C |w_2|^{\frac{3\alpha}{2} - 1} + C |w_1|^{\frac{3\alpha}{2} - 1}
\leq C, \qquad (\theta^1, \theta^2) \in S_3.$$
Hence equation (\ref{hsector3}) implies that (\ref{hboundsa}) holds in $S_3$. 
Also, by Lemma \ref{PQRSFlemma} $(d)$, since $\alpha \geq 2$, the following estimate is valid in $S_3$:
\begin{align*}
\frac{|h(\theta^1, \theta^2) - h(\theta^1, \pi)|}{\sin^{\alpha -1}\theta^1} & \leq C(|w_1| + |w_2|)  
+ C |w_2|^{\frac{3\alpha}{2} - 1} + C |w_1|^{\frac{3\alpha}{2} - 1}
\leq C(|w_1| + |w_2|)
	\\
& \leq C|\pi - \theta^2| + C\Big|\frac{\sin\theta^2}{\sin\theta^1}\Big|
\leq C\Big|\frac{\sin\theta^2}{\sin\theta^1}\Big|,
\end{align*}
where
$$h(\theta^1, \pi) = \frac{\sin^{\alpha-1}\theta^1}{\hat{c}} \im\big[e^{-i\pi \alpha} R_1(0, 0)\big].$$
For $a,b,c,d$ given by (\ref{abcddef}), we have
\begin{align*}\nonumber
R_1(0,0) 
= & \frac{e^{2\pi i a} - 1}{e^{2\pi i (a+b + c)} -1}  
 \int_A^{(0+,1+,0-,1-)} v^{a + b + c} (1-v)^{d} dv
 	\\ \nonumber
= &\; -\frac{(e^{2 i \pi  a}-1)(e^{2 i \pi  d}-1) \Gamma (d+1) \Gamma(a+b+c+1)}{\Gamma (a+b+c+d+2)}
	\\ 
= &\; -\frac{(e^{2 i \pi \alpha }-1) (e^{- i \pi  \alpha}-1) 
\Gamma(1- \frac{\alpha}{2}) \Gamma(\frac{3\alpha}{2} -1)}{\Gamma(\alpha)}.
\end{align*} 
Taking the definition (\ref{hatcdef}) of $\hat{c}$ into account, it follows that $h(\theta^1, \pi)= \sin^{\alpha-1}\theta^1$. This proves (\ref{hboundsc}).

Lemma \ref{w1w2lemma} (4), and Lemma \ref{PQRSFlemma} $(c)$ and $(f)$ show that
$$|e^{-\frac{i\pi \alpha}{2}}T_1 + w_1^{2\alpha -1} Q_2| 
\leq C + C|w_1|^{-\frac{\alpha}{2}} \leq C, \qquad (\theta^1, \theta^2) \in S_4.
$$
Hence equation (\ref{hsector4}) implies that (\ref{hboundsa}) holds in $S_4$. 

Lemma \ref{w1w2lemma} (5), and Lemma \ref{PQRSFlemma} $(g)$ show that
$$|F| \leq C|w_2|^{\frac{3\alpha}{2}-1} \leq C, \qquad (\theta^1, \theta^2) \in S_5.$$
Hence equation (\ref{hF}) shows that (\ref{hboundsa}) holds in $S_5$. 
This completes the proof of the lemma in the case when $\frac{3\alpha}{2}$ and $2\alpha$ are not integers.

Assume now that $\frac{3\alpha}{2}$ and/or $2\alpha$ is an integer. Then some of the functions in Lemma \ref{PQRSFlemma} degenerate, so a slightly different argument is required. We do not give complete details, but outline the relevant steps.

Suppose first that $\alpha \notin\Z$ but $\frac{3\alpha}{2}$ or $2\alpha$ is an integer. Then the limit $w_2 \to 1$ can still be treated as before, because $c+d = -\alpha$ is not an integer. However, the limits involving $w_1 \to 0$ or $w_2 \to 0$ cannot be treated in the same way in general, because $a+b = 2\alpha - 2$ and/or $a+b+c = \frac{3\alpha}{2} - 2$ is an integer. However, since $\alpha \geq 2$, we have $a+b>0$ and $a+b+c >0$. Hence the integral (\ref{Fdefalpha}) defining $F$ is nonsingular at $v = 0$ (also in the limit as $w_1$ and $w_2$ approach zero). Hence, we can derive the leading behavior of $F$ in these regimes using the following alternative approach: First, we collapse the two loops of the Pochhammer contour enclosing the origin down to the interval $[0,A]$. Then we find the leading-order asymptotics by Taylor expanding the integrand as $w_1$ and/or $w_2$ approaches zero.

Assume finally that $\alpha = n \geq 2$ is an integer. This case is considered in \cite{LV2018_A}, where an expression for $h(\theta^1, \theta^2; n)$ is derived by taking the limit of the defining equation (\ref{Fdefalpha}) for $F$ as $\alpha \to n$. In order to prove (\ref{hboundsa})-(\ref{hboundsc}) in this case, we compute the limits as $\alpha \to n$ of each of the four equations in Lemma \ref{hsectorlemma}. This gives four analogous equations valid for $\alpha = n$. As above, it follows from these equations that $h$ satisfies (\ref{hboundsa})-(\ref{hboundsc}). The crucial point is that the singular contribution from $P_2$ vanishes as a consequence of (\ref{imP2zero}). 
\end{proof}

\section{Proof of Proposition \ref{schrammprop}}\label{schrammsec}

\begin{lemma}\label{Jsmoothlemma}
The function $J(z, \xi)$  defined in (\ref{Jdef}) is a well-defined smooth function of $(z, \xi) \in \mathbb{H} \times (0, \infty)$.
\end{lemma}
\begin{proof}
Since $\alpha > 1$, the integral defining $J(z, \xi)$ is convergent for each  $z \in \mathbb{H}$ and each $\xi > 0$.
To prove the smoothness of $J$, we first assume that $\alpha > 1$ is an integer. In this case the integral in (\ref{Jdef}) can be computed explicitly in terms of logarithms and powers of $z$, $\bar{z}$, $z - \xi$, and $\bar{z} - \xi$ (see \cite{LV2018_A} for the case $\alpha = 2$). Hence $J(z, \xi)$ is smooth for $(z, \xi) \in \mathbb{H} \times (0, \infty)$.

Assume $\alpha > 1$ is not an integer. Then, fixing a basepoint $A > \xi$, we can rewrite the expression (\ref{Jdef}) for $J(z, \xi)$ as 
 \begin{align}\nonumber
J(z, \xi) = & \; \frac{1}{(1 - e^{2i\pi \alpha})^2} \int_A^{(z+,\bar{z}+,z-,\bar{z}-)} (u-z)^{\alpha}(u- \bar{z})^{\alpha - 2} u^{-\frac{\alpha}{2}}(u - \xi)^{-\frac{\alpha}{2}} du, 
	\\ \label{Jdef2}
& \hspace{8cm} z \in  \mathbb{H}, \  \xi > 0,
\end{align}
where the integration contour is the composition of four loops $\{l_j\}_1^4$ based at $A$ (see Figure   \ref{Jfigureeight.pdf}) and the integrand is evaluated using analytic continuation along the contour. More precisely, the loop $l_1$ encircles $z$ once in the counterclockwise direction, $l_2$ encircles $\bar{z}$ once in the counterclockwise direction, $l_3$ encircles $z$ once in the clockwise direction, and $l_4$ encircles $\bar{z}$ once in the clockwise direction. On the first half of $l_1$, the principal branch is used, but as the contour $l_1$ encircles $z$ in the counterclockwise direction, the power $(u-z)^{\alpha}$ in the integrand picks up an additional factor of $e^{2i\pi \alpha}$ with respect to the principal branch; then, as $l_2$ encircles $\bar{z}$ in the counterclockwise direction, the power $(u-\bar{z})^{\alpha - 2}$ in the integrand picks up the factor $e^{2i\pi(\alpha-2)}$ and so on. Collapsing the contour onto a single path from $\bar{z}$ to $z$ and collecting the exponential factors, we see that (\ref{Jdef2}) reduces to (\ref{Jdef}). Since the contour in (\ref{Jdef2}) avoids the branch points, the integral in (\ref{Jdef2}) can be differentiated an unlimited number of times with respect to $z, \bar{z}$, and $\xi$. This completes the proof of the lemma.
\end{proof}

\begin{figure}
\begin{center}
\bigskip
\begin{overpic}[width=.60\textwidth]{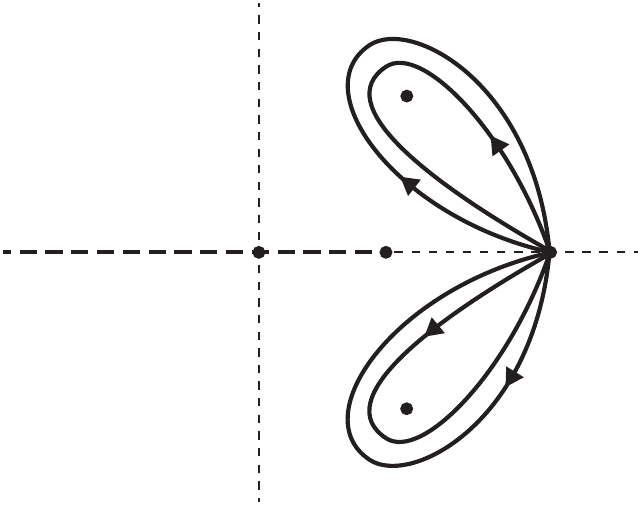}
      \put(73.5,54.5){\small $l_1$}
      \put(69,24){\small $l_2$}
      \put(59.5,48){\small $l_3$}
      \put(81.7,17){\small $l_4$}
      \put(86.5,35.5){\small $A$}
      \put(65,63){\small $z$}
      \put(65,14){\small $\bar{z}$}
      \put(37.3,35.5){\small $0$}
      \put(59,35){\small $\xi$}
      \put(103,38.6){\small $\re z$}
      \put(37,82){\small $\im z$}
\end{overpic}
     \begin{figuretext}\label{Jfigureeight.pdf}
       The integration contour in (\ref{Jdef2}) is the composition of the four loops $l_j$, $j = 1, \dots, 4$, based at the point $A > \xi$.
       \end{figuretext}
     \end{center}
\end{figure}

\begin{remark}[A main difference between the two examples]\label{differentexamplesremark}\upshape
The integrals relevant for Example 2 have better convergence properties than those of Example 1. Indeed, the Pochhammer integral expression (\ref{Jdef2}) for the function $J(z,\xi)$ relevant for Example 2 is well-defined for any $\alpha \in \R \setminus \Z$. If $\alpha > 1$ is not an integer, the expression (\ref{Jdef}) for $J$ can be recovered from (\ref{Jdef2}) by collapsing the contour down to a simple path from $\bar{z}$ to $z$. The integral in (\ref{Jdef}) converges because the exponents of the factors $(u-z)^\alpha$ and $(u-\bar{z})^{\alpha -2}$ are both $> -1$ when $\alpha > 1$. 

On the other hand, the Pochhammer contour appearing in the definition (\ref{Idef}) of the function $I(z,\xi^1, \xi^2)$ relevant for Example 1 encircles the points $z$ and $\xi^2$ and the corresponding integrand involves the factors $(u-z)^{\alpha -1}$ and $(\xi^2 - u)^{-\alpha/2}$. Thus, this contour can only be collapsed down to simple path from $z$ to $\xi^2$ if $0 < \alpha < 2$. Since Proposition \ref{greenprop} is stated under the assumption that $\alpha \geq 2$, this means that the analysis of Example 1 requires the use of a Pochhammer contour (in fact, the contour can be collapsed at $z$, but not at $\xi^2$). 
\end{remark}
 
Remark \ref{differentexamplesremark} suggests that it should be easier to prove Proposition \ref{schrammprop} than Proposition \ref{greenprop}. This is indeed the case. 
It is of course still possible to proceed as in Example 1 and use the full machinery of Section \ref{mainsec}-\ref{methodsec} to prove Proposition \ref{schrammprop}. However, in order to complement the discussion of Example 1, we will in what follows present a more elementary proof of Proposition \ref{schrammprop} which relies on direct estimates. The takeaway is that in some simpler cases one has the choice of either using the machinery of Section \ref{mainsec}-\ref{methodsec} or a more naive approach based on direct estimates.

\begin{lemma}\label{Jestimateslemma}
The function $J(z, \xi)$  defined in (\ref{Jdef}) satisfies the following estimates:
\begin{subequations}
\begin{align}\label{Jeverywhere}
& |J(z, \xi)| \leq	C |z - \xi|^{\alpha -1}, && z \in \mathbb{H}, \  \xi > 0,
	\\ \label{Jatplusinfinity}
& |J(z, \xi)| \leq C |x|^{-\frac{\alpha}{2}}|x - \xi|^{-\frac{\alpha}{2}} y^{2\alpha - 1}, && x > \xi, \  y > 0, \  \xi > 0.
	\\\label{reJeverywhere}
& |\re J(z, \xi)| \leq C y |z|^{\alpha - 2}, &&  |z| \geq 2\xi, \  z \in \mathbb{H}, \  \xi > 0,
\end{align}
\end{subequations}
where $z = x + iy$. 
\end{lemma}
\begin{proof}
To prove (\ref{Jeverywhere}), we let $z = \xi + re^{i\theta}$ and choose the following parametrization of the integration contour in (\ref{Jdef}) (see Figure \ref{Jcircle.pdf}):
\begin{align}\label{circlearoundxi}
u = \xi + re^{i\varphi}, \qquad -\theta \leq \varphi \leq \theta.
\end{align}
This yields after simplification
\begin{align}\nonumber
J(\xi + re^{i\theta}, \xi) & = i r^{\frac{3\alpha}{2} - 1} \int_{-\theta}^\theta (e^{i\varphi} - e^{i\theta})^\alpha (e^{i\varphi} - e^{-i\theta})^{\alpha -2} (\xi + re^{i\varphi})^{-\frac{\alpha}{2}} e^{i\varphi(1 -\frac{\alpha}{2})} d\varphi,
	\\ \label{Jxiplusreitheta}
&\hspace{6cm} r > 0, \  \theta \in (0, \pi), \  \xi > 0.	
\end{align}
It follows that
$$|J(\xi + re^{i\theta}, \xi)| \leq r^{\frac{3\alpha}{2} -1}\int_{-\theta}^\theta |e^{i\varphi} - e^{i\theta}|^{\alpha} |e^{i\varphi} - e^{-i\theta}|^{\alpha -2} |\xi + re^{i\varphi}|^{-\frac{\alpha}{2}} d\varphi.$$
Since
\begin{align*}
|\xi + re^{i\varphi}| \geq \frac{|re^{i\varphi} + r|}{2},
\end{align*}
for all $r > 0$, $\varphi\in (-\pi, \pi)$, and $\xi > 0$, we obtain the estimate
\begin{align*}
|J(\xi + re^{i\theta}, \xi)|  \leq C r^{\alpha -1}\int_{-\theta}^\theta |e^{i\varphi} - e^{i\theta}|^{\alpha} |e^{i\varphi} - e^{-i\theta}|^{\alpha -2} |e^{i\varphi} +1|^{-\frac{\alpha}{2}} d\varphi.
\end{align*}
The integral remains bounded as $\theta \uparrow \pi$, because $\frac{3\alpha}{2} - 2 > -1$. 
Thus we arrive at
$$|J(\xi + re^{i\theta}, \xi)|  \leq  Cr^{\alpha -1}, \qquad r > 0, \  \theta \in (0, \pi), \  \xi > 0,$$
which is (\ref{Jeverywhere}).
\begin{figure}
\begin{center}
\begin{overpic}[width=.55\textwidth]{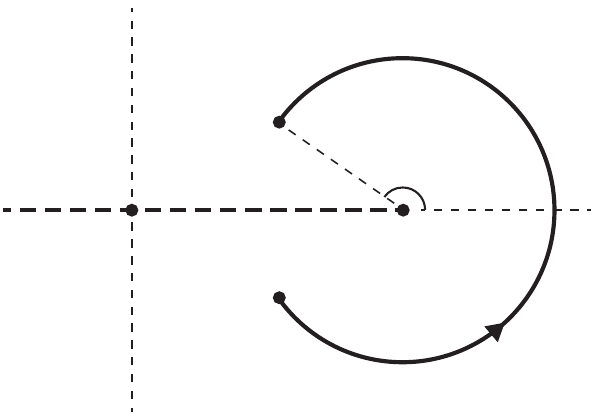}
      \put(69,39){\small $\theta$}
      \put(57,43.5){\small $r$}
      \put(43,48){\small $z$}
      \put(43,18.5){\small $\bar{z}$}
      \put(19,30){\small $0$}
      \put(66.5,29.5){\small $\xi$}
      \put(102,33.3){\small $\re z$}
      \put(18,72){\small $\im z$}
\end{overpic}
     \begin{figuretext}\label{Jcircle.pdf}
       The contour from $\bar{z} = \xi + re^{-i\theta}$ to $z = \xi + re^{i\theta}$ defined in equation (\ref{circlearoundxi}).
       \end{figuretext}
     \end{center}
\end{figure}

To prove (\ref{Jatplusinfinity}), we let $z = x + iy$ in (\ref{Jdef}) and use the parametrization $u = x + is$, $-y \leq s \leq y$, of the contour from $\bar{z}$ to $z$. Assuming that $x > \xi$, this yields
\begin{align}\label{Jrep}
J(z, \xi) & = \int_{-y}^y (is - iy)^{\alpha}(is + iy)^{\alpha - 2} (x + is)^{-\frac{\alpha}{2}}(x - \xi + is)^{-\frac{\alpha}{2}} ids.
\end{align}
It follows that
\begin{align*}
|J(z, \xi)|  \leq & \int_{-y}^y |s - y|^{\alpha}|s + y|^{\alpha - 2} |x + is|^{-\frac{\alpha}{2}}|x - \xi + is|^{-\frac{\alpha}{2}} ds
	\\
\leq &\; |x|^{-\frac{\alpha}{2}}|x - \xi|^{-\frac{\alpha}{2}} (2y)^{\alpha} \int_{-y}^y |s + y|^{\alpha - 2} ds
	\\
\leq &\; C |x|^{-\frac{\alpha}{2}}|x - \xi|^{-\frac{\alpha}{2}} y^{2\alpha - 1}, \qquad x > \xi, \  y > 0, \  \xi > 0.
\end{align*}
This proves (\ref{Jatplusinfinity}).

To prove (\ref{reJeverywhere}), we note that if $f(u)$ is an analytic function, then
\begin{align}\label{intschwartz}
\overline{\int_{\bar{z}}^z f(u) du} = - \int_{\bar{z}}^z \overline{f(\bar{u})} du.
\end{align}
Hence
\begin{align*}
2\re J(z, \xi) = & \int_{\bar{z}}^z (u - z)^\alpha (u - \bar{z})^{\alpha -2} u^{-\frac{\alpha}{2}} (u - \xi)^{-\frac{\alpha}{2}} du
	\\
& - \int_{\bar{z}}^z (u - z)^{\alpha-2} (u - \bar{z})^{\alpha} u^{-\frac{\alpha}{2}} (u - \xi)^{-\frac{\alpha}{2}} du.
\end{align*}
Since
$$(u - z)^2 - (u - \bar{z})^2 = -4iy(u - x),$$
this can be written as
\begin{align}\label{reJrepresentation}
\re J(z, \xi) = -2iy \int_{\bar{z}}^z (u - z)^{\alpha - 2}(u - \bar{z})^{\alpha - 2} u^{-\frac{\alpha}{2}} (u - \xi)^{-\frac{\alpha}{2}} (u - x) du.
\end{align}
Assuming that $z = re^{i\theta}$ satisfies $|z| > \xi$ and adopting the parametrization 
$$u = re^{i\varphi}, \qquad -\theta \leq \varphi \leq \theta,$$
of the integration contour from $\bar{z}$ to $z$, we arrive at
\begin{align*}
\re J(re^{i\theta}, \xi) = &\; 2y r^{\frac{3\alpha}{2} - 2} \int_{-\theta}^\theta (e^{i\varphi} - e^{i\theta})^{\alpha - 2}(e^{i\varphi} - e^{-i\theta})^{\alpha - 2} (re^{i\varphi} - \xi)^{-\frac{\alpha}{2}} 
	\\
& \times (e^{i\varphi} - \cos\theta) e^{i\varphi(1 -\frac{\alpha}{2})} d\varphi.
\end{align*}
In view of the estimate
$$|re^{i\varphi} - \xi| \geq \frac{r}{2}$$
valid for all $r \geq 2\xi$, $\varphi \in (0, \pi)$, and $\xi > 0$, this implies
\begin{align*}
|\re J(re^{i\theta}, \xi)| \leq &\; C y r^{\alpha - 2} \int_{-\theta}^\theta |e^{i\varphi} - e^{i\theta}|^{\alpha - 2}|e^{i\varphi} - e^{-i\theta}|^{\alpha - 2}  |e^{i\varphi} - \cos\theta| d\varphi
	\\
& \hspace{5cm}  r \geq 2\xi, \  \theta \in (0, \pi), \  \xi > 0.
\end{align*}
Since $2 \alpha -3 > -1$, the integral remains bounded as $\theta \uparrow \pi$.
This proves (\ref{reJeverywhere}).
\end{proof}

We next extend the function $J(z, \xi)$ continuously to all of $\bar{\mathbb{H}} \times (0, \infty)$. 
As suggested by (\ref{Jdef}), we define $J(x, \xi)$ for $x \in \R$ by 
\begin{align}\label{Jdefextension}
J(x, \xi) = \begin{cases} 0, & x \geq \xi, 
	\\
\int_{L_x} (u- x)^{2\alpha - 2} u^{-\frac{\alpha}{2}}(u - \xi)^{-\frac{\alpha}{2}} du, & x < \xi,
\end{cases}
\end{align}
where the integration contour $L_x$ is a loop starting and ending at $x$ which avoids the branch cut along $(-\infty, \xi)$ and which encloses $\xi$ in the counterclockwise direction, see Figure \ref{Jloop.pdf}.

\begin{figure}
\begin{center}
\begin{overpic}[width=.55\textwidth]{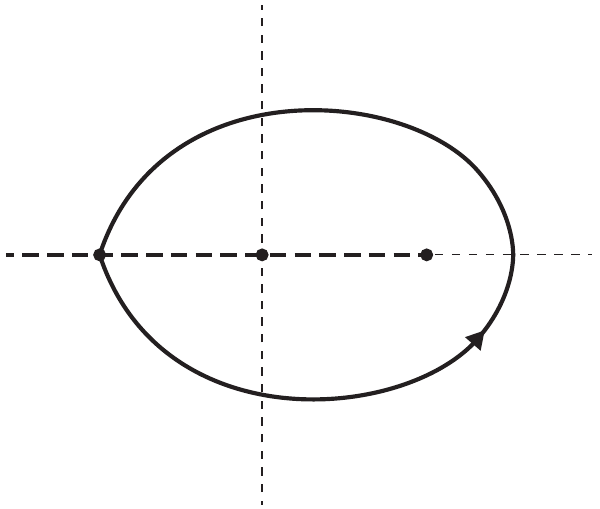}
      \put(81.5,25){\small $L_x$}
      \put(13,38.5){\small $x$}
      \put(40.5,38){\small $0$}
      \put(70,37.5){\small $\xi$}
      \put(102,41.3){\small $\re z$}
      \put(40,87){\small $\im z$}
\end{overpic}
     \begin{figuretext}\label{Jloop.pdf}
       The integration contour $L_x$ in (\ref{Jdefextension}) is a loop which encloses $\xi$ in the counterclockwise direction. 
             \end{figuretext}
     \end{center}
\end{figure}

\begin{lemma}\label{Jcontinuouslemma}
For each  $\xi > 0$, the function $J(z, \xi)$  defined by (\ref{Jdef}) and (\ref{Jdefextension}) is a continuous function of $z  \in \bar{\mathbb{H}}$.
\end{lemma}
\begin{proof}
Fix $\xi > 0$. Since $\alpha > 1$, the integral in (\ref{Jdefextension}) converges for every $x < \xi$ (including $x = 0$).
Equation (\ref{Jatplusinfinity}) implies that $z \mapsto J(z, \xi)$ is continuous at each point in $(\xi, \infty)$. Moreover, equation (\ref{Jeverywhere}) implies that $J$ is continuous at $z = \xi$.

We next show that $J$ is continuous at each point in $(-\infty, 0) \cup (0, \xi)$. 
Letting $s = \frac{\pi}{\theta}\varphi$ and simplifying, we can write (\ref{Jxiplusreitheta}) as
\begin{align*}
J(\xi + re^{i\theta}, \xi) & = \int_{-\pi}^\pi g_{r,\theta}(s) ds,
\end{align*}
where
$$g_{r,\theta}(s) = ir^{\frac{3\alpha}{2} - 1} \frac{\theta}{\pi} \big(e^{\frac{i\theta s}{\pi}} - e^{i\theta}\big)^\alpha \big(e^{\frac{i\theta s}{\pi}} - e^{-i\theta}\big)^{\alpha -2} \big(\xi + re^{\frac{i\theta s}{\pi}}\big)^{-\frac{\alpha}{2}} e^{\frac{i\theta s}{\pi}(1- \frac{\alpha}{2})}.$$
Let $\epsilon > 0$. Then
\begin{align*}
|g_{r,\theta}(s)| & \leq C \big|e^{\frac{i\theta s}{\pi}} - e^{i\theta}\big|^\alpha \big|e^{\frac{i\theta s}{\pi}} - e^{-i\theta}\big|^{\alpha -2}
	\\
& \leq C 2^\alpha \max\bigg(\bigg|\theta + \frac{\theta s}{\pi}\bigg|^{\alpha -2}, \bigg|\theta + \frac{\theta s}{\pi} - 2\pi\bigg|^{\alpha -2}\bigg)
	\\
& \leq C \max\big(|\pi + s|^{\alpha -2}, |\pi - s|^{\alpha -2}\big), \qquad s \in (-\pi, \pi),
\end{align*}
for all $r \in (\epsilon, \epsilon^{-1})$ with $|r - \xi|  > \epsilon$ and all $\theta \in [\frac{\pi}{2}, \pi]$.
Since $|\pi \pm s|^{\alpha -2} \in L^1([-\pi, \pi])$, this shows that there exists a function $G(s)$ in $L^1([-\pi, \pi])$ such that
\begin{align}\label{grthetaG}
|g_{r,\theta}(s)| \leq G(s), \qquad s \in (-\pi, \pi),
\end{align}
for all $r \in (\epsilon, \epsilon^{-1})$ with $|r - \xi|  > \epsilon$ and all $\theta \in [\frac{\pi}{2}, \pi]$. 
Since $\epsilon > 0$ was arbitrary, if the point $x_0 = \xi + r_0 e^{i\pi}$ belongs to $(-\infty, 0) \cup (0, \xi)$, dominated convergence gives
$$\lim_{z \to x_0} J(z, \xi) =\lim_{\substack{\theta \uparrow \pi \\ r \to r_0}} \int_{-\pi}^\pi g_{r,\theta}(s) ds = \int_{-\pi}^{\pi} g_{r_0,\pi}(s) ds = J(x_0, \xi).$$
This shows that $J(\cdot, \xi)$ is continuous at each point in $(-\infty, 0) \cup (0, \xi)$.

It only remains to show that $J(\cdot, \xi):\bar{\mathbb{H}} \to \C$ is continuous at $z = 0$.
To prove this, let $c = \frac{\xi}{2}$ and let $z = re^{i\theta}$ with   $\theta \in [0, \pi]$ and $r \in (0, c)$. Let $\gamma$  denote the contour from  $\bar{z}$ to $z$ used in the definition (\ref{Jdef}) of $J(z, \xi)$. We write $\gamma$ as the union of five subcontours as follows (see Figure \ref{Jcontour2.pdf}):
$$\gamma = \gamma_1 + \gamma_2 + \gamma_3 + \gamma_4 + \gamma_5,$$
where
\begin{align}\nonumber
& \gamma_1 : \{re^{i\varphi} \, | \, -\theta \leq \varphi \leq 0\}, && \gamma_2 : \{u - i0 \, | \, r \leq u \leq \xi - c \},
	\\\nonumber
& \gamma_3: \{\xi + ce^{i\varphi} \, | \, -\pi \leq \varphi \leq \pi\}, &&\gamma_4 : \{u + i0 \, | \, r \leq u \leq \xi - c \},
	\\\label{gammajdef}
& \gamma_5 : \{re^{i\varphi} \, | \, 0 \leq \varphi \leq \theta\}.
\end{align}
\begin{figure}
\begin{center}
\begin{overpic}[width=.60\textwidth]{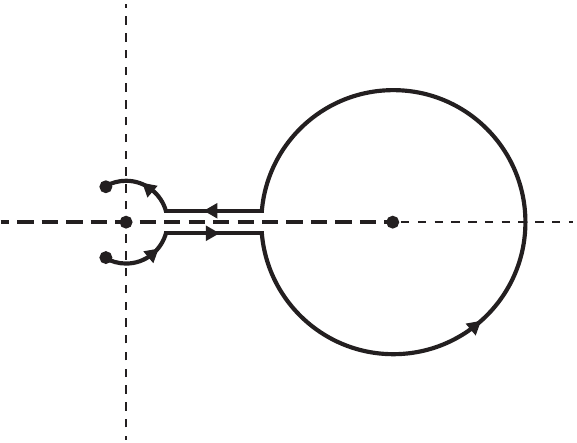}
      \put(26,28.5){\small $\gamma_1$}
      \put(35,32.4){\small $\gamma_2$}
      \put(82,16){\small $\gamma_3$}
      \put(35,43.5){\small $\gamma_4$}
      \put(26,46.5){\small $\gamma_5$}
      \put(14.5,43.5){\small $z$}
      \put(14.5,31){\small $\bar{z}$}
      \put(67,33.5){\small $\xi$}
      \put(102,37){\small $\re z$}
      \put(18,79){\small $\im z$}
\end{overpic}
     \begin{figuretext}\label{Jcontour2.pdf}
       The integration contour $\gamma$ is the composition of the five subcontours $\gamma_j$, $j = 1, \dots, 5$, defined in (\ref{gammajdef}).
       \end{figuretext}
     \end{center}
\end{figure}
Then
$$J(z, \xi) = \sum_{j=1}^5 J_j(z),$$
where
$$J_j(z) = \int_{\gamma_j} (u-z)^{\alpha}(u- \bar{z})^{\alpha - 2} u^{-\frac{\alpha}{2}}(u - \xi)^{-\frac{\alpha}{2}} du, \qquad j = 1, \dots, 5.$$

We claim that
\begin{align}\label{J1J5tozero}
  |J_j(z)| \leq Cr^{\frac{3\alpha}{2} - 1}, \qquad |z| < c, \  z \in \mathbb{H}, \  j = 1, 5.
\end{align}
Indeed, let us consider the case of $J_1(z)$. 
Since $|u - \xi| \geq c$ for $u \in \gamma_1$, we have
$$|J_1(z)| \leq \int_{\gamma_1} |u-z|^{\alpha} |u- \bar{z}|^{\alpha - 2} r^{-\frac{\alpha}{2}} c^{-\frac{\alpha}{2}} |du|.$$
Moreover, the inequalities
$$|u - \bar{z}| \leq |u - z| \leq 2r$$
are valid for $u \in \gamma_1$.
Hence, if $\alpha \geq 2$, then 
\begin{align*}
|J_1(z)|& \leq C \int_{\gamma_1} |u-z|^{2\alpha - 2} r^{-\frac{\alpha}{2}} |du|
\leq  C \int_{\gamma_1} r^{\frac{3\alpha}{2} - 2} |du| \leq C r^{\frac{3\alpha}{2} - 1}.
\end{align*}
On the other hand, if $1 < \alpha < 2$, then
\begin{align*}
|J_1(z)| & \leq C \int_{\gamma_1} (2r)^{\alpha} |u- \bar{z}|^{\alpha - 2} r^{-\frac{\alpha}{2}} |du|
 \leq C r^{\frac{\alpha}{2}} \int_0^{\theta} |r e^{-i\varphi} - r e^{-i\theta}|^{\alpha - 2} r d\varphi
 	\\
& \leq C r^{\frac{3\alpha}{2} - 1} \int_0^{\theta} \bigg|\frac{2}{\pi}(\theta - \varphi)\bigg|^{\alpha - 2} d\varphi
 \leq C r^{\frac{3\alpha}{2} - 1} \frac{\theta^{\alpha -1}}{\alpha - 1} \leq C r^{\frac{3\alpha}{2} - 1}.
 \end{align*}
This proves (\ref{J1J5tozero}) for $J_1(z)$; the proof for $J_5(z)$ is similar. 

We next show that 
\begin{align}\label{J2J4limit}
  \lim_{z \to 0} J_j(z) = J_j(0), \qquad j = 2,4.
\end{align}
To establish (\ref{J2J4limit}), we let $u = r + s$ and write
\begin{align*}
J_2(z) = \int_0^{c} f_{r,\theta}(s) ds,
\end{align*}
where 
$$f_{r,\theta}(s) = \chi_{[0, c-r]}(s) (s + r - re^{i\theta})^\alpha (s + r - re^{-i\theta})^{\alpha - 2} (s+r)^{-\frac{\alpha}{2}} |s + r - \xi|^{-\frac{\alpha}{2}} e^{\frac{\alpha i \pi}{2}}$$
and $\chi_{[0,c-r]}$ denotes the characteristic function of the interval $[0,c-r]$. 
We will show that there exists a function $F(s)$ in $L^1((0,c))$ such that
\begin{align}\label{frthetaF}
|f_{r,\theta}(s)| \leq F(s), \qquad s \in (0, c),
\end{align}
for all $r\in (0, c)$ and all $\theta \in [0, \pi]$. 
Dominated convergence then gives
$$\lim_{z \to 0} J_2(z) = \lim_{r \to 0} \int_0^{c} f_{r,\theta}(s) ds = \int_0^{c} f_{0,\theta}(s) ds
= J_2(0),$$
showing that $J_2(z)$ satisfies (\ref{J2J4limit}).

In order to prove (\ref{frthetaF}), we note that 
$$|s+ r - re^{i\theta}| = |s + r - re^{-i\theta}| \quad \text{and}  \quad |s + r - \xi| \geq c$$
for $s \in (0, c-r)$. This gives
$$|f_{r, \theta}(s)| \leq |s + r - re^{i\theta}|^{2\alpha - 2} (s+r)^{-\frac{\alpha}{2}} c^{-\frac{\alpha}{2}}.$$
Using the inequalities
$$|s + r - re^{i\theta}| \leq s + 2r, \qquad s + r \geq \frac{s + 2r}{2},$$
we find
$$|f_{r, \theta}(s)| \leq |s + 2r|^{\frac{3\alpha}{2} - 2} (2/c)^{\frac{\alpha}{2}}, \qquad s \in (0, c), \  r \in (0, c), \  \theta \in [0, \pi].$$
Since 
$$|s + 2r|^{\frac{3\alpha}{2} - 2}
\leq \begin{cases}
(3c)^{\frac{3\alpha}{2} - 2}, & \alpha \geq \frac{4}{3}, 
	\\
s^{\frac{3\alpha}{2} - 2}, & 1 < \alpha < \frac{4}{3},
\end{cases}$$
we deduce that (\ref{frthetaF}) holds with 
$$F(s) = (2/c)^{\frac{\alpha}{2}} \max((3c)^{\frac{3\alpha}{2} - 2}, s^{\frac{3\alpha}{2} - 2}), \qquad 0 < s < c.$$
This proves (\ref{J2J4limit}) for $j = 2$; the proof when $j = 4$ is similar.

Finally, since the integration contour is independent of $z$, it is easy to see that
\begin{align}\label{J3limit}
\lim_{z \to 0} J_3(z) = J_3(0).
\end{align}
Since $J(z, \xi) = \sum_{j=1}^5 J_j(z)$, the continuity of $J(z, \xi)$ at $z = 0$ follows from equations (\ref{J1J5tozero}), (\ref{J2J4limit}), and (\ref{J3limit}). This completes the proof of the lemma.
\end{proof}

\begin{lemma}\label{Jxycontinuouslemma}
For each $\xi > 0$, the partial derivatives $\partial_xJ(z, \xi)$ and $\partial_yJ(z, \xi)$ have continuous extensions to $\bar{\mathbb{H}} \setminus \{0, \xi\}$.
\end{lemma}
\begin{proof}
Fix $\xi > 0$. Defining $W(u)$ by
\begin{align}\label{Wdef}
  W(u) = u^{-\frac{\alpha}{2}}(u - \xi)^{-\frac{\alpha}{2}},
\end{align} 
we can write
\begin{align*}
J(z, \xi) = &\; \int_{\bar{z}}^z (u-z)^\alpha (u- \bar{z})^{\alpha - 2} W(u) du.
\end{align*}
An integration by parts gives
\begin{align}\nonumber
J(z, \xi) = & -\frac{\alpha}{\alpha -1} \int_{\bar{z}}^z (u-z)^{\alpha -1} (u- \bar{z})^{\alpha - 1} W(u) du
	\\
& - \frac{1}{\alpha -1}\int_{\bar{z}}^z (u-z)^{\alpha } (u- \bar{z})^{\alpha - 1} W'(u) du.
\end{align}
Differentiating with respect to $z$ and $\bar{z}$, we find
\begin{align*}
J_z(z, \xi) = &\; \alpha \int_{\bar{z}}^z (u-z)^{\alpha -2} (u- \bar{z})^{\alpha - 1} W(u) du
	\\
& + \frac{\alpha}{\alpha -1} \int_{\bar{z}}^z (u-z)^{\alpha-1} (u- \bar{z})^{\alpha - 1} W'(u) du
\end{align*}
and
\begin{align*}
J_{\bar{z}}(z, \xi) = &\; \alpha \int_{\bar{z}}^z (u-z)^{\alpha -1} (u- \bar{z})^{\alpha - 2} W(u) du
	\\
& + \int_{\bar{z}}^z (u-z)^{\alpha} (u- \bar{z})^{\alpha - 2} W'(u) du.
\end{align*}
Since $\alpha > 1$, these expressions for $J_z$ and  $J_{\bar{z}}$ are well-defined for each $z \in \R \setminus \{0, \xi\}$.
Repeating the above arguments that led to the continuity of $J(z, \xi)$ at each point $z \in \R \setminus \{0, \xi\}$, we infer that this provides continuous extensions of $J_z$ and $J_{\bar{z}}$ to $\bar{\mathbb{H}} \setminus \{0, \xi\}$.
\end{proof}

\begin{lemma}\label{reJzerolemma}
For each fixed $\xi > 0$, the function $J(z, \xi)$ satisfies
\begin{align}\label{reJxzero}
\re J(x, \xi) = 0, \qquad x \in \R, \  \xi > 0.
\end{align}
Moreover,
\begin{subequations}\label{JnearR}
\begin{align}\label{JnearRa}
\begin{cases} J(x+iy, \xi) = O(1), \\ 
\re J(x+iy, \xi) = O(y), 
\end{cases} \qquad y \downarrow 0,\  x \in \R \setminus \{0, \xi\},
\end{align}
and
\begin{align}\label{JnearRb}
J(x+iy, \xi) = O(y^{2\alpha -1}), \qquad y \downarrow 0,\  x > \xi,
\end{align}
\end{subequations}
where the error terms are uniform with respect to $x$ in compact subsets of $\R \setminus \{0, \xi\}$.
\end{lemma}
\begin{proof}
Equation (\ref{reJxzero}) follows by letting $y \to 0$ in (\ref{reJrepresentation}).
The asymptotic formulas (\ref{JnearRa}) are then a direct consequence of Lemma \ref{Jcontinuouslemma} and Lemma \ref{Jxycontinuouslemma}. 
Equation (\ref{JnearRb}) follows from (\ref{Jatplusinfinity}).
\end{proof}

In the above lemmas, we have established several properties of the function $J$. We next turn to the analysis of the functions $\mathcal{M}$ and $P$ defined in (\ref{Mdef}) and (\ref{Pdef}), respectively.

\begin{lemma}
The function $\mathcal{M}(z, \xi)$  defined in (\ref{Mdef}) satisfies the following estimates:
\begin{subequations}
\begin{align}\label{Meverywhere}
 |\mathcal{M}(z, \xi)| \leq &\;	C y^{\alpha -2}|z|^{1-\alpha}, \qquad z \in \mathbb{H}, \  \xi > 0,
	\\ \label{Matplusinfinity}
|\mathcal{M}(z, \xi)| \leq &\; C y^{3\alpha - 3}|z|^{1-\alpha}|z - \xi|^{1- \alpha} |x|^{-\frac{\alpha}{2}}|x - \xi|^{-\frac{\alpha}{2}}, \quad x > \xi, \  y > 0, \  \xi > 0,
	\\ \nonumber
 |\re \mathcal{M}(z, \xi)| \leq &\; C y^{\alpha -1}|z|^{-\alpha}|z - \xi|^{- \alpha}\Big[(|x||x - \xi| + y^2)|z|^{\alpha -2} 
	\\ \label{reMeverywhere}
& + (|x| + |\xi - x|)|z - \xi|^{\alpha -1}\Big], \qquad  |z| \geq 2\xi, \  z \in \mathbb{H}, \  \xi > 0,
\end{align}
\end{subequations}
where $z = x + iy$. In particular, for each fixed $\xi > 0$, 
\begin{subequations}\label{Minfinity}
\begin{align}\label{Minfinitya}
 \mathcal{M}(z, \xi) & = O(|x|^{1-\alpha}), \qquad  x \to \pm \infty,  \  y > 0,
	\\\label{Minfinityb}
 \re \mathcal{M}(z, \xi) & = O(|x|^{-\alpha}), \qquad 	x \to \pm \infty, \  y > 0,
\end{align}
\end{subequations}
where the error terms are uniform with respect to $y$ in compact subsets of $(0, \infty)$.
\end{lemma}
 \begin{proof}
The estimates (\ref{Meverywhere}) and (\ref{Matplusinfinity}) follow immediately from the definition of $\mathcal{M}$ together with the estimates (\ref{Jeverywhere}) and (\ref{Jatplusinfinity}) of Lemma \ref{Jestimateslemma}. 
The estimate (\ref{reMeverywhere}) follows by applying (\ref{Jeverywhere}) and (\ref{reJeverywhere}) to the identity
\begin{align}\nonumber
\re \mathcal{M}(z, \xi) = &\; 
y^{\alpha - 2} |z|^{-\alpha} |z- \xi|^{-\alpha} \re\big[\bar{z}(\bar{z} - \xi) J(z, \xi)\big]
	\\ \label{reMidentity}
= &\; y^{\alpha - 2} |z|^{-\alpha} |z- \xi|^{-\alpha} \big[(x^2 - x\xi - y^2)\re J(z, \xi) - y(\xi - 2x)\im J(z, \xi)\big].
\end{align}
 Finally, the asymptotic equations in (\ref{Minfinity}) are an immediate consequence of (\ref{Meverywhere}) and (\ref{reMeverywhere}).
\end{proof}

\begin{lemma}\label{M1yM2xclaim}
The function $\mathcal{M} = \mathcal{M}_1 + i \mathcal{M}_2$ satisfies
$$\partial_y\mathcal{M}_{1}(z, \xi) = - \partial_x \mathcal{M}_{2}(z, \xi), \qquad z \in \mathbb{H}, \  \xi > 0.$$
\end{lemma}
\begin{proof}
Since the statement only involves derivatives with respect to $x$ and $y$, we can assume that $\xi > 0$ is fixed. We need to prove that $\im \bar{\partial} \mathcal{M} = 0$. 
In terms of the function $W(u)$ defined in (\ref{Wdef}) we can write
\begin{align*}
\mathcal{M}(z, \xi) = &\; y^{\alpha - 2}z^{-\frac{\alpha}{2}}(z - \xi)^{-\frac{\alpha}{2}}\bar{z}^{1-\frac{\alpha}{2}}(\bar{z} - \xi)^{1-\frac{\alpha}{2}}
 \int_{\bar{z}}^z (u-z)^\alpha (u- \bar{z})^{\alpha - 2} W(u) du.
\end{align*}
An integration by parts gives
\begin{align}\nonumber
& \mathcal{M}(z, \xi) = -y^{\alpha - 2}z^{-\frac{\alpha}{2}}(z - \xi)^{-\frac{\alpha}{2}}\bar{z}^{1-\frac{\alpha}{2}}(\bar{z} - \xi)^{1-\frac{\alpha}{2}} \frac{1}{\alpha - 1} 
	\\ \label{mathcalMX}
& \times \bigg(\alpha \int_{\bar{z}}^z (u-z)^{\alpha -1} (u- \bar{z})^{\alpha - 1} W(u) du
+ \int_{\bar{z}}^z (u-z)^{\alpha } (u- \bar{z})^{\alpha - 1} W'(u) du\bigg).
\end{align}
Differentiating with respect to $\bar{z}$, we find
\begin{align}\nonumber
& \bar{\partial} \mathcal{M}(z, \xi) = \bigg(\frac{i}{2}\frac{\alpha -2}{y} + \frac{1 - \frac{\alpha}{2}}{\bar{z}} + \frac{1 - \frac{\alpha}{2}}{\bar{z} - \xi}\bigg)\mathcal{M}(z, \xi) 
	\\\nonumber
& + y^{\alpha - 2}z^{-\frac{\alpha}{2}}(z - \xi)^{-\frac{\alpha}{2}}\bar{z}^{1-\frac{\alpha}{2}}(\bar{z} - \xi)^{1-\frac{\alpha}{2}}
	\\  \label{barpartialcalM}
&\times \bigg(\alpha \int_{\bar{z}}^z (u-z)^{\alpha -1} (u- \bar{z})^{\alpha - 2} W(u) du
+ \int_{\bar{z}}^z (u-z)^{\alpha } (u- \bar{z})^{\alpha - 2} W'(u) du\bigg).
\end{align}
The integrals in (\ref{barpartialcalM}) are convergent at the endpoints $z$ and $\bar{z}$ because $\alpha > 1$. 
Substituting the expression (\ref{mathcalMX}) for $\mathcal{M}$ into (\ref{barpartialcalM}) and simplifying, we obtain 
\begin{align}\nonumber
& \bar{\partial} \mathcal{M}(z, \xi) = R(z, \xi)\bigg\{ -i(\alpha -2) \big[x(x - \xi) + y^2\big]
	\\\nonumber
&\quad \times \bigg(\alpha \int_{\bar{z}}^z (u-z)^{\alpha -1} (u- \bar{z})^{\alpha - 1} W(u) du + \int_{\bar{z}}^z (u-z)^{\alpha } (u- \bar{z})^{\alpha - 1} W'(u) du\bigg)
	\\\nonumber
&+ 2(\alpha -1)y\bar{z}(\bar{z} - \xi)
	\\ \label{barpartialcalM2}
& \quad \times \bigg(\alpha \int_{\bar{z}}^z (u-z)^{\alpha -1} (u- \bar{z})^{\alpha - 2} W(u) du + \int_{\bar{z}}^z (u-z)^{\alpha } (u- \bar{z})^{\alpha - 2} W'(u) du\bigg)\bigg\},
\end{align}
where  the real-valued function $R(z, \xi)$ is given by
$$R(z, \xi) = \frac{1}{2(\alpha -1)} y^{\alpha -3} |z|^{-\alpha}|z - \xi|^{-\alpha}.$$

To establish the identity $\im \bar{\partial} \mathcal{M} = 0$ it is enough to show that
\begin{align}\label{imbarpartialM}
\frac{\bar{\partial} \mathcal{M} - \overline{\bar{\partial} \mathcal{M}}}{R} = 0.
\end{align}
Using (\ref{intschwartz}) and (\ref{barpartialcalM2}), we can write the left-hand side of (\ref{imbarpartialM}) as
\begin{align}\nonumber
& -i(\alpha -2) \big[x(x - \xi) + y^2\big]
	\\\nonumber
&\quad \times \bigg(\alpha \int_{\bar{z}}^z (u-z)^{\alpha -1} (u- \bar{z})^{\alpha - 1} W(u) du
+ \int_{\bar{z}}^z (u-z)^{\alpha } (u- \bar{z})^{\alpha - 1} W'(u) du\bigg)
	\\\nonumber
& + 2(\alpha -1)y\bar{z}(\bar{z} - \xi)
	\\\nonumber
& \quad\times \bigg(\alpha \int_{\bar{z}}^z (u-z)^{\alpha -1} (u- \bar{z})^{\alpha - 2} W(u) du
+ \int_{\bar{z}}^z (u-z)^{\alpha } (u- \bar{z})^{\alpha - 2} W'(u) du\bigg)
	\\\nonumber
& -i(\alpha -2) \big[x(x - \xi) + y^2\big]
	\\\nonumber
& \quad\times \bigg(-\alpha \int_{\bar{z}}^z (u-z)^{\alpha -1} (u- \bar{z})^{\alpha - 1} W(u) du
- \int_{\bar{z}}^z (u-z)^{\alpha -1} (u- \bar{z})^{\alpha} W'(u) du\bigg)
	\\\nonumber
& - 2(\alpha -1)yz(z - \xi)
	\\\label{longbarpartialM}
& \quad\times \bigg(-\alpha \int_{\bar{z}}^z (u-z)^{\alpha -2} (u- \bar{z})^{\alpha - 1} W(u) du
- \int_{\bar{z}}^z (u-z)^{\alpha -2} (u- \bar{z})^{\alpha } W'(u) du\bigg).	
\end{align}
The right-hand side of (\ref{longbarpartialM}) involves eight integrals. Integrating by parts in the third, fourth, and eighth of these integrals and using that the first and fifth integrals cancel, we see that the expression in (\ref{longbarpartialM}) can be written as
\begin{align}\nonumber
& i(\alpha -2) \big[x(x - \xi) + y^2\big]
	\\\nonumber
& \quad \times \bigg(-\int_{\bar{z}}^z (u-z)^{\alpha} (u- \bar{z})^{\alpha - 1} W'(u) du + \int_{\bar{z}}^z (u-z)^{\alpha-1} (u- \bar{z})^{\alpha} W'(u) du \bigg) 
	\\\nonumber
&+ 2(\alpha-1)y\bar{z}(\bar{z} - \xi)
	\\\nonumber
& \quad \times \bigg(- \alpha \int_{\bar{z}}^z (u-z)^{\alpha-2} (u- \bar{z})^{\alpha - 1} W(u) du - \alpha \int_{\bar{z}}^z (u-z)^{\alpha-1} \frac{(u- \bar{z})^{\alpha - 1}}{\alpha -1} W'(u) du 
	\\\nonumber
&\quad\qquad - \alpha \int_{\bar{z}}^z (u-z)^{\alpha-1} \frac{(u- \bar{z})^{\alpha - 1}}{\alpha -1} W'(u) du 
- \int_{\bar{z}}^z (u-z)^{\alpha} \frac{(u- \bar{z})^{\alpha - 1}}{\alpha -1} W''(u) du\bigg) 
	\\\nonumber
& - 2(\alpha-1)yz(z - \xi)
	\\\nonumber
& \quad \times \bigg( - \alpha \int_{\bar{z}}^z (u-z)^{\alpha-2} (u- \bar{z})^{\alpha - 1} W(u) du 
	\\ \label{longbarpartialM2}
&\quad\qquad + \alpha \int_{\bar{z}}^z \frac{(u-z)^{\alpha-1}}{\alpha -1} (u- \bar{z})^{\alpha - 1} W'(u) du 
+ \int_{\bar{z}}^z \frac{(u-z)^{\alpha-1}}{\alpha -1} (u- \bar{z})^{\alpha} W''(u) du\bigg). 
\end{align}
A long but straightforward computation shows that the expression in (\ref{longbarpartialM2}) equals
\begin{align}\nonumber
2\alpha y \int_{\bar{z}}^z \frac{d}{du}\bigg[&(u-z)^{\alpha} (u- \bar{z})^{\alpha}  u^{-\frac{\alpha}{2}}(u - \xi)^{-\frac{\alpha}{2}} 
 \bigg(\frac{2x - \xi}{u - z}-\frac{x - \xi}{u} - \frac{x}{u - \xi}\bigg)\bigg] du.
\end{align}
Since $\alpha > 1$, the fundamental theorem of calculus implies that the integral vanishes. 
This proves (\ref{imbarpartialM}) and completes the proof of the lemma. 
\end{proof}

\begin{lemma} \label{App-lemma-P}
The function $P(z, \xi)$ defined in (\ref{Pdef}) is a well-defined smooth function of $(z, \xi) \in \mathbb{H} \times (0, \infty)$.
\end{lemma}

\begin{proof}
By Lemma \ref{Jsmoothlemma}, $\mathcal{M} = \mathcal{M}_1 + i\mathcal{M}_2$ is a smooth function of $(z, \xi) \in \mathbb{H} \times (0, \infty)$. Moreover, by equation (\ref{Minfinity}), there exists an $\epsilon > 0$ such that $\mathcal{M}_1(x + iy) = O(|x|^{-1-\epsilon})$ and $\mathcal{M}_2(x + iy) = O(|x|^{-\epsilon})$
as $|x| \to \infty$ uniformly with respect to $y$ in compact subsets of   $(0, \infty)$. It follows that the integral in the definition (\ref{Pdef}) of $P$ converges. 
Furthermore, Lemma~\ref{M1yM2xclaim} shows that the integral $\int (\mathcal{M}_1 dx - \mathcal{M}_2 dy)$ is independent of the path. We infer that $P(z, \xi)$ can be written as
\begin{align} \nonumber
P(z, \xi) & = -\frac{1}{c_\alpha} \int_\infty^z \Big(\mathcal{M}_1(z', \xi) dx' - \mathcal{M}_2(z', \xi) dy'\Big)
	\\ \label{Pcontourdef}
& = -\frac{1}{c_\alpha}  \int_\infty^z \re\big[\mathcal{M}(z', \xi) dz'\big], \qquad z \in \mathbb{H},
\end{align}
where $z' = x' + iy'$ and the contour of integration runs from $\infty + ic$, $c > 0$, to $z$. 
Since $\mathcal{M}(z, \xi)$ is a smooth function of $(z, \xi) \in \mathbb{H} \times (0, \infty)$, so is $P(z, \xi)$.
\end{proof}

\begin{lemma}\label{MnearRclaim}
For each $\xi > 0$, the function $\mathcal{M}(z, \xi)$   satisfies
\begin{align}\label{MnearR}
\begin{cases} 
\mathcal{M}(x+iy, \xi) = O(y^{\alpha -2}), \\ 
\re \mathcal{M}(x+iy, \xi) = O(y^{\alpha -1}), 
\end{cases} \qquad y \downarrow 0,\  x \in \R \setminus \{0, \xi\},
\end{align}
and
$$\mathcal{M}(x+iy, \xi) = O(y^{3\alpha -3}), \qquad y \downarrow 0,\  x > \xi,$$
where the error terms are uniform with respect to $x$ in compact subsets of $\R \setminus \{0, \xi\}$.
\end{lemma}
\begin{proof}
This follows from Lemma \ref{reJzerolemma} and the definition (\ref{Mdef}) of $\mathcal{M}$. 
\end{proof}

\begin{lemma}\label{Pcontinuousclaim}
For each $\xi > 0$, the function $z \mapsto P(z, \xi)$ defined in (\ref{Pdef}) has a continuous extension to $\bar{\mathbb{H}} \setminus \{0\}$ which satisfies
\begin{align}\label{PonR}
P(x, \xi) = \begin{cases} 1, &\quad  -\infty < x < 0, \\
0, &\qquad 0 < x < \infty,
\end{cases} \quad  \xi > 0.
\end{align}
\end{lemma}
\begin{proof}
Fix $\xi > 0$. 
The expression (\ref{Pcontourdef}) for $P$ together with Lemma \ref{MnearRclaim} imply that there exist real constants $\{P_j\}_1^3$ such that the function
$$z \mapsto \begin{cases} P(z, \xi), & z \in \mathbb{H}, \\
P_1, & z \in (-\infty, 0), \\
P_2, & z \in (0, \xi), \\
P_3, & z \in (\xi, \infty),
\end{cases}$$
is continuous $\bar{\mathbb{H}} \setminus \{0, \xi\} \to \R$.
Letting $z$ approach $\infty + i0$ in (\ref{Pcontourdef}), we deduce that $P_3 = 0$.  
It follows that
$$P(\xi + re^{i\theta}, \xi) = -\frac{1}{c_\alpha} \re\bigg(ir \int_0^\theta  \mathcal{M}(\xi + re^{i\varphi}, \xi) e^{i\varphi}d\varphi\bigg)$$
for $r > 0$ and $\theta \in (0, \pi)$. In view of the estimate (\ref{Meverywhere}), this yields
\begin{align}\nonumber
|P(\xi + re^{i\theta}, \xi)| \leq &\; C r \int_0^\theta |\mathcal{M}(\xi + re^{i\varphi}, \xi) | d\varphi
	\\ \label{Patxi}
\leq &\; C r^{\alpha -1} \int_0^\theta (\sin^{\alpha -2}\varphi)|\xi + re^{i\varphi}|^{1-\alpha} d\varphi, \qquad r> 0, \  \xi > 0.
\end{align}
Letting $r \downarrow 0$, we infer that if we set $P(\xi, \xi) = 0$, then $P(z, \xi)$ is continuous at $z = \xi$. Moreover, taking the limit $r \downarrow 0$ in (\ref{Patxi}) with $\theta = \pi$, it follows that $P_2 = 0$.

It only remains to prove that $P_3 = 1$. This will follow if we can show that the normalization constant $c_\alpha$ defined in (\ref{calphadef}) satisfies
\begin{align} \label{Calphaexpression}
  c_\alpha = -\re \int_{S_r} \mathcal{M}(z, \xi) dz, \qquad r > 0,
\end{align}
where $S_r$ is a counterclockwise semicircle of radius $r$ centered at $0$:
$$S_r : \; re^{i\varphi}, \qquad 0 \leq \varphi \leq \pi.$$
Lemma \ref{M1yM2xclaim} and Lemma \ref{MnearRclaim} show that the right-hand side of (\ref{Calphaexpression}) is independent of $r > 0$. We can therefore evaluate it in the limit as $r \downarrow 0$. 
Recalling the definition of $\mathcal{M}$, we write
\begin{align*}
   \int_{S_r} \mathcal{M}(z, \xi) dz = \int_{0}^{\pi} f_{r}(\varphi) d\varphi,
\end{align*}
where
$$f_{r}(\varphi) = i(\sin^{\alpha -2}\varphi)
(re^{i\varphi} - \xi)^{-\frac{\alpha}{2}} (re^{-i\varphi} - \xi)^{1 - \frac{\alpha}{2}} J(re^{i\varphi}, \xi).$$
The function $J$ is bounded on each compact subset of $\bar{\mathbb{H}}$ by Lemma \ref{Jcontinuouslemma}. Hence $f_r(\varphi)$ obeys the estimate
$$|f_r(\varphi)|
\leq C \sin^{\alpha -2}\varphi, \qquad r < \xi/2, \  \varphi  \in [0, \pi].$$ 
Since the function $\sin^{\alpha -2}\varphi$ belongs to $L^1((0, \pi))$, dominated convergence yields
\begin{align*}
-\re \int_{S_r} \mathcal{M}(z, \xi) dz = & -\re \lim_{r \to 0}  \int_{0}^{\pi} f_{r}(\varphi) d\varphi
= -\re \int_{0}^{\pi}  \lim_{r \to 0}f_r(\varphi) d\varphi
	\\
=&  -\re \int_{0}^{\pi} i(\sin^{\alpha -2}\varphi)
(-\xi + i0)^{-\frac{\alpha}{2}} (-\xi-i0)^{1 - \frac{\alpha}{2}} J(0, \xi) d\varphi
	\\
 = & \; -(\im J(0, \xi)) \xi^{1 - \alpha} \int_0^{\pi} \sin^{\alpha -2}\varphi d\varphi.
 \end{align*}
But setting $r = \xi$ and $\theta = \pi$ in (\ref{Jxiplusreitheta}), we find that the value of $J(z,\xi)$ at $z = 0$ is given by
\begin{align*}
J(0, \xi) = i\xi^{\alpha -1} \int_{-\pi}^\pi (1 + e^{i\varphi})^{\frac{3\alpha}{2} - 2} e^{i\varphi(1-\frac{\alpha}{2})} d\varphi, \qquad \xi > 0.
\end{align*}
It follows that
\begin{align*}
-\re \int_{S_r} \mathcal{M}(z, \xi) dz = - \bigg(\int_{-\pi}^\pi (1 + e^{i\varphi})^{\frac{3\alpha}{2} - 2} e^{i\varphi(1-\frac{\alpha}{2})} d\varphi\bigg) \bigg(\int_0^{\pi} \sin^{\alpha -2}\varphi d\varphi\bigg).
\end{align*}
Since
$$\int_{-\pi}^\pi (1 + e^{i\varphi})^{\frac{3\alpha}{2} - 2} e^{i\varphi(1-\frac{\alpha}{2})} d\varphi
= \frac{2 \pi  \Gamma \left(\frac{3 \alpha }{2}-1\right)}{\Gamma \left(\frac{\alpha}{2}\right) \Gamma (\alpha )}, \qquad \int_0^{\pi} \sin^{\alpha -2}\varphi d\varphi = \frac{\sqrt{\pi}\Gamma(\frac{\alpha-1}{2})}{\Gamma(\frac{\alpha}{2})},$$
equation (\ref{Calphaexpression}) follows.
\end{proof}

\begin{remark}\upshape
The proof of Lemma~\ref{Pcontinuousclaim} shows that the constant $c_\alpha$ can be alternatively expressed as
$$c_\alpha =  \int_{-\infty}^\infty  \re \mathcal{M}(x, y, \xi) dx,$$
where the right-hand side is independent of the choice of $y > 0$ and $\xi > 0$.
\end{remark}

In view of Lemma \ref{App-lemma-P}, the next result completes the proof of Proposition \ref{schrammprop}.

\begin{lemma}
We have
\begin{subequations}
\begin{align}\label{Peverywherea}
&  |P(z, \xi)| \leq C  (\arg z)^{\alpha -1}, \qquad z \in \mathbb{H}, \  \xi > 0,
  	\\ \label{Peverywhereb}
 & |P(z, \xi) - 1| \leq C  (\pi - \arg z)^{\alpha -1}, \qquad z \in \mathbb{H}, \  \xi > 0.
\end{align}
\end{subequations}
\end{lemma}
\begin{proof}
Using that $P(x, \xi) = 0$ for $x > 0$, we can write
$$P(re^{i\theta}, \xi) = -\frac{r}{c_\alpha} \re \int_{0}^{\theta} \mathcal{M}(re^{i\varphi}, \xi) ie^{i\varphi} d\varphi, \qquad r > 0, \  \theta \in (0, \pi),\  \xi > 0.$$
The estimate (\ref{Meverywhere}) now yields
\begin{align*}
|P(re^{i\theta}, \xi)| \leq &\; C r \int_{0}^{\theta} |\mathcal{M}(re^{i\varphi}, \xi)|  d\varphi
 \leq C \int_0^\theta (\sin^{\alpha -2} \varphi) d\varphi 
	\\
\leq &\; C   \theta^{\alpha -1}, \hspace{4cm} r > 0, \  \theta \in (0, \pi),\  \xi > 0,
\end{align*}
which is (\ref{Peverywherea}).
Similarly, since $P(x, \xi) = 1$ for $x < 0$, we can write
$$P(re^{i\theta}, \xi) = 1 + \frac{r}{c_\alpha} \re \int_{\theta}^{\pi} \mathcal{M}(re^{i\varphi}, \xi) ie^{i\varphi} d\varphi, \qquad r > 0, \  \theta \in (0, \pi),\  \xi > 0.$$
The estimate (\ref{Meverywhere}) now yields
\begin{align*}
|P(re^{i\theta}, \xi) - 1| \leq &\; C r \int_{\theta}^\pi |\mathcal{M}(re^{i\varphi}, \xi)|  d\varphi
\leq C \int_\theta^\pi (\sin^{\alpha -2} \varphi) d\varphi 
	\\
\leq &\; C (\pi - \theta)^{\alpha -1}, \hspace{3cm} r > 0, \  \theta \in (0, \pi),\  \xi > 0,
\end{align*}
which is (\ref{Peverywhereb}).
\end{proof}

\section{Acknowledgements} 
Lenells acknowledges support from the European Research Council, Grant Agreement No. 682537, the Swedish Research Council, Grant No. 2015-05430, and the Gustafsson Foundation, Sweden. Viklund acknowledges support from the Knut and Alice Wallenberg Foundation, the
Swedish Research Council, the National Science Foundation, and the Gustafsson Foundation, Sweden. 


\bibliographystyle{plain}
\bibliography{is}

\end{document}